\def\showauthnotes{1}

\documentclass[letterpaper,11pt]{article}
\usepackage[margin=0.94in]{geometry}

\usepackage{cmap} % Load before fontenc 
\usepackage[utf8]{inputenc}
\usepackage[english]{babel}
\usepackage[T1]{fontenc}
\usepackage{amsmath}
\usepackage{enumitem}
\usepackage{amsfonts}
\usepackage{bm}
\usepackage{xcolor}
\definecolor{sunset}{rgb}{0.85, 0.44, 0.3} % A warm sunset color
\definecolor{forestgreen}{rgb}{0.13, 0.55, 0.13} % Deep green, reminiscent of a forest
\definecolor{coral}{rgb}{1.0, 0.5, 0.31} % A vibrant coral color
\definecolor{midnightblue}{rgb}{0.1, 0.3, 0.95} % A deep, dark blue
\definecolor{goldenrod}{rgb}{0.85, 0.65, 0.13} % A rich goldenrod color
\definecolor{lavender}{rgb}{0.9, 0.9, 0.98} % Soft lavender
\definecolor{skyblue}{rgb}{0.15, 0.25, 0.8} % Light sky blue
\definecolor{crimson}{rgb}{0.86, 0.08, 0.24} % Bold crimson red
\definecolor{teal}{rgb}{0.0, 0.5, 0.5} % Cool teal
\definecolor{peachpuff}{rgb}{1.0, 0.85, 0.73} % Soft peachy color
\definecolor{blueviolet}{rgb}{0.15, 0.15, 0.6}
\definecolor{webgreen}{rgb}{0,.5,0}
\definecolor{webbrown}{rgb}{.6,0,0}
\usepackage{setspace}
\usepackage[pdftex,
	bookmarks=false,
	colorlinks=true, %allcolors=blueviolet,
	urlcolor=webbrown, 
	linkcolor=blueviolet, 
	citecolor=webgreen,
	pdfstartpage=1,
	pdfstartview={FitH},  % FitBH
	bookmarksopen=false
	]{hyperref}
\allowdisplaybreaks
\usepackage{tikz}
\usepackage{braket}
\usepackage[numbers,sort&compress]{natbib}
\usepackage{amsthm}
\usepackage{dsfont}
\usepackage{listings}
\usepackage{natbib}
\usepackage{amssymb}
\usepackage[capitalize]{cleveref}
\usepackage{appendix}

\crefname{claim}{Claim}{Claims}

\usepackage{physics}
\usepackage{mathtools}
\usepackage{bbm}
\usepackage{xcolor}
\numberwithin{equation}{section}

\newcounter{lemc}

\newcounter{coroc}
\numberwithin{lemc}{section}
\numberwithin{propc}{section}
\numberwithin{coroc}{section}
\newtheorem{theorem}{Theorem}
\newtheorem{claim}{Claim}
\newtheorem{corollary}[coroc]{Corollary}

\newtheorem{definition}{Definition}
\newtheorem{notation}{Notation}

\newtheorem{lemma}[lemc]{Lemma}

\newtheorem{fact}{Fact}

\theoremstyle{definition}

\newcommand{\vertiii}[1]{{\left\vert\kern-0.25ex\left\vert\kern-0.25ex\left\vert #1 \right\vert\kern-0.25ex\right\vert\kern-0.25ex\right\vert}}

\ifnum\showauthnotes=1
\newcommand{\robert}[1]{\textcolor{red}{\small {\textbf{(Robert:} #1\textbf{) }}}}
\newcommand{\fermi}[1]{\textcolor{blue}{\small {\textbf{(Fermi:} #1\textbf{) }}}}
\else 
\newcommand{\robert}[1]{}
\newcommand{\fermi}[1]{}
\fi

\providecommand{\darkgray}[1]{\textcolor{darkgray}{#1}}

\providecommand{\Id}{\mathsf{Id}}

\providecommand{\eq}{\mathsf{eq}}
\providecommand{\opnorm}{\mathrm{op}}

\providecommand{\xydist}{\mathrm{dist}_{X,Y}}

\providecommand{\bij}{\mathsf{bij}}
\providecommand{\inj}{\mathsf{inj}}

\providecommand{\Compress}{\mathsf{Comp}}
\providecommand{\Uncompress}{\mathsf{Uncompress}}

\providecommand{\num}{\mathsf{num}}

\providecommand{\sSym}{\mathsf{Sym}}
\providecommand{\ssym}{\mathsf{sym}}

\providecommand{\EPR}{\mathsf{EPR}}

\providecommand{\pf}{\mathsf{pf}}

\providecommand{\Adv}{\mathcal{A}}

\providecommand{\TD}{\mathsf{TD}}

\providecommand{\sA}{\mathsf{A}}
\providecommand{\sB}{\mathsf{B}}
\providecommand{\sC}{\mathsf{C}}
\providecommand{\sD}{\mathsf{D}}
\providecommand{\sE}{\mathsf{E}}
\providecommand{\sF}{\mathsf{F}}

\providecommand{\sL}{\mathsf{L}}

\providecommand{\sP}{\mathsf{P}}

\providecommand{\sR}{\mathsf{R}}

\providecommand{\sX}{\mathsf{X}}
\providecommand{\sY}{\mathsf{Y}}

\providecommand{\scC}{\mathsf{cC}}
\providecommand{\scD}{\mathsf{cD}}
\providecommand{\scQ}{\mathsf{cQ}}
\providecommand{\frakD}{\mathfrak{D}}

\providecommand{\pr}{V}

\providecommand{\pr}{\mathsf{CU}}
\providecommand{\pr}{\mathsf{CU}}
\providecommand{\pfo}{\mathsf{pfO}}
\providecommand{\pf}{\mathsf{pfO}}

\providecommand{\spfo}{\mathsf{pfO}}

\providecommand{\pru}{\mathsf{PRU}}
\providecommand{\spru}{\mathsf{sPRU}}
\providecommand{\init}{\mathsf{init}}

\providecommand{\gsA}{{\textcolor{darkgray}{\mathsf{A}}}}
\providecommand{\gsB}{{\textcolor{darkgray}{\mathsf{B}}}}
\providecommand{\gsC}{{\textcolor{darkgray}{\mathsf{C}}}}
\providecommand{\gsD}{{\textcolor{darkgray}{\mathsf{D}}}}
\providecommand{\gsE}{{\textcolor{darkgray}{\mathsf{E}}}}
\providecommand{\gsF}{{\textcolor{darkgray}{\mathsf{F}}}}

\providecommand{\gsL}{{\textcolor{darkgray}{\mathsf{L}}}}

\providecommand{\gsP}{{\textcolor{darkgray}{\mathsf{P}}}}

\providecommand{\gsR}{{\textcolor{darkgray}{\mathsf{R}}}}

\providecommand{\calA}{\mathcal{A}}

\providecommand{\calD}{\mathcal{D}}

\providecommand{\calH}{\mathcal{H}}
\providecommand{\calI}{\mathcal{I}}

\providecommand{\calO}{\mathcal{O}}

\providecommand{\calR}{\mathcal{R}}

\providecommand{\calU}{\mathcal{U}}

\providecommand{\calZ}{\mathcal{Z}}

\newcommand{\rom}[1]{\mathtt{\uppercase\expandafter{\romannumeral #1\relax}}}

\DeclareMathOperator{\Dom}{Dom}
\DeclareMathOperator{\dist}{{dist}}

\DeclareMathOperator*{\E}{{\mathbb{E}}}

\DeclareMathOperator{\poly}{poly}

\usepackage{xargs}
\usepackage[colorinlistoftodos,prependcaption,textsize=tiny]{todonotes}
\newcommandx{\lz}[2][1=]{\todo[linecolor=red,backgroundcolor=red!10,bordercolor=red,#1]{LZ: #2}}

\usepackage{tocloft}

\begin{document}

\title{How to Construct Random Unitaries}

\author{Fermi Ma\thanks{Simons Institute $\&$ UC Berkeley. Email: \texttt{fermima1@gmail.com}} \and Hsin-Yuan Huang\thanks{Google Quantum AI, Caltech, $\&$ MIT. This work was conducted while Hsin-Yuan Huang was visiting the Simons Institute for the Theory of Computing. Email: \texttt{hsinyuan@google.com}, \texttt{hsinyuan@caltech.edu}}}
% \author[1]{Fermi Ma\thanks{Simons Institute $\&$ UC Berkeley. Email: \texttt{fermima1@gmail.com}.}}
% \author[2]{Hsin-Yuan Huang\thanks{Google Quantum AI, Caltech, $\&$ MIT. This work was conducted while Hsin-Yuan Huang was visiting the Simons Institute for the Theory of Computing. Email: \texttt{hsinyuan@google.com}, \texttt{hsinyuan@caltech.edu}}}
% \author[3,4,5,*]{Hsin-Yuan Huang}
% \affil[1]{Simons Institute for the Theory of Computing}
% \affil[2]{University of California, Berkeley}
% \affil[3]{Google Quantum AI}
% \affil[4]{California Institute of Technology}
% \affil[5]{Massachusetts Institute of Technology}

\date{}

\maketitle

\begin{abstract}
\normalsize
The existence of pseudorandom unitaries (PRUs)---efficient quantum circuits that are computationally indistinguishable from Haar-random unitaries---has been a central open question, with significant implications for cryptography, complexity theory, and fundamental physics.
In this work, we close this question by proving that PRUs exist, assuming that any quantum-secure one-way function exists.
We establish this result for both (1) the standard notion of PRUs, which are secure against any efficient adversary that makes queries to the unitary $U$, and (2) a stronger notion of PRUs, which are secure even against adversaries that can query both the unitary $U$ and its inverse $U^\dagger$. 
In the process, we prove that any algorithm that makes queries to a Haar-random unitary can be \emph{efficiently} simulated on a quantum computer, up to inverse-exponential trace distance.

\end{abstract}

\thispagestyle{empty} % This page will not have a page number
\setcounter{page}{0}

\renewcommand{\thefootnote}{\fnsymbol{footnote}}
% \footnotetext{*This work was conducted when Hsin-Yuan Huang was visiting the Simons Institute for the Theory of Computing.}

\newpage

\thispagestyle{empty} % This page will not have a page number
\setcounter{page}{0}

\tableofcontents
\vspace{1em}

\addtocontents{toc}{\protect\thispagestyle{empty}}

\thispagestyle{empty} % This page will not have a page number
\setcounter{page}{0}

\newpage

\renewcommand{\thefootnote}{\arabic{footnote}}

\section{Introduction}

This paper resolves the question: can efficient quantum circuits behave like truly random unitaries? Specifically, we prove that \emph{pseudorandom unitaries} (PRUs) exist assuming the existence of any quantum-secure one-way function. First proposed by Ji, Liu, and Song in 2017~\cite{ji2017pseudorandom}, a PRU is the unitary analogue of a pseudorandom function (PRF) \cite{goldreich1986construct}. A PRU consists of a family of efficiently computable quantum circuits with the guarantee that no polynomial-time quantum algorithm can distinguish between queries to a unitary sampled from the PRU family and a unitary sampled from the Haar measure.

Random unitaries play an essential role throughout quantum information science, arising in quantum algorithms, quantum supremacy experiments, quantum learning, cryptographic protocols, and much more~\cite{hayden2004randomizing, knill2008randomized, arute2019quantum, bouland2019complexity, huang2020predicting, ananth2022cryptography, huang2022provably, elben2022randomized, huang2022quantum, movassagh2023hardness, kretschmer2023quantum, lombardi2024one}. In physics, highly chaotic systems such as black holes are often modeled as Haar-random unitary transformations \cite{cotler2017black, nahum2018operator, cotler2017chaos, kim2020ghost, choi2023preparing}. However, this approach has a fundamental problem: Haar-random unitaries are inherently unphysical, requiring exponential complexity to even specify. The notion of a PRU offers a tantalizing solution: efficient circuits that are as good as Haar-random. In fact, the idea that PRUs are a more accurate model of black hole dynamics is behind recent advances in fundamental physics~\cite{kim2020ghost, yang2023complexity, engelhardt2024cryptographic}.

Despite considerable interest, the question of whether PRUs actually \emph{exist} has remained open. In the past couple of years, a series of works has established that weaker notions are possible~\cite{lu2023quantum, brakerski2024real,haug2023pseudorandom, metger2024simple,ananth2024pseudorandom}. For example,~\cite{metger2024simple, chen2024efficient} constructed \emph{non-adaptive} PRUs, which are secure against restricted adversaries that makes all of their queries \emph{at once} in parallel. While these works represent important progress, the broader goal remains elusive, and constructing a PRU remains one of the central challenges in quantum cryptography.

\subsection{Our results} In this work, we give the first proof that PRUs exist.
\begin{theorem}
\label{theorem:intro-standard-pru}
    PRUs exist assuming the existence of any quantum-secure one-way function.
\end{theorem}
\noindent In fact, we go one step further.~\cref{theorem:intro-standard-pru} is about PRUs that satisfy the original definition of~\cite{ji2018pseudorandom}, which are secure against adversaries that can query an oracle for $U$, but \emph{not} the inverse unitary $U^\dagger$. We therefore define \emph{strong PRUs}, which are indistinguishable from Haar-random even to adverasaries that can query both $U$ and $U^\dagger$. Our second main result builds strong PRUs from one-way functions.\footnote{The notion of strong PRUs is also discussed in~\cite{metger2024simple} as an open question.}
\begin{theorem}
\label{theorem:intro-strong-pru}
    Strong PRUs exist assuming the existence of any quantum-secure one-way function.
\end{theorem}
While~\cref{theorem:intro-strong-pru} technically subsumes~\cref{theorem:intro-standard-pru}, the proof of~\cref{theorem:intro-strong-pru} is significantly more involved. Since~\cref{theorem:intro-standard-pru} may suffice for many applications, we present them separately.
By establishing the existence of PRUs, our work provides the foundation for new avenues of research in quantum computation, cryptography, and fundamental physics.

\subsection{Our techniques}

We achieve our results on PRUs by proving that any quantum oracle algorithm $\Adv^U$ that queries an $n$-qubit Haar-random unitary $U$ can be \emph{efficiently simulated} with a remarkably simple procedure:
\begin{enumerate}
    \item Initialize an external register $\sE$ to the state $\ket*{\varnothing}$, where $\varnothing$ denotes the empty set. (\textbf{Aside:} When we write a set inside a ket, e.g., $\ket*{S}_{\sE}$, we are simply using the set $S$ as a label for a unit vector. The inner product $\braket*{R}{S}$ equals $1$ if $R = S$ and $0$ otherwise.)

    \item Run the oracle algorithm $\calA$, replacing each query to $U$ with the following linear map:
    \begin{align}
        V: \ket*{x} \ket*{S}_{\gsE} \mapsto 
        \frac{1}{\sqrt{2^n - \abs{S}}} \sum_{\substack{y \in \{0,1\}^n: \\ y\not\in S_Y}} \ket*{y} \ket*{S \cup \{(x,y)\}}_{\gsE},
    \end{align}
    where $S_Y$ denotes the set of all $y$ such that $(x,y) \in S$ for some $x$. In words, $V$ maps $x$ to a uniform superposition over $y \in \{0,1\}^n$, except those that already appear in $S$, and simultaneously ``records'' $(x,y)$ by inserting it into $S$. We refer to $V$ as the \emph{path-recording oracle}.
\end{enumerate}
We prove that the following mixed states have trace distance $O(t^2/2^n)$:
\begin{itemize}
    \item $\E_{U} \ketbra*{\Adv^{U}}$, the state of $\Adv$ after $t$ queries to a Haar-random unitary $U$, where $\ket*{\Adv^U} \coloneqq U \cdot A_t \cdots U \cdot A_1 \ket*{0}$
    denotes the state of the algorithm after $t$ queries to $U$, and $\ket*{0}$ denotes an arbitrary initial state.
    \item $\Tr_{\sE}(\ketbra*{\Adv^{V}})$, where $\ket*{\calA^V}_{\gsA \gsE} \coloneqq V \cdot A_t \cdot \cdots V \cdot A_1 \ket*{0} \ket*{\varnothing}_{\sE}$ denotes the global state of the algorithm and the external register $\sE$ after $t$ queries to $V$.
\end{itemize}
Despite the extensive literature on Haar-random unitaries, to the best of our knowledge, this ``path-recording'' characterization was not known before.\footnote{We note that~\cite{alagic2020efficient} proves that there \emph{exists} a space-efficient (but otherwise inefficient) way to exactly simulate Haar-random unitaries. Moreover, their proof is non-constructive, i.e., they do not give a simulator.}\footnote{This can also be viewed as an analog of Zhandry's compressed oracles for Haar-random unitaries~\cite{zhandry2019record}.}
Furthermore, it is easy to show that $V$ can be efficiently implemented on a quantum computer; see \cref{sec:efficient-implementation-pro}. This establishes the following fact:
\begin{center}
    \emph{Any algorithm that queries a Haar-random unitary can be efficiently simulated \\
    on a quantum computer up to inverse-exponential trace distance.}
\end{center}
As we now explain, this new path-recording perspective is the key to our PRU proof.

\paragraph{How to construct PRUs.} The main technical step in our PRU proof is to show that a $t$-query oracle algorithm $\Adv$ can only distinguish between
\begin{itemize}
    \item $P_{\pi} \cdot F_f \cdot C$, where $P_{\pi} = \sum_{x} \ketbra*{\pi(x)}{x}$ for a random permutation $\pi \gets S_{2^n}$, $F_f = \sum_{x} (-1)^{f(x)} \ketbra*{x}$ for a random function $f \gets \{0,1\}^{2^n}$, and $C$ is a random $n$-qubit Clifford.\footnote{This $PFC$ construction was introduced by~\cite{metger2024simple}, who proved security against \emph{non-adaptive} adversaries, i.e., adversaries that make all of their oracle queries at once, in parallel.}
    \item a Haar-random $n$-qubit unitary $U$,
\end{itemize}
with probability $1/2 + t^2/2^n$.

Our proof works by \emph{purifying} the randomness of the PRU. Ignoring $C$ for now, suppose we initialize an external register to the uniform superposition $\propto \sum_{\pi \in S_{2^n}} \ket*{\pi} \otimes \sum_{f \in \{0,1\}^{2^n}} \ket*{f}$ over all permutations $\pi$ and functions $f$. In this view, a query to a random $P_{\pi} \cdot F_f$ is equivalent to a query to a fixed unitary that applies $P_{\pi} \cdot F_f$ controlled on $\ket*{\pi} \ket*{f}$, i.e., the map
\begin{align}
    \ket*{x} \otimes \ket*{\pi,f} \mapsto (-1)^{f(x)} \cdot \ket*{\pi(x)} \otimes \ket*{\pi, f}.
\end{align}
Equivalently, we can view this map as sending $x$ to a superposition over all $y$, while simultaneously multiplying the purifying register by the coefficient $\delta_{\pi(x) = y} \cdot (-1)^{f(x)}$: 
\begin{align}
    \ket*{x} \otimes \ket*{\pi,f} \mapsto \sum_{y \in \{0,1\}^n} \ket*{y}  \otimes \Big( \textcolor{skyblue}{\delta_{\pi(x) = y} \cdot (-1)^{f(x)}} \cdot \ket*{\pi, f} \Big).
\end{align}
After $t$ queries to the purified $P_{\pi} \cdot F_f$, the global state including the purifying registers is (proportional to) a sum of terms
\begin{equation}
    \ketbra*{y_t}{x_t} \cdot A_t \cdots \ketbra*{y_1}{x_1} \cdot A_1 \ket*{0^{n}} \otimes \underbrace{\sum_{\pi \in S_{2^n}} \ket*{\pi,f} \cdot \textcolor{skyblue}{\delta_{\pi(x_1) = y_1} \cdots \delta_{\pi(x_t) = y_t} \cdot (-1)^{f(x_1) + \cdots + f(x_t)}}}_{\propto \ket*{\pf_{\{(x_1,y_1),\dots,(x_t,y_t)\}}}},
\end{equation}
over all possible $x_1,y_1,\dots,x_t,y_t \in \{0,1\}^n$, i.e., over all \emph{Feynman paths}. 

Crucially, when all the $x_1,\dots,x_t$ are distinct, these $\ket*{\pf_{\{(x_1,y_1),\dots,(x_t,y_t)\}}}$ states are orthogonal and is isometric to $\ket*{\{(x_1,y_1),\dots,(x_t,y_t)\}}$. Since the algorithm is not given the purifying registers, a query to a random $P_{\pi} \cdot F_f$ is \emph{identical} to a query to the path-recording oracle $V$ described earlier---except on paths where there is a collision among the inputs $x_1,\dots,x_t$.

This is where $C$ comes in. We prove that $V$ satisfies a key property: for any $n$-qubit unitary $C$, 
\begin{equation}
    (V \cdot C) \cdot A_t \cdots (V \cdot C) \cdot A_1 \ket*{0^{n}} \ket*{\varnothing}_{\sE} = ((C\otimes \Id)^{\otimes t})_{\sE} \cdot V \cdot A_t \cdots V \cdot A_1 \ket*{0^{n}} \ket*{\varnothing}_{\sE}.
\end{equation}
This says that applying $C$ to the \textbf{adversary's register} before each query to $V$ is equivalent to applying $C$ to each $x_i$ in the \textbf{purifying register} $\ket*{\{(x_1,y_1),\dots,(x_t,y_t)\}}$. When $C$ is sampled from any $2$-design, the randomness of $C$ ensures there are no collisions in the $x_1,\dots,x_t$ with overwhelming probability. Consequently, we show that queries to $V$ are indistinguishable from queries to $P_{\pi} \cdot F_f \cdot C$, as long as $C$ is sampled from \emph{any} $2$-design. By instantiating the $2$-design to be either (1) a random Clifford or (2) a Haar-random unitary, we show that both $P_{\pi} \cdot F_f \cdot C$ and Haar-random unitaries are indistinguishable from $V$, and thus, from each other.

\paragraph{Strong PRUs and a symmetrized path-recording oracle $\widetilde{V}$.} To obtain strong PRUs, we use the construction: $D \cdot P_{\pi} \cdot F_f \cdot C$, where $D,C$ are both random $n$-qubit Cliffords, $P_{\pi}$ is the same as before, and $F_f$ is a random $q$-ary phase (for any $q \geq 3$). By analyzing the purification of $P_{\pi} \cdot F_f$, we show that when $\Adv$ makes forward and inverse queries, the purifying registers, viewed in the right basis, ``record'' information from \emph{two Feynman paths}: one set $S^{\mathsf{for}}$ consists of $(x,y)$ tuples corresponding to the forward queries, and another set $S^{\mathsf{inv}}$ of tuples $(x,y)$ corresponds to the inverse queries. Whereas each query in the standard PRU proof always inserts a tuple $(x,y)$ into the set $S$, when both forward and inverse queries are allowed, the effect is more intricate:
\begin{itemize}
    \item A forward query will sometimes add a tuple to $S^{\mathsf{for}}$, but other times delete a tuple from $S^{\mathsf{inv}}$.
    \item An inverse query will sometimes add a tuple to $S^{\mathsf{inv}}$, but other times delete a tuple from $S^{\mathsf{for}}$.
\end{itemize}
We prove that this behavior corresponds to a more general ``symmetrized'' path recording oracle $\widetilde{V}$. Moreover, as long as $D,C$ are sampled from any $2$-design, the adversary cannot distinguish between queries to $D \cdot P_{\pi} \cdot F_f \cdot C$ and queries to $\widetilde{V}$, and using similar reasoning as the standard PRU proof, conclude both of the following (1) strong PRUs exist and (2) $\widetilde{V}$ is indistinguishable from Haar-random even under inverse queries. As we show in~\cref{sec:efficient-implementation-pro}, $\widetilde{V}$ can also be implemented efficiently, and consequently any algorithm that makes \emph{forward and inverse} queries to a Haar-random unitary can also be simulated to inverse exponential error.

Our proof leverages the following property of $2$-designs: if one samples $C$ from a $2$-design and applies $C \otimes \overline{C}$ to any state (where $\overline{C}$ denotes the complex conjugate), then with overwhelmingly high probability, the result is either (a) a pair of distinct elements, or (b) the maximally entangled state. At a very high level, the fact that there are two kinds of outcomes after twirling by $C \otimes \overline{C}$ is related to how the purification ``decides'' whether it should add or delete a tuple $(x, y)$. 

We remark that the strong PRU proof is significantly more involved than standard PRU proof, and the reader may find it beneficial to start with the standard PRU proof. 

\paragraph{A new approach to random unitaries.}
More broadly, the path-recording oracle unlocks a new way to proving theorems about random unitaries. Before this work, analyzing mixed states such as $\mathbb{E}_{U} \ketbra*{\text{Adv}^U}$ often necessitated the use of Weingarten calculus, involving intricate asymptotic bounds on Weingarten functions through sophisticated combinatorial and representation-theoretic calculations. Our approach circumvents this complexity entirely.\footnote{Alternatively, one can view our technique as deriving a simplified and approximate version of the Weingarten calculus from purely elementary arguments.}

We demonstrate the power of this approach by giving an elementary proof of the ``gluing lemma'' recently proven by~\cite{schuster2024random}. This lemma states that if two Haar-random unitaries $U_1$ and $U_2$ \emph{overlap}, with $U_1$ acting on systems $\sA, \sB$ and $U_2$ on $\sB, \sC$ (where $\sB$ has a super-logarithmic number of qubits), then queries to $U_2 \cdot U_1$ are indistinguishable from queries to a larger Haar-random unitary $U$ acting on $\sA,\sB,\sC$. Using this lemma (and our~\cref{theorem:intro-standard-pru}), \cite{schuster2024random} constructed low-depth PRUs secure against forward queries. However, their proof of the gluing lemma is highly technical, relying on careful representation-theoretic analysis and tight bounds on Weingarten functions.

The path-recording oracle yields an elementary proof of the gluing lemma (see \cref{part:app}). The key insight is to replace the Haar-random unitaries with path-recording oracles. This reduces to showing that the composition of two independent path-recording oracles $V_2 \cdot V_1$, where $V_1$ acts on $(\sA,\sB,\sE_1)$ and $V_2$ acts on $(\sB,\sC,\sE_2)$, approximates a single path-recording oracle $V$ acting on $(\sA,\sB,\sC,\sE)$.

Given the central role of random unitaries in physics and quantum computing, we expect the path-recording framework to have broad applications in the future.

\subsection{Acknowledgments}

Special thanks to John Wright for many helpful suggestions and extensive discussions at every stage of this project, and to Ewin Tang for providing significant feedback on the manuscript. We also thank Thiago Bergamaschi, John Bostanci, Adam Bouland, Chi-Fang (Anthony) Chen, Lijie Chen, Tudor Giurgica-Tiron, Jeongwan Haah, Jonas Haferkamp, William Kretschmer, Alex Lombardi, Tony Metger, Thomas Schuster, Joseph Slote, Xinyu (Norah) Tan, Umesh Vazirani, Henry Yuen, and Mark Zhandry for valuable discussions and feedback.

This work was done while Fermi Ma was a postdoctoral fellow at the Simons Institute for the Theory of Computing, supported by DOE QSA grant FP00010905, NSF QLCI Grant 2016245 and DOE grant DE-SC0024124. Hsin-Yuan Huang acknowledges the visiting position at Center for Theoretical Physics, MIT. This work was conducted while both authors were at the Simons Institute for the Theory of Computing, supported by DOE QSA grant FP00010905.

\section{Preliminaries}

This section establishes basic notation, definitions, and lemmas that we use throughout the paper.

\paragraph{Notation.} We write $N \coloneqq 2^n$, where $n$ typically denotes the number of qubits. We write $[N] \coloneqq \{1, \ldots, N\}$ to denote the set of integers from $1$ to $N$, and we will identify $[N]$ with $\{0,1\}^n$ by associating each integer $i \in [N]$ with the string $x \in \{0,1\}^n$ corresponding to the binary representation of $i-1$. For any integer $1 \leq t \leq N$, let $[N]^t_{\dist}$ denote the set of length-$t$ sequences of distinct integers from $1$ to $N$, i.e.,
    \begin{equation}
        [N]^t_{\dist} \coloneqq \{(x_1,\dots,x_t) \in [N]^t: x_i \neq x_j \ \text{for all} \ i \neq j \}.
    \end{equation}
    For $t = 0$, we adopt the convention that $[N]^t_{\dist} \coloneqq \{ () \}$ is a set with a single element $()$ denoting a length-$0$ sequence. For any permutation $\pi \in \sSym_t$, let $S_{\pi}$ be a unitary that acts on $(\mathbb{C}^N)^t$ as follows:
    \begin{equation}
        S_{\pi}: \ket*{x_1,\dots,x_t} \mapsto \ket*{x_{\pi^{-1}(1)},\dots,x_{\pi^{-1}(t)}}.
    \end{equation}
    
\paragraph{Quantum registers.} We use capital sans-serif letters to label quantum registers. For a register~$\sA$, the associated Hilbert space is denoted $\calH_{\sA}$. When a quantum state is supported on multiple registers, such as $(\sA, \sB)$, this means that $\ket*{\psi} \in \calH_{\sA} \otimes \calH_{\sB}$. To clarify which systems a state is defined on, we sometimes include the register labels as subscripts in dark gray sans-serif font, e.g., $\ket*{\psi}_{\gsA \gsB}$. If a linear operator $U$ acts only on subsystem $\sA$, we may write this as $U_{\gsA}$. Such an operator can be extended to a larger system by acting trivially on other registers; for example, $(U_{\gsA} \otimes \Id_{\gsB}) \cdot \ket*{\psi}_{\gsA \gsB}$. To reduce notational clutter, we often omit the ``$\otimes \Id_{\gsB}$'' and simply write $U_{\gsA} \cdot \ket*{\psi}_{\gsA \gsB}$. Similarly, when summing operators that act on different registers, such as $U_{\gsA}$ and $V_{\gsA \gsB}$, we write $U_{\gsA} + V_{\gsA \gsB}$ to mean $U_{\gsA} \otimes \Id_{\gsB} + V_{\gsA \gsB}$.

Given a projector $\Pi$ acting on register $\sA$, we say that a state $\ket*{\psi} \in \calH_{\sA}$ is in the image of $\Pi$ if $\Pi \ket*{\psi} = \ket*{\psi}$. For a state $\ket*{\psi} \in \calH_{\sA} \otimes \calH_{\sB}$, we similarly say that $\ket*{\psi}$ is in the image of $\Pi_{\gsA}$ if $\Pi_{\gsA} \ket*{\psi}_{\gsA \gsB} = (\Pi_{\gsA} \otimes \Id_{\gsB}) \cdot \ket*{\psi}_{\gsA \gsB} = \ket*{\psi}_{\gsA \gsB}$.

Given a state $\ket*{\psi}$ on systems $(\sA, \sB)$, we denote the partial trace over system $\sB$ as $\Tr_{\sB}(\ketbra{\psi})$. Occasionally, we will write this as $\Tr_{-\sA}(\ketbra{\psi})$, where the minus sign indicates tracing out all systems except $\sA$.

\subsection{Relations and variable-length registers}
\label{subsec:relation-states}

Fix a choice of $N = 2^n$. A relation $R$ is defined as a \emph{multiset} $R = \{(x_1, y_1), \dots, (x_t, y_t)\}$ of ordered pairs $(x_i, y_i) \in [N]^2$. This definition deviates slightly from the standard notion of a relation, which is typically an ordinary set of ordered pairs without repeated elements. The \emph{size} of the relation refers to the number of ordered pairs in the relation, including multiplicities. We denote this by $\abs{R}$, as the size corresponds to the cardinality of $R$ as a multiset.

\begin{definition}
\label{def:set-of-relations}
    Let $\calR$ denote the infinite set of all relations $R$. For any $t \geq 0$, let $\calR_t$ denote the set of all size-$t$ relations. 
\end{definition}

\begin{definition}
    For a relation $R$, we use $\Dom(R)$ to denote the set
    \begin{equation}
        \Dom(R) = \{x: x \in [N], \exists y \ \text{s.t.} (x,y) \in R\},
    \end{equation}
    and $\Im(R)$ to denote the set
    \begin{equation}
        \Im(R) = \{y: y \in [N], \exists x \ \text{s.t.} (x,y) \in R\}.
    \end{equation}
    Note that while $R$ may be a multi-set, $\Dom(R)$ and $\Im(R)$ are ordinary sets, i.e., they will not have repeated elements.
\end{definition}

Each relation $R \in \calR$ is associated with a \emph{relation state} $\ket*{R}$, defined as follows.

\begin{notation}[Relation states]
\label{def:relation-states-repeated}
For a relation $R = \{(x_1,y_1),\dots,(x_t,y_t)\}$, define the corresponding relation state $\ket*{R}$ to be the state
\begin{align}
    \ket*{R} \coloneqq \frac{\sum_{\pi \in \sSym_t} \ket*{x_{\pi(1)},y_{\pi(1)},\dots,x_{\pi(t)},y_{\pi(t)}}}{\sqrt{t!\cdot  \sum_{(x, y) \in [N]^2} \num(R,(x,y))!}}.
\end{align}
where $\num(R,(x,y))$ denotes the number of times the tuple $(x,y)$ appears in $R$.
\end{notation}

An elementary counting argument yields the following result.

\begin{fact}
    For any relation $R \in \calR$, the state $\ket*{R}$ is a unit vector.
\end{fact}

The relation states $\ket*{R}$ for $R \in \calR_t$ can also be viewed as the standard basis for the symmetric subspace of $(\mathbb{C}^{N^2})^{\otimes t}$. Note that this is only true because we allow for multi-set relations. Specifically, if $\Pi_{\ssym}^{N^2,t}$ denotes the projector onto the symmetric subspace of $(\mathbb{C}^{N^2})^{\otimes t}$, we have the equality
\begin{align}
    \Pi_{\ssym}^{N^2,t} = \sum_{R \in \calR_t} \ketbra{R}.
\end{align}
However, we will typically use the following notation to refer to this projector.

\begin{notation}
    For any integer $t \geq 0$, we define
    \begin{align}
        \Pi^{\calR}_{t} \coloneq \sum_{R \in \calR: \abs{R} = t} \ketbra*{R} = \Pi^{N^2,t}_{\ssym}.
    \end{align}
\end{notation}

\begin{notation}[Restricted sets of relations]
    Define the following restricted sets of relations:
    \begin{itemize}
        \item Let $\mathcal{R}_t^{\inj}$ be the set of all \emph{injective} relations, i.e., relations $R = \{(x_1,y_1),\dots,(x_t,y_t)\}$ of size $t$, where $(y_1,\dots,y_t) \in [N]^t_{\dist}$. Let $\mathcal{R}^{\inj} \coloneqq \cup_{t=0}^N \mathcal{R}_t^{\inj}$. 
        \item Let $\mathcal{R}_t^{\bij}$ be the set of all \emph{bijective} relations, i.e., relations $R = \{(x_1,y_1),\dots,(x_t,y_t)\}$ of size $t$, where $(x_1,\dots,x_t) \in [N]^t_{\dist}$ and $(y_1,\dots,y_t) \in [N]^t_{\dist}$. Let $\mathcal{R}^{\bij} \coloneqq \cup_{t=0}^N \mathcal{R}_t^{\bij}$. 
    \end{itemize}
\end{notation}
If the tuples in a relation $R = \{(x_1, y_1), \dots, (x_t, y_t)\}$ are distinct, i.e., $(x_i, y_i) \neq (x_j, y_j)$ for $i \neq j$, the normalization factor simplifies to $1/\sqrt{t!}$, i.e.,
\begin{align}
    \ket*{R} = \frac{1}{\sqrt{t!}} \sum_{\pi \in \sSym_t} \ket*{x_{\pi(1)}, y_{\pi(1)} , \dots, x_{\pi(t)}, y_{\pi(t)}}.
\end{align}
Note that any relation $R \in \calR^{\inj}$ or $R \in \calR^{\bij}$ satisfies this condition.

In both~\cref{part:standard,part:strong}, we will consider linear maps that send superpositions of $\ket*{R}$ for $R \in \calR_{t}$ to superpositions of $\ket*{R'}$ for $R' \in \calR_{t+1}$. This motivates the definition of \emph{variable-length registers}.

\subsubsection{Variable-length registers}
\label{prelim:variable-length-registers}

For every integer $t \geq 0$ let $\sR^{(t)}$ be a register associated with the Hilbert space $\calH_{\sR^{(t)}} \coloneqq (\mathbb{C}^{N} \otimes \mathbb{C}^{N})^{\otimes t}$. Let $\sR$ be a register corresponding to the infinite dimensional Hilbert space 
\begin{align}
    \calH_{\sR} \coloneqq \bigoplus_{t = 0}^\infty \calH_{\sR^{(t)}} = \bigoplus_{t = 0}^\infty (\mathbb{C}^{N} \otimes \mathbb{C}^{N})^{\otimes t}.
\end{align}
When $t = 0$, the space $(\mathbb{C}^{N} \otimes \mathbb{C}^{N})^{\otimes 0} = \mathbb{C}$ is a one-dimensional Hilbert space. Thus, $\calH_{\sR^{(t)}}$ is spanned by the states $\ket*{x_1,y_1,\dots,x_t,y_t}$ where $x_i,y_i \in [N]$. Note that the relation states $\ket*{R}$ for $R \in \calR_t$ span the symmetric subspace of $\calH_{\sR^{(t)}}$.

We will sometimes divide up the $\sR^{(t)}$ register into $\sR^{(t)} \coloneqq (\sR^{(t)}_{\sX},\sR^{(t)}_{\sY})$ where $\sR^{(t)}_{\sX}$ refers to the registers containing $\ket*{x_1,\dots,x_t}$ and $\sR^{(t)}_{\sY}$ refers to the registers containing $\ket*{y_1,\dots,y_t}$. We denote $\sR^{(t)}{\sX,i}$ as the register containing $\ket*{x_i}$ and $\sR^{(t)}{\sY,i}$ as the register containing $\ket*{y_i}$. Following our convention for defining the length/size of a relation $R$, we say that a state $\ket*{x_1,y_1,\dots,x_t,y_t}$ has length/size $t$. Two states of different lengths are orthogonal by definition, since $\calH_{\sR}$ is a direct sum $\bigoplus_{t=0}^\infty \calH_{\sR^{(t)}}$.

\begin{notation}[Extending fixed-length operators to variable-length]
    For any operator $O$ defined on the fixed-size Hilbert space $\calH_{\sR^{(t)}}$, we abuse notation by using $O$ to also refer to its extension on all of $\calH_{\sR}$. The extended operator is the direct sum of $O$ and the $0$ operator on $\calH_{\sR^{(t')}}$ for all $t' \neq t$. 
\end{notation}
Hence, if two operators $O_1$ and $O_2$ act on $\calH_{\sR^{(t)}}$ and $\calH_{\sR^{(t')}}$, respectively, then $O_1 + O_2$ is the sum of their extensions over all of $\calH_{\sR}$. We can now define the projector $\Pi^{\calR}$ that projects onto the span of all relation states.

\begin{notation} \label{not:pi-R-t-all-t}
    We define the projector
    \begin{align}
        \Pi^{\calR} \coloneqq \sum_{t = 0}^\infty \Pi^{\calR}_{t}
        = \sum_{R \in \calR} \ketbra*{R},
    \end{align}
    that projects onto the span of all relation states $\ket*{R}$ for all $R \in \calR$.
\end{notation}

Finally, we introduce the notion of variable-length tensor powers, which will be useful to describe applying an operator to each $\ket*{x_i,y_i}$ in a state $\ket*{x_1,y_1,\dots,x_t,y_t}$, in settings where $t$ is not explicitly known. 

\begin{notation}[Variable-length tensor powers]
For any unitary $U \in \mathcal{U}(N^2)$, let
\begin{align}
    U^{\otimes *} \coloneqq \sum_{t = 0}^{\infty} U^{\otimes t}
\end{align}
be a unitary that acts on the Hilbert space $\calH_{\sR}$.
\end{notation}

\subsubsection{Pairs of variable-length registers}

In~\cref{part:strong}, we will consider states of the form $\ket*{L}_{\gsL} \ket*{R}_{\gsR}$, where $\ket*{L}$ and $\ket*{R}$ are both relation states, and $\sL$ is another variable-length register defined analogously to $\sR$. Throughout~\cref{part:strong}, we will use the following definitions.

\begin{notation}[Fixed-length projectors] \label{not:fixed-leng-proj}
    For any integers $\ell,r \geq 0$, let $\Pi_{\ell,r}$ denote the projector acting on $\calH_{\sL} \otimes \calH_{\sR}$ that projects onto the fixed-length Hilbert space $\calH_{\sL^{(\ell)}} \otimes \calH_{\sR^{(r)}}$.
\end{notation}

\begin{notation}[Maximum-length projectors] \label{notation:pi-leq-t}
    For any integer $t \geq 0$, let $\Pi_{\leq t}$ denote the projector acting on $\calH_{\sL} \otimes \calH_{\sR}$ onto the Hilbert space $\bigoplus_{\ell,r \geq 0: \ell + r \leq t} \calH_{\sL^{(\ell)}} \otimes \calH_{\sR^{(r)}}$.
\end{notation}

\begin{notation}[Length-restricted operators]
\label{notation:length-restricted-ops}
For any operator $B$ that acts on the variable-length registers $\sR$ and $\sR$, let $B_{\ell,r} \coloneqq B \cdot \Pi_{\ell,r}$ denote the restriction of $B$ to input states where registers $\sR$ and $\sR$ have lengths $\ell$ and $r$. Let $B_{\leq t} \coloneqq B \cdot \Pi_{\leq t}$ denote the restriction of $B$ to inputs states where the combined length of $\sL$ and $\sR$ is at most $t$.
\end{notation}

Note that, with this notation, $(B{\leq t})^\dagger$ does not necessarily equal $(B^\dagger)_{\leq t}$. We adopt the convention that $B_{\leq t}^\dagger$ refers to $(B_{\leq t})^{\dagger}$.

\subsection{The Haar measure, unitary $t$-designs, and twirling channels}
\label{sec:haar-review}

\begin{definition}[Haar measure]
The Haar measure over the $n$-qubit unitary group $\calU(2^n)$ is the unique probability measure $\mu$ on $\calU(2^n)$ that is:
\begin{enumerate}
    \item Left-invariant: For any measurable set $S \subseteq \calU(2^n)$ and 
    any $V \in \calU(2^n)$, $\mu(VS) = \mu(S)$.
    \item Right-invariant: For any measurable set $S \subseteq \calU(2^n)$ and 
    any $V \in \calU(2^n)$, $\mu(SV) = \mu(S)$.
    \item Normalized: $\mu(\calU(2^n)) = 1$.
\end{enumerate}
The Haar measure provides a notion of uniform distribution over the unitary 
group.
\end{definition}

We will refer to the Haar measure as $\mu_{\mathsf{Haar}}$.

\begin{definition}[Unitary $t$-design]
\label{def:unitary-design}
A distribution $\frakD$ on $n$-qubit unitaries is a unitary $t$-design if
\begin{equation}
    \E_{U \sim \frakD} [U^{\otimes t} \otimes U^{\dagger,\otimes t}] = 
    \int_{\calU(2^n)} U^{\otimes t} \otimes U^{\dagger,\otimes t} d\mu(U),
\end{equation}
where $\mu$ is the Haar measure over the unitary group $\calU(2^n)$.
\end{definition}

\begin{notation}
    Define the equality projector
    \begin{equation} \label{eq:eq-projector}
    \Pi^{\mathsf{eq}} = \sum_{x \in [N]} \ketbra*{x} \otimes \ketbra*{x}.
    \end{equation}
\end{notation}

In the following, when we write $\mathbb{E}_{\psi}$ and $\mathbb{E}_{U}$ without any specified distribution, we always refer to the uniform distribution over pure states and the Haar measure over unitary groups, respectively. We will use the following standard fact about Haar-random states and the symmetric subspace.

\begin{fact}
\label{prop:sym-haar}
    The expectation over Haar measure satisfies
    \begin{align}
        \E_{\psi \gets \mathbb{C}^N} \ketbra*{\psi}^{\otimes 2} = \frac{\Pi_{\ssym}^{N,2}}{\Tr(\Pi_{\ssym}^{N,2})} = \frac{\Pi_{\ssym}^{N,2}}{\binom{N+1}{2}},
    \end{align}
    where $\Pi_{\ssym}^{N,k}$ is the projector onto the symmetric subspace of $(\mathbb{C}^N)^{\otimes k}$.
\end{fact}

We will use the following elementary claim about unitary $2$-designs in~\cref{part:standard,part:strong}.

\begin{claim}[Standard twirling]
\label{claim:equality-clifford-twirl}
    For any $n$-qubit unitary $2$-design $\frakD$,
    \begin{align}
        \E_{U\gets \frakD} \Big[ (U \otimes U)^\dagger \cdot \Pi^{\mathsf{eq}} \cdot (U \otimes U)\Big] = \frac{2}{N+1} \cdot \Pi_{\ssym}^{N, 2}.
    \end{align}
\end{claim}
\begin{proof}
    \begin{align}
        \E_{U \gets \frakD} \Big[ (U^\dagger \otimes U^\dagger) \cdot \Pi^{\mathsf{eq}} \cdot (U \otimes U)\Big] &= \E_{U \gets \frakD} \sum_{x \in [N]} U^\dagger \ketbra*{x} U \otimes U^\dagger \ketbra*{x} U \tag{definition of $\Pi^{\mathsf{eq}}$}\\
        &= \E_{U} \sum_{x \in [N]} U^\dagger \ketbra*{x} U \otimes U^\dagger \ketbra*{x} U \tag{$\frakD$ is a $2$-design}\\
        &= N \cdot \E_{\psi} \ketbra*{\psi} \otimes \ketbra*{\psi} \tag{$U^\dagger \ket*{x}$ is a Haar-random state}\\
        &= \frac{2}{N+1} \cdot \Pi_{\ssym}^{N, 2} \tag{\cref{prop:sym-haar}}.
    \end{align}
\end{proof}

From the above claim, we immediately obtain the following lemma, which was also also used by~\cite{metger2024simple} to construct \emph{non-adaptive} PRUs.

\begin{lemma}[Twirling into the distinct subspace] \label{lem:almost-distinct-on-X}  
    Given two integers $n, t > 0$.
    Define the distinct subspace projector acting on $nt$ qubits as follows,
    \begin{equation}
        \Pi^{\dist} \coloneq \sum_{(x_1, \ldots, x_t) \in [N]_{\dist}^t} \ketbra*{x_1} \otimes \ldots \otimes \ketbra*{x_t}.
    \end{equation}
    For any $n$-qubit unitary $2$-design $\frakD$ and any state $\rho$ on at least $nt$ qubits, we have
    \begin{equation}
    \label{eq:pi-dist-rho-D}
    \Tr(\E_{ C \gets \frakD} (\Pi^{\dist} \otimes \Id) \cdot (C^{\otimes t} \otimes \Id) \cdot \rho \cdot (C^{\dagger, \otimes t} \otimes \Id) \cdot (\Pi^{\dist} \otimes \Id) ) \geq 1 - \frac{t(t-1)}{N+1}.
    \end{equation}
\end{lemma}
\begin{proof}
From the definition of the distinct subspace projector, we have
\begin{equation}
    \Id - \Pi^{\dist} = \sum_{(x_1, \ldots, x_t) \in [N]^t \setminus [N]_{\dist}^t} \ketbra*{x_1, \ldots, x_t}.
\end{equation}
Because for any $(x_1, \ldots, x_t)\in [N]^t \setminus [N]_{\dist}^t$, there exists $i\neq j$, such that $x_i = x_j$, we have
    \begin{equation}
        \sum_{(x_1, \ldots, x_t) \in [N]^t \setminus [N]_{\dist}^t} \ketbra*{x_1, \ldots, x_t} \preceq \sum_{1 \leq i < j \leq t} \Pi^{\mathsf{eq}}_{\darkgray{\sX_i}, \darkgray{\sX_j}},
    \end{equation}
    where $\preceq$ here denotes the PSD order and $\Pi^{\mathsf{eq}}_{\darkgray{\sX_i}, \darkgray{\sX_j}}$ is the equality projector in Eq.~\eqref{eq:eq-projector} on the $i$-th and $j$-th $n$-qubit register $\sX_{i}, \sX_{j}$. 
This implies the following:
    \begin{align} \label{eq:1minusdist}
        &1 - \Tr(\E_{ C \gets \frakD} (\Pi^{\dist} \otimes \Id) \cdot (C^{\otimes t} \otimes \Id) \cdot \rho \cdot (C^{\dagger, \otimes t} \otimes \Id) \cdot (\Pi^{\dist} \otimes \Id) )\\
        &= 1 - \Tr(\E_{ C \gets \frakD} (\Pi^{\dist} \otimes \Id) \cdot (C^{\otimes t} \otimes \Id) \cdot \rho \cdot (C^{\dagger, \otimes t} \otimes \Id) )\\
        &= \Tr(\E_{ C \gets \frakD} \left( \sum_{(x_1, \ldots, x_t) \in [N]^t \setminus [N]_{\dist}^t} \ketbra*{x_1, \ldots, x_t} \otimes \Id \right) \cdot (C^{\otimes t} \otimes \Id) \cdot \rho \cdot (C^{\dagger, \otimes t} \otimes \Id) ) \\
        &\leq \sum_{1 \leq i < j \leq t} \E_{ C \gets \frakD} \Tr( (\Pi^{\mathsf{eq}}_{\darkgray{\sX_i}, \darkgray{\sX_j}} \otimes \Id)  \cdot (C^{\otimes t} \otimes \Id) \cdot \rho \cdot (C^{\dagger, \otimes t} \otimes \Id) )\\
        &= \sum_{1 \leq i < j \leq t} \E_{ C \gets \frakD} \Tr( \Pi^{\mathsf{eq}}_{\darkgray{\sX_i}, \darkgray{\sX_j}}  \cdot C^{\otimes 2} \cdot \rho_{\darkgray{\sX_i}, \darkgray{\sX_j}} \cdot C^{\dagger, \otimes 2} ) \tag{where $\rho_{\darkgray{\sX_i}, \darkgray{\sX_j}} \coloneq \Tr_{-\sX_i, \sX_j}(\rho)$} \\
        &= \sum_{1 \leq i < j \leq t} \frac{2}{N+1} \Tr(\Pi_{\ssym}^{N, 2} \cdot \rho_{\darkgray{\sX_i}, \darkgray{\sX_j}}) \leq \sum_{1 \leq i < j \leq t} \frac{2}{N+1} = \frac{t(t-1)}{N}.
    \end{align}
    This completes the proof.
\end{proof}

The following claim will only be used in~\cref{part:strong}.

\begin{notation}
    Let
    \begin{align}
        \ket*{\EPR_N} \coloneqq \frac{1}{\sqrt{N}} \sum_{x \in [N]} \ket*{x} \ket*{x}.
    \end{align}
\end{notation}

\begin{claim}[Mixed twirling]
\label{claim:equality-clifford-conjugated-twirl}
    For any $n$-qubit unitary $2$-design $\frakD$,
    \begin{align}
        \E_{U\gets \frakD} \Big[ (U \otimes \overline{U})^\dagger \cdot \Pi^{\mathsf{eq}} \cdot (U \otimes \overline{U})\Big] = \ketbra*{\EPR_N} + \frac{1}{N+1}(\Id - \ketbra*{\EPR_N}).
    \end{align}
\end{claim}
\begin{proof}
    Label the registers that $U$ and $\overline{U}$ act on as $A$ and $B$ respectively. For any operator $X$ acting on $A,B$, define the partial transpose as
    \begin{align}
        X^{T_B} = \Big( \sum_{i,j,k,\ell} X_{ijkl} \ketbra*{i}{j}_A \otimes \ketbra*{k}{\ell}_B \Big)^{T_B} =  \sum_{i,j,k,\ell} X_{ijkl} \ketbra*{i}{j}_A \otimes \ketbra*{\ell}{k}_B.
    \end{align}
    We will use the identity
    \begin{align}
        (U \otimes \overline{U})^\dagger \cdot X \cdot (U \otimes \overline{U}) = \Big( (U \otimes U)^\dagger \cdot X^{T_B} \cdot (U \otimes U) \Big)^{T_B}.
    \end{align}
    Since $(\Pi^{\eq})^{T_B} = \Pi^{\eq}$,
    \begin{align}
        &\E_{U\gets \frakD} (U \otimes \overline{U})^\dagger \cdot \Pi^{\mathsf{eq}} \cdot (U \otimes \overline{U})\\
        &= \Big( \E_{U \gets \frakD} (U \otimes U)^\dagger \cdot \Pi^{\mathsf{eq}} \cdot (U \otimes U)\Big)^{T_B}\\
        &= \Big( \frac{2}{N+1} \cdot \Pi^{N,2}_{\ssym} \Big)^{T_B} \tag{by~\cref{claim:equality-clifford-twirl}}\\
        &= \frac{2}{N+1} \cdot \Bigg( \sum_{x \in [N]} \ketbra*{xx}{xx} + \sum_{x,y \in [N], x<y} \Big(\frac{\ket*{xy} + \ket*{yx}}{\sqrt{2}} \Big) \Big(\frac{\bra*{xy} + \bra*{yx}}{\sqrt{2}} \Big) \Bigg)^{T_B}\\
        &= \frac{2}{N+1} \cdot \Bigg( \sum_{x \in [N]} \ketbra*{xx}{xx} + \frac{1}{2} \sum_{x,y \in [N], x<y} \Big(\ketbra*{xy} + \ketbra*{xy}{yx} + \ketbra*{yx}{xy} + \ketbra*{yx} \Big) \Bigg)^{T_B}\\
        &= \frac{2}{N+1} \cdot \Bigg( \sum_{x \in [N]} \ketbra*{xx}{xx} + \frac{1}{2} \sum_{x,y \in [N], x<y} \Big(\ketbra*{xy} + \ketbra*{xx}{yy} + \ketbra*{yy}{xx} + \ketbra*{yx} \Big) \Bigg)\\
        &= \frac{2}{N+1} \cdot \Bigg( \frac{1}{2}\sum_{x,y \in [N]} \ketbra*{xx}{yy} + \frac{1}{2}\sum_{x,y \in [N]} \ketbra*{xy}\Bigg) \\
        &= \frac{1}{N+1} \cdot \Id + \frac{N}{N+1} \ketbra*{\EPR_N} \\
        &= \ketbra*{\EPR_N} + \frac{1}{N+1}(\Id - \ketbra*{\EPR_N}).
    \end{align}
    This completes the proof.
\end{proof}

\subsection{Oracle adversaries}

We first define oracle adversaries that make only forward queries to an $n$-qubit unitary oracle $\calO$. This definition will be used exclusively in~\cref{part:standard}.

\begin{definition}[Oracle adversaries with forward queries, used in~\cref{part:standard}]
    A $t$-query oracle adversary $\Adv$ that makes only forward queries is parameterized by a sequence of $(n+m)$-qubit unitaries $(A_1,\dots,A_t)$, which act on registers $(\sA,\sB)$, where $\sA$ is the $n$-qubit query register and $\sB$ is an $m$-qubit ancilla. We assume without loss of generality that the adversary's initial state is $\ket*{0^{n+m}}_{\gsA \gsB}$. The state of the algorithm after $t$ queries to $\calO$ is
    \begin{align}
        \ket*{\Adv_t^{\calO}}_{\gsA \gsB} \coloneqq \prod_{i = 1}^t \Big( \calO_{\gsA} \cdot A_{i,\gsA \gsB} \Big) \ket*{0^{n+m}}_{\gsA \gsB}.
    \end{align}
\end{definition} 

We also define an oracle adversary that can make both forward and inverse queries to an $n$-qubit unitary oracle $\calO$. This definition will be used exclusively in~\cref{part:strong}.

\begin{definition}[Oracle adversaries with forward and inverse queries, used in~\cref{part:strong}]
    A $t$-query oracle adversary $\Adv$ that makes both forward and inverse queries is parameterized by 
    \begin{itemize}
        \item a sequence of $(n+m)$-qubit unitaries $(A_1,\dots,A_t)$, which act on registers $(\sA,\sB)$, where $\sA$ is the $n$-qubit query register and $\sB$ is an $m$-qubit ancilla, and
        \item a sequence of bits $(b_1,\dots,b_t)$ where $b_i = 0$ means that the adversary's $i$th oracle query is to $\calO$, and $b_i = 1$ means that query is to $\calO^\dagger$.
    \end{itemize}
    We assume without loss of generality that the adversary's initial state is $\ket*{0^{n+m}}_{\gsA \gsB}$. The state of the algorithm after $t$ queries to $\calO$ is
    \begin{align}
        \ket*{\Adv_t^{\calO}}_{\gsA \gsB} \coloneqq \prod_{i = 1}^t \Bigg( \Big( (1-b_i) \cdot \calO_{\gsA} + b_i \cdot \calO^\dagger_{\gsA} \Big) \cdot A_{i,\gsA \gsB} \Bigg) \ket*{0^{n+m}}_{\gsA \gsB}.
    \end{align}
\end{definition} 

\subsection{Pseudorandom unitaries}

\begin{definition}[pseudorandom unitaries] \label{def: PRU-t(lambda)}
We say $\{ \mathcal{U}_n \}_{n \in \mathbb{N}}$ is a secure PRU if, for all $n \in \mathbb{N}$, $\mathcal{U}_n = \{ U_{k} \}_{k \in \mathcal{K}_{n}}$ is a set of $n$-qubit unitaries where $\mathcal{K}_{n}$ denotes the keyspace, satisfying the following:
\begin{itemize}
    \item \textbf{Efficient computation:} There exists a $\mathrm{poly}(n)$-time quantum algorithm that implements the $n$-qubit unitary $U_{k}$ for all $k \in \mathcal{K}_{n}$.
    \item \textbf{Indistinguishability from Haar:} For any oracle adversary $\calA$ that runs in time $\poly(n)$ (the runtime is the total number of gates that $\calA$ uses, counting oracle gates as $1$), and measures a two-outcome observable $D_{\calA}$ with eigenvalues $\{0, 1\}$ after the queries, we have
    \begin{equation} \label{eq:distinguishability-adv}
        \left| \E_{\mathcal{O} \leftarrow \mathcal{U}_n}  \Tr\left( \, D_{\calA} \cdot \ketbra*{\calA^{\calO}}_{\gsA \gsB} \, \right) - \E_{\mathcal{O} \sim \mathsf{Haar}} \Tr\left( \, D_{\calA} \cdot \ketbra*{\calA^{\calO}}_{\gsA \gsB} \, \right) \right| \leq \mathsf{negl}(n),
    \end{equation}
    where $\mathsf{negl}(n)$ is any function that is $o(1/n^c)$ for all $c >0$. 
\end{itemize}
A \textbf{standard PRU} (i.e., the original~\cite{ji2018pseudorandom} notion) is one where indistinguishability holds against oracle adversaries that only make forward queries to $\calO$. A \textbf{strong PRU} is one where indistinguishability holds against oracle adversaries that make both forward and inverse queries to $\calO$.
\end{definition}

\subsection{Useful lemmas}

The following lemma will be used in~\cref{part:standard} to bound the distance between a pair of mixed states who purifications are related by a projection that acts only on the purifying register. 

\begin{lemma}
\label{lemma:gentle}
    Let $\rho_{\gsC \gsD}$ be a density matrix on registers $\sC, \sD$ and let $\Pi_{\gsD}$ be a projector that acts on register $\sD$. Then
    \begin{align}
        \norm{\Tr_{\sD}(\rho_{\gsC \gsD}) - \Tr_{\sD}(\Pi_{\gsD} \cdot \rho_{\gsC \gsD} \cdot \Pi_{\gsD})}_1 = 1- \Tr(\Pi_{\gsD} \cdot \rho_{\gsC \gsD}).
    \end{align}
\end{lemma}
\begin{proof}
    We can decompose $\Tr_{\sD}(\rho_{\gsC \gsD})$ as follows:
    \begin{align}
        \Tr_{\sD}(\rho_{\gsC \gsD}) &= \Tr_{\sD}(\rho_{\gsC \gsD} \cdot \Pi_{\gsD}) + \Tr_{\sD}(\rho_{\gsC \gsD} \cdot (\Id - \Pi_{\gsD}))\\
        &= \Tr_{\sD}(\Pi_{\gsD} \cdot \rho_{\gsC \gsD} \cdot \Pi_{\gsD}) + \Tr_{\sD}((\Id - \Pi_{\gsD}) \cdot \rho_{\gsC \gsD} \cdot (\Id - \Pi_{\gsD})) \label{eq:cyclic-trace-lemma}
    \end{align}
    where the second equality uses the fact that $\Pi_{\gsD} = \Id_{\gsC} \otimes \Pi'_{\gsD}$, which allows us to invoke the cyclic property of $\Tr_{\sD}$. Using~\cref{eq:cyclic-trace-lemma}, we have
    \begin{align}
        &\norm{\Tr_{\sD}(\rho_{\gsC \gsD}) - \Tr_{\sD}(\Pi_{\gsD} \cdot \rho_{\gsC \gsD} \cdot \Pi_{\gsD})}_1\\
        &= \norm{\Tr_{\sD}((\Id - \Pi_{\gsD}) \cdot \rho_{\gsC \gsD} \cdot (\Id - \Pi_{\gsD}))}_1\\
        &= \Tr((\Id - \Pi_{\gsD}) \cdot \rho_{\gsC \gsD} \cdot (\Id - \Pi_{\gsD})) \tag{since $\norm{M}_1 = \Tr(M)$ for PSD $M$}\\
        &= \Tr((\Id - \Pi_{\gsD}) \cdot \rho_{\gsC \gsD})\\
        &= 1 - \Tr(\Pi_{\gsD} \cdot \rho_{\gsC \gsD}).
    \end{align}
\end{proof}

We will use the following ``sequential'' gentle measurement lemma in~\cref{part:strong}.

\begin{lemma}[sequential gentle measurement]  \label{lem:seq-gentleM-pure}
    Let $\ket*{\psi}$ be a normalized state, $P_1,\dots,P_t$ be projectors, and $U_1,\dots,U_t$ be unitaries.
    \begin{align}
        \norm{U_t \ldots U_1 \ket*{\psi} -  P_t U_t \ldots P_{1} U_1 \ket*{\psi}}_2 \leq t \sqrt{1 - \norm{P_t U_t \ldots P_{1} U_1 \ket*{\psi}}_2^2}.
    \end{align}
\end{lemma}

To prove this, we will need the following version of the standard gentle measurement lemma.

\begin{lemma}[gentle measurement] \label{lem:gentleM-pure}
For any projector $\Pi$ and sub-normalized state $\ket*{\psi}$ satisfying $\braket{\psi} \leq 1$, we have
\begin{equation}
    \norm{ (\Id - \Pi) \ket*{\psi}}_2 \leq \sqrt{1 - \norm{ \Pi \ket*{\psi} }_2^2 }.
\end{equation}
\end{lemma}
\begin{proof}[Proof of~\cref{lem:gentleM-pure}]
    By direct expansion, we have
    \begin{align}
        \norm{\ket*{\psi} - \Pi \ket*{\psi}}_2^2 = \bra*{\psi} (\Id - \Pi) \ket*{\psi} = \braket{\psi} - \bra*{\psi} \Pi \ket*{\psi} \leq 1 - \norm{\Pi \ket*{\psi}}^2_2.
    \end{align}
\end{proof}

\begin{proof}[Proof of~\cref{lem:seq-gentleM-pure}]
    We prove this lemma by induction.
    For $t = 0$, we have $\norm{\ket*{\psi} - \ket*{\psi}}_2 = 0 = 1 - \norm{\ket*{\psi}}_2^2$. So the base case holds.
    Suppose the inductive hypothesis holds for $t-1$, i.e.,
    \begin{align}
        \norm{U_{t-1} \ldots U_1 \ket*{\psi} -  P_{t-1} U_{t-1} \ldots P_{1} U_1 \ket*{\psi}}_2 & \leq (t-1) \sqrt{1 - \norm{P_{t-1} U_{t-1} \ldots P_{1} U_1 \ket*{\psi}}_2^2}\\
        & = (t-1) \sqrt{1 - \norm{U_t P_{t-1} U_{t-1} \ldots P_{1} U_1 \ket*{\psi}}_2^2}\\
        & \leq (t-1) \sqrt{1 - \norm{P_t U_t P_{t-1} U_{t-1} \ldots P_{1} U_1 \ket*{\psi}}_2^2}.
    \end{align}
    The second line uses the unitary invariance of $\norm{\cdot}_2$.
    The third line uses the fact that $P_t$ is a projector and hence cannot increase the norm.
    We can use the unitary invariance of $\norm{\cdot}_2$ to obtain
    \begin{equation}
        \norm{U_{t-1} \ldots U_1 \ket*{\psi} -  P_{t-1} U_{t-1} \ldots P_{1} U_1 \ket*{\psi}}_2 = \norm{U_{t} \ldots U_1 \ket*{\psi} -  U_t P_{t-1} U_{t-1} \ldots P_{1} U_1 \ket*{\psi}}_2.
    \end{equation}
    Next we use \cref{lem:gentleM-pure} to obtain
    \begin{equation}
        \norm{(\Id - P_t) U_t P_{t-1} U_{t-1} \ldots P_{1} U_1 \ket*{\psi}}_2 \leq \sqrt{1 - \norm{P_t U_t \ldots P_{1} U_1 \ket*{\psi}}_2^2}.
    \end{equation}
    Together, we have
    \begin{align}
        &\norm{U_t \ldots U_1 \ket*{\psi} -  P_t U_t \ldots P_{1} U_1 \ket*{\psi}}_2\\
        &\leq \norm{U_{t-1} \ldots U_1 \ket*{\psi} -  P_{t-1} U_{t-1} \ldots P_{1} U_1 \ket*{\psi}}_2 + \norm{(\Id - P_t) U_t P_{t-1} U_{t-1} \ldots P_{1} U_1 \ket*{\psi}}_2\\
        &\leq (t-1) \sqrt{1 - \norm{P_t U_t \ldots P_{1} U_1 \ket*{\psi}}_2^2} + \sqrt{1 - \norm{P_t U_t \ldots P_{1} U_1 \ket*{\psi}}_2^2}.
    \end{align}
    This concludes the proof.
\end{proof}

\newpage
\part{Standard PRUs}
\label{part:standard}
\vspace{1em}

The goal of~\cref{part:standard} is to construct standard PRUs (i.e., the definition of~\cite{ji2018pseudorandom}), which are secure against adversaries that only make forward queries to the unitary oracle.

\section{The purified permutation-function oracle}
\label{sec:PF-oracle}

In this section, we analyze the view of an adversary that makes forward queries to an oracle for $P_{\pi} \cdot F_{f}$, for uniformly random $\pi \gets \sSym_N$ and $f \gets \{0,1\}^N$. These operators are defined as
\begin{align}
    P_{\pi} \coloneqq \sum_{x \in [N]} \ketbra{\pi(x)}{x} \quad \text{and} \quad F_f \coloneqq \sum_{x \in [N]} (-1)^{f(x)} \ketbra{x}.
\end{align}
Our first step will be to consider a purification of the adversary's state where the randomness of $\pi$ and $f$ is replaced by the uniform superposition
\begin{align}
    \frac{1}{\sqrt{N!}} \sum_{\pi \in \sSym_N} \ket*{\pi}_{\gsP} \otimes \frac{1}{\sqrt{2^N}} \sum_{f \in \{0,1\}^N} \ket*{f}_{\gsF},
\end{align}
and each query is implemented by the purified permutation-function oracle $\pfo$, which applies $P_{\pi} \cdot F_f$ controlled on $\ket*{\pi} \ket*{f}$.

\begin{definition}[purified permutation-function oracle]
    The purified permutation-function oracle $\pfo$ is a unitary acting on registers $\sA,\sP,\sF$, where \begin{itemize}
        \item $\sP$ is a register associated with the Hilbert space $\calH_{\sP}$,  defined to be the span of the orthonormal states $\ket*{\pi}$ for all $\pi \in \sSym_N$. 
        \item $\sF$ is a register associated with the Hilbert space $\calH_{\sF}$, defined to be the span of the orthonormal states $\ket*{f}$ for all $f \in \{0,1\}^N$. 
    \end{itemize}
    The unitary $\pfo$ is defined to act as follows:
    \begin{equation}
        \pfo_{\gsA\gsP\gsF} \ket*{x}_{\gsA} \ket*{\pi}_{\gsP} \ket*{f}_{\gsF} \coloneqq (-1)^{f(x)} \ket*{\pi(x)}_{\gsA} \ket*{\pi}_{\gsP} \ket*{f}_{\gsF}, \label{eq:pfo-definition}
    \end{equation}
    for all $x \in [N], \pi \in \sSym_N,$ and $f \in \{0, 1\}^N$.
\end{definition}

% We will do this by considering the \emph{purified permutation-function} oracle, which uses a purification of the random $\pi$ andwhich the randomness of $\pi$ and $f$ are  by a uni uses a purification of $\pi$ and $f$

When $\sP$ and $\sF$ are initialized to the uniform superposition over permutations and functions respectively, the view of an adversary that queries the $\pfo$ is equivalent to the view of an adversary that queries the standard oracle $P_\pi \cdot F_f$, for uniformly random $\pi \gets \sSym_N$ and $f \gets \{0,1\}^N$. 

\begin{claim}[Equivalence of the purified and standard oracles] \label{claim:purified-vs-standard-PFO}
    For any oracle adversary $\Adv$, the following oracle instantiations are perfectly indistinguishable:
    \begin{itemize}
        \item (Queries to a random $P_\pi \cdot F_f$) Sample a uniformly random $\pi \gets \sSym_N, f \gets \{0,1\}^N$. On each query, apply $P_\pi \cdot F_f$ to register $\sA$.
        \item (Queries to $\pfo$) Initialize registers $\sP,\sF$ to $\frac{1}{\sqrt{N!}} \sum_{\pi \in \sSym_N} \ket*{\pi}_{\gsP} \otimes \frac{1}{\sqrt{2^N}} \sum_{f \in \{0,1\}^N} \ket*{f}_{\gsF}$. At each query, apply $\pfo$ to registers $\sA,\sP,\sF$.
    \end{itemize}
\end{claim}
\begin{proof}
    Since the adversary's view does not contain the $\sP,\sF$ registers, the adversary's view in the second case is unchanged if the $\sP,\sF$ registers are measured at the end. 
    Since $\pfo$ is controlled on the $\sP,\sF$ registers, the queries to $\pfo$ commute with the measurement of the $\sP,\sF$ registers.
    Hence, measuring the $\sP,\sF$ registers at the end produces the same view as measuring at the beginning, which is equivalent to the first case.
\end{proof}

The key to understanding the oracle $\pfo$ is to consider how it acts on the following ``$\pf$-relation states'', defined below.

\begin{definition}[$\pf$-relation state]
    For $0 \leq t \leq N$ and $R = \{(x_1, y_1), \dots, (x_t, y_t)\} \in \calR_t$, let
    \begin{align}
        \ket*{\pf_R}_{\gsP \gsF} \coloneqq \frac{1}{\sqrt{(N-t)!}} \sum_{\pi \in \sSym_N} \delta_{\pi,R} \ket*{\pi}_{\gsP} \otimes \frac{1}{\sqrt{2^N}} \sum_{f \in \{0,1\}^N} (-1)^{f(x_1) + \cdots + f(x_t)} \ket*{f}_{\gsF},\label{eq:def-phi-S}
    \end{align}
    where $\delta_{\pi,R}$ is an indicator variable that equals $1$ if $\pi(x) = y$ for all $(x,y) \in R$, and is $0$ otherwise.
\end{definition}

Note that for $t = 0$ and $R = \varnothing$, the $\pf$-relation state $\ket*{\pf_{\varnothing}}_{\gsP \gsF}$ is the uniform superposition over all permutations $\pi \in \sSym_N$ and all functions $f \in \{0, 1\}^N$,
\begin{align} \label{eq:phi-varnothing}
    \ket*{\pf_{\varnothing}}_{\gsP \gsF} \coloneqq \frac{1}{\sqrt{N!}} \sum_{\pi \in \sSym_N} \ket*{\pi}_{\gsP} \otimes \frac{1}{\sqrt{2^N}} \sum_{f \in \{0,1\}^N} \ket*{f}_{\gsF}.
\end{align}

\subsection{Orthonormality of the $\pf$-relation states}

\begin{claim}[Orthonormality of the distinct sets of $\pf$-relation states]
\label{claim:phi-S-orthogonal}
    $\{\ket*{\pf_R}\}_{R \in \calR^{\bij}}$ forms a set of orthonormal vectors.
\end{claim}

\begin{proof}[Proof of~\cref{claim:phi-S-orthogonal}]
    We first recall the definition of $\ket*{\pf_R}$:
    \begin{align}
        \ket*{\pf_R}_{\gsP\gsF} = \frac{1}{\sqrt{(N-t)!}} \sum_{\pi \in \sSym_N} \delta_{\pi,R} \ket*{\pi}_{\gsP} \otimes \frac{1}{\sqrt{2^N}} \sum_{f \in \{0,1\}^N} (-1)^{f(x_1) + \cdots + f(x_t)} \ket*{f}_{\gsF}.
    \end{align}
    For $x \in [N]$, let $e_x \in \{0,1\}^N$ denote the $N$-dimensional vector that has a $1$ in the $x$-th position, and is $0$ everywhere else. Then by writing $f(x)$ as $f(x) = f \cdot e_x$, we get
    \begin{align}
        \frac{1}{\sqrt{2^N}} \sum_{f \in \{0,1\}^N} (-1)^{f(x_1) + \cdots + f(x_t)} \ket*{f}_{\gsF} &= \frac{1}{\sqrt{2^N}} \sum_{f \in \{0,1\}^N} (-1)^{f \cdot (e_{x_1} + \cdots + e_{x_t})} \ket*{f}_{\gsF} \\
        &= H^{\otimes N} \ket*{e_{x_1} + \cdots + e_{x_t} \ (\mathrm{mod} \ 2)}_{\gsF}.
    \end{align}
    When $x_1,\dots,x_t$ are distinct, $e_{x_1} + \cdots + e_{x_t} (\mod 2)$ is a vector in $\{0,1\}^N$ whose $x$-th entry is $1$ if $x \in \{x_1,\dots,x_t\}$, and $0$ otherwise. Since this is simply the indicator vector for the set $\{x_1,\dots,x_t\}$, there exists an isometry that maps
    \begin{align}
        \frac{1}{\sqrt{2^N}} \sum_{f \in \{0,1\}^N} (-1)^{f(x_1) + \cdots + f(x_t)} \ket*{f}_{\gsF} \mapsto \ket*{\{x_1,\dots,x_t\}}.
    \end{align}
    Applying this to the $\sF$ register of $\ket*{\pf_R}$, this tells us there is an isometry $M$ such that for all $R \in\calR^{\bij}$, 
    \begin{align}
        M: \ket*{\pf_R} \mapsto \frac{1}{\sqrt{(N-t)!}} \sum_{\pi \in \sSym_N} \delta_{\pi,R} \ket*{\pi}_{\gsP} \otimes \ket*{\{x_1,\dots,x_t\}}.
    \end{align}
    Consider $R, S \in\calR^{\bij}$, where $R = \{(x_1,y_1),\dots,(x_{\abs{R}},y_{\abs{R}})\}$ and $S = \{(x'_1,y'_1),\dots,(x'_{\abs{S}},y'_{\abs{S}})\}$.
    \begin{align}
        \braket{\pf_R}{\pf_S} &= \bra*{\pf_R} M^\dagger \cdot M \ket*{\pf_S}\\
        &= \frac{1}{\sqrt{(N-\abs{R})! (N-\abs{S})!}} \sum_{\pi \in \sSym_N} \delta_{\pi,R} \cdot \delta_{\pi,S} \braket*{\{x_1,\dots,x_{\abs{R}}\}}{\{x'_1,\dots,x'_{\abs{S}}\}} \label{eq:inner-prod-phi}.
    \end{align}
    This expression is equal to zero if $\Dom(R) \neq \Dom(S)$ due to the $\braket*{\{x_1,\dots,x_{\abs{R}}\}}{\{x'_1,\dots,x'_{\abs{S}}\}}$ term. Thus, it remains to consider $R,S$ such that $\Dom(R) = \Dom(S)$. This means that $\abs{R} = \abs{S}$ and thus~\cref{eq:inner-prod-phi} simplifies to
    \begin{align}
        \frac{1}{(N-\abs{R})!} \sum_{\pi \in \sSym_N} \delta_{\pi,R} \cdot \delta_{\pi,S}.
    \end{align}
    There are two cases to consider:
    \begin{itemize}
        \item In the first case, $R \neq S$. Then there exists $x, y, y'$ such that $(x,y) \in R$, $(x,y') \in S$, and $y \neq y'$. But then the above expression will be $0$, since there are no permutations $\pi$ satisfying both $\pi(x) = y$ and $\pi(x) = y'$. 
        \item In the other case, $R = S$. Then the sum is over all permutations $P$ such that $\pi(x) = y$ for all $(x,y) \in R$. There are $(N - \abs{R})!$ such permutations, and so in this case the sum becomes $1$. 
    \end{itemize}
    This completes the proof.
\end{proof}

\subsection{How $\pfo$ acts on the $\pf$-relation states}

\begin{claim}[Action of $\pfo$ on $\pf$-relation states]
\label{claim:pfo-action}
    For $0 \leq t < N$, $R \in \calR_t$ and $x \in [N]$,
    \begin{align}
        \pfo \ket*{x}_{\gsA} \ket*{\pf_R}_{\gsP \gsF} = \frac{1}{\sqrt{N-\abs{R}}} \sum_{y \in [N]} \ket*{y}_{\gsA} \ket*{\pf_{R \cup \{(x,y)\}}}_{\gsP \gsF}. \label{eq:pfo-map}
    \end{align}
\end{claim}

\begin{proof}[Proof of~\cref{claim:pfo-action}]
    From the definitions of $\pfo$ and $\ket*{\pf_R}$ (\cref{eq:pfo-definition} and \cref{eq:def-phi-S}), we have
    \begin{align}
        &\pfo \ket*{x}_{\gsA} \ket*{\pf_R}_{\gsP\gsF} \nonumber \\
        &= \sum_{\pi \in \sSym_N} (-1)^{f(x)}\ket*{\pi(x)}_{\gsA} \frac{1}{\sqrt{(N-t)!}}
        \delta_{\pi,R} \ket*{\pi}_{\gsP}  \otimes \frac{1}{\sqrt{2^N}} \sum_{f \in \{0,1\}^N} (-1)^{f(x_1) + \cdots + f(x_t)} \ket*{f}_{\gsF}.\label{eq:plug-in-pfo}
    \end{align}
    We now rewrite the right-hand side of~\cref{eq:plug-in-pfo} using the substitution $\ket*{\pi(x)} = \sum_{y \in [N]} \delta_{\pi(x) =y} \ket*{y}$. This gives
    \begin{align}
        (\ref{eq:plug-in-pfo}) &=  \sum_{\pi \in \sSym_N} (-1)^{f(x)} \sum_{y \in [N]} \delta_{\pi,\{(x,y)\}} \ket*{y}_{\gsA} \frac{1}{\sqrt{(N-t)!}}
        \delta_{\pi,R} \ket*{\pi}_{\gsP} \nonumber \\
        & \quad \otimes \frac{1}{\sqrt{2^N}} \sum_{f \in \{0,1\}^N} (-1)^{f(x_1) + \cdots + f(x_t)} \ket*{f}_{\gsF} \label{eq:plug-in-pfo-2}.
    \end{align}
    Since $\delta_{\pi,R} \cdot \delta_{\pi,\{(x,y)\}} = \delta_{\pi, R \cup\{(x,y)\}}$, we can rearrange the expression to get
    \begin{align}
        (\ref{eq:plug-in-pfo-2}) &= \frac{1}{\sqrt{N-t}} \sum_{y \in [N]} \ket*{y}_{\gsA} \frac{1}{\sqrt{(N-t-1)!}} \sum_{\pi \in \sSym_N} \delta_{\pi,R \cup\{(x,y)\}} \ket*{\pi} \nonumber \\
        & \quad \otimes \frac{1}{\sqrt{2^N}} \sum_{f \in \{0,1\}^N} (-1)^{f(x_1) + \cdots + f(x_t) + f(x)} \ket*{f}_{\gsF}\\
        &= \frac{1}{\sqrt{N-t}} \sum_{y \in [N]} \ket*{y}_{\gsA} \ket*{\pf_{R \cup\{(x,y)\}}}_{\gsP\gsF},
    \end{align}
    which completes the proof.
\end{proof}

\section{The path-recording oracle $V$}
\label{sec:compressed-haar}

In this section, we define the \emph{path-recording oracle}. The path-recording oracle $\pr$ acts on an $n$-qubit query register $\sA$ held by the adversary, as well as a variable-length relation $\sR$ containing a relation state $\ket*{R}$ (see~\cref{subsec:relation-states}). In section~\cref{subsec:relation-cho-pfo}, we connect the path-recording oracle $\pr$ to the $\pfo$ oracle. In~\cref{subsec:imp-forward-q}, we sketch how to implement $V$ efficiently. 

\subsection{Defining $\pr$}
\label{subsec:define-pr-oracle}

\begin{definition}[Path-recording oracle]
\label{def:choisometry}
The path-recording oracle $\pr$ is a linear map $\pr: \calH_{\sA} \otimes \calH_{\sR} \rightarrow \calH_{\sA} \otimes \calH_{\sR}$ defined as follows. For all $x \in [N]$ and $R \in \calR^{\inj}$ such that $\abs{R} < N$,
\begin{equation}
    \pr: \ket*{x}_{\gsA} \ket*{R}_{\gsR} \mapsto \frac{1}{\sqrt{N - \abs{R}}} \sum_{\substack{y \in [N], \\ y \not\in \Im(R)}} \ket*{y}_{\gsA} \ket*{R \cup\{(x,y)\}}_{\gsR}.
\end{equation}
Note that $R \cup\{(x,y)\} \in \calR^{\inj}$ since $y \notin \Im(R)$.
\end{definition}

\begin{lemma}[Partial isometry]
\label{lem:cho-isometry}
The path-recording oracle $\pr$ is an isometry on the subspace of $\calH_{\sA} \otimes \calH_{\sR}$ spanned by the states $\ket*{x}\ket*{R}$ for $x \in [N]$ and $R \in \calR^{\inj}$ such that $\abs{R} < N$.
\end{lemma}

\begin{proof}[Proof of Lemma~\ref{lem:cho-isometry}]
To prove that $\pr$ is an isometry on the specified subspace, it suffices to show that for all $x,x' \in [N]$ and $R,R' \in \calR^{\inj}$ with $\abs{R}, \abs{R'} < N$,
\begin{equation}
    \bra*{x'}_{\gsA} \bra*{R'}_{\gsR} \pr^\dagger \cdot \pr \ket*{x}_{\gsA} \ket*{R}_{\gsR} = \braket*{x'}{x}_{\gsA} \cdot \braket*{R'}{R}_{\gsR} \label{eq:cho-partial-isometry}.
\end{equation}
We proceed by considering two cases:

\begin{itemize}
\item \textbf{Case 1:} $\abs{R} \neq \abs{R'}$. $\pr \ket*{x}_{\gsA} \ket*{R}_{\gsR}$ and $\pr \ket*{x'}_{\gsA} \ket*{R'}_{\gsR}$ are orthogonal because, by the definition of $\pr$, these two states are supported on relation states of different sizes. Therefore, the left-hand side of~\cref{eq:cho-partial-isometry} is zero, which equals the right-hand side, since $\braket*{R'}{R}_{\gsR} = 0$ for $\abs{R} \neq \abs{R'}$.
\item \textbf{Case 2:} $\abs{R} = \abs{R'} = t$ for some $0 \leq t \leq N-1$. In this case, we expand the left-hand side:
\begin{align}
    &\bra*{x'}_{\gsA} \bra*{R'}_{\gsR} \pr^\dagger \cdot \pr \ket*{x}_{\gsA} \ket*{R}_{\gsR} \nonumber \\
    &= \Bigg( \frac{1}{\sqrt{N - t}} \sum_{\substack{y' \in [N],\\ y' \not\in \Im(R')}} \bra*{y'}_{\gsA} \bra*{R' \cup \{(x',y')\}}_{\gsR} \Bigg) \cdot \Bigg( \frac{1}{\sqrt{N - t}} \sum_{\substack{y \in [N],\\ y \not\in \Im(R)}} \ket*{y}_{\gsA} \ket*{R \cup\{(x,y)\}}_{\gsR} \Bigg) \\
    &= \frac{1}{N-t} \sum_{\substack{y \in [N] \\ y\not \in \Im(R') \cup \Im(R)}} \braket*{R' \cup \{(x',y)\}}{R \cup\{(x,y)\}}_{\gsR} \label{eq:chot-is-isometry}
\end{align}
Now, we consider two sub-cases:
\begin{itemize}
\item \textbf{Case 2a:} $(x,R) \neq (x',R')$. For $y \not\in \Im(R) \cup \Im(R')$, the term $\braket*{R' \cup \{(x',y)\}}{R \cup \{(x,y)\}}_{\gsR}$ is always zero because either $x \neq x'$ or $R \neq R'$. Therefore, Eq.~\eqref{eq:chot-is-isometry} is equal to zero, which matches the right-hand side of the original equation.
\item \textbf{Case 2b:} $(x,R) = (x',R')$. In this case, we have:
\begin{align}
    \eqref{eq:chot-is-isometry} &= \frac{1}{N-t} \sum_{y \in [N] \setminus \Im(R)} \braket*{R \cup \{(x,y)\}}_{\gsR} \\
    &= \frac{1}{N-t} \cdot (N-t) \cdot 1 = 1,
\end{align}
which again matches the right-hand side of the original equation.
\end{itemize}
\end{itemize}

This shows that~\cref{eq:cho-partial-isometry} holds in all cases, completing the proof.
\end{proof}

Next, we define the state $\ket*{\Adv^{V}_t}_{\gsA \gsB \gsR}$ to be the state of the state of the entire system after the adversary has made $t$ queries to the path recording oracle, with the $\sR$ register initialized to $\ket*{\varnothing}$, the state associated with the empty set.

\begin{definition}
Given a $t$-query adversary $\Adv$ specified by a $t$-tuple of unitaries $(A_{1, \gsA \gsB},\dots,A_{t, \gsA \gsB})$, define the state
\begin{equation}
    \ket*{\Adv^{\pr}_t}_{\gsA \gsB \gsR} \coloneqq \prod_{i = 1}^t \Big( \pr \cdot A_{i, \gsA \gsB} \Big) \ket*{0}_{\gsA \gsB} \ket*{\varnothing}_{\gsR}.
\end{equation}
\end{definition}

In fact, it will be useful to define a version of this state in which an arbitrary $n$-qubit unitary $G$ is applied to the adversary's query register $\sA$ before each query to $V$.  

\begin{definition}\label{def:chostate}
Given an $n$-qubit unitary $G$ and a $t$-query adversary $\Adv$ specified by a $t$-tuple of unitaries $(A_{1, \gsA \gsB},\dots,A_{t, \gsA \gsB})$, define the state
\begin{equation}
    \ket*{\Adv^{\pr \cdot G}_t}_{\gsA \gsB \gsR} \coloneqq \prod_{i = 1}^t \Big( \pr \cdot G_{\gsA} \cdot A_{i, \gsA \gsB} \Big) \ket*{0}_{\gsA \gsB} \ket*{\varnothing}_{\gsR}.
\end{equation}
\end{definition}

One consequence of~\cref{lem:cho-isometry} is that $\ket*{\Adv^{\pr \cdot G}_t}_{\gsA \gsB \gsR}$ has unit norm as long as $t \leq N$. 

\begin{lemma}[$\ket*{\Adv^{\pr \cdot G}_t}_{\gsA \gsB \gsR}$ has unit norm]
\label{lemma:purified-cho-state-unit-norm}
    For any adversary $\Adv$ making $t \leq N$ forward queries, and any $n$-qubit unitary $G$, $\ket*{\Adv^{\pr \cdot G}_t}_{\gsA \gsB \gsR}$ has unit norm.
\end{lemma}

\begin{proof}[Proof of~\cref{lemma:purified-cho-state-unit-norm}]
We say that a state on registers $(\sA, \sB, \sR)$ is supported on $\calR^{\inj}$ if the state is contained in the span of $\ket*{x}_{\gsA} \ket*{z}_{\gsB} \ket*{R}_{\gsR}$ for $R \in \calR^{\inj}$ and any $x,z$.
We will prove by induction on $t$ that for all $0 \leq t \leq N$, $\ket*{\Adv^{\pr \cdot G}_t}_{\gsA \gsB \gsR}$ is a unit-norm state supported on $\calR_t^{\inj}$.

\vspace{0.3em}
\noindent \textbf{Base case ($t=0$):} 
$\ket*{\Adv^{\pr \cdot G}_0} = \ket*{0}_{\gsA \gsB} \ket*{\varnothing}_{\gsR}$. This state clearly has unit norm, and $\ket*{\varnothing}_{\gsR} \in \calR_0^{\inj}$, so the claim holds for $t=0$.

\vspace{0.3em}
\noindent \textbf{Inductive step:} 
Assume the claim is true for some $0 \leq t < N$, i.e., $\ket*{\Adv^{\pr \cdot G}_t}_{\gsA \gsB \gsR}$ is a unit-norm state supported on $\calR_t^{\inj}$. We will prove that it must hold for $t+1$.
By definition, we have:
\begin{equation}
    \ket*{\Adv^{\pr \cdot G}_{t+1}} = \pr \cdot G_{\gsA} \cdot A_{t+1, \gsA \gsB} \ket*{\Adv^{\pr \cdot G}_t}
\end{equation}
This state is unit norm because:
\begin{enumerate}
    \item $G_{\gsA} \cdot A_{t+1, \gsA \gsB}$ is a unitary that acts only on the $\sA$ and $\sB$ registers, and so $G_{\gsA} \cdot A_{t+1, \gsA \gsB} \ket*{\Adv^{\pr \cdot G}_t}$ is still a unit-norm state supported on $\calR_t^{\inj}$.
    \item By~\cref{lem:cho-isometry}, $\pr$ is an isometry on states supported on $\calR_t^{\inj}$. Moreover, the definition of $\pr$, ensures that it maps states supported on $\calR_t^{\inj}$ to states supported on $\calR_{t+1}^{\inj}$ for $0 \leq t < N$. Thus, $\ket*{\Psi_{t+1}^G}$ is a unit-norm state supported on $\calR_{t+1}^{\inj}$.
\end{enumerate}
Hence, for all $0 \leq t \leq N$, $\ket*{\Adv^{\pr \cdot G}_t}_{\gsA \gsB \gsR}$ is a unit-norm state supported on $\calR_t^{\inj}$.
\end{proof}

\subsection{Right unitary invariance}
\label{subsec:cuo-right-unitary-inv}

Our next step is to prove that $\pr$ satisfies \emph{right unitary invariance}: for any unitary $G$, queries to $\pr \cdot G_{\gsA}$ are perfectly indistinguishable from queries to $V$, from the point of view of the adversary who cannot access the purifying register $\sR$. This is captured by the following lemma.

\begin{lemma}[Right unitary invariance]\label{lem:cho-transfer}
For any $n$-qubit unitary $G$, we have
\begin{equation}
    \ket*{\Adv^{\pr \cdot G}_t}_{\gsA \gsB \gsR} = (G_{\darkgray{\sR_{\sX, 1}^{(t)}}} \otimes \ldots \otimes G_{\darkgray{\sR_{\sX, t}^{(t)}}}) \ket*{\Adv^{\pr}_t}_{\gsA \gsB \gsR}.
\end{equation}
\end{lemma}
Note that
\begin{align}
    &\Tr_{\sR}(\ketbra*{\Adv^{\pr \cdot G}_t}_{\gsA \gsB \gsR}) \nonumber \\
    &= \Tr_{\sR}((G_{\darkgray{\sR_{\sX, 1}^{(t)}}} \otimes \ldots \otimes G_{\darkgray{\sR_{\sX, t}^{(t)}}}) \ketbra*{\Adv^{\pr}_t}_{\gsA \gsB \gsR} (G_{\darkgray{\sR_{\sX, 1}^{(t)}}} \otimes \ldots \otimes G_{\darkgray{\sR_{\sX, t}^{(t)}}})^\dagger) \tag{by~\cref{lem:cho-transfer}}\\
    &= \Tr_{\sR}(\ketbra*{\Adv^{\pr}_t}_{\gsA \gsB \gsR}), \tag{by the cyclic property of $\Tr_{\sR}$}
\end{align}
where the first line corresponds to the adversary's view after making $t$ queries to $\pr \cdot G_{\gsA}$, and the last line corresponds to its view after making $t$ queries to $\pr$.

\begin{fact}[Explicit form] \label{fact:expli-form-cho-state}
From the definition of $\pr$ and $\ket*{R}_{\gsR}$, we can expand out $\ket*{\Adv^{\pr \cdot G}_t}_{\gsA \gsB \gsR}$ to obtain
\begin{align}
\ket*{\Adv^{\pr \cdot G}_t}_{\gsA \gsB \gsR} &= \sqrt{\frac{(N-t)!}{N!}} \sum_{\substack{(x_1, \ldots, x_t) \in [N]^t \\ (y_1, \ldots, y_t) \in [N]_{\dist}^t}}
\left[ \, \prod_{i = 1}^t \Big( \ketbra*{y_i}{x_i}_{\gsA} \cdot G_{\gsA} \cdot A_{i, \gsA \gsB} \Big) \ket*{0}_{\gsA \gsB} \, \right] \otimes \ket*{ \{ (x_i, y_i) \}_{i=1}^t }_{\gsR}\\
&= \sqrt{\frac{(N-t)!}{N!}} \sum_{\substack{(x_1, \ldots, x_t) \in [N]^t \\ (y_1, \ldots, y_t) \in [N]_{\dist}^t}}
\left[ \, \prod_{i = 1}^t \Big( \ketbra*{y_i}{x_i}_{\gsA} \cdot G_{\gsA} \cdot A_{i, \gsA \gsB} \Big) \ket*{0}_{\gsA \gsB} \, \right] \nonumber \\
&\otimes \frac{1}{\sqrt{t!}} \sum_{\pi \in \sSym_t} \left( S_{\pi} \ket*{x_1}_{\darkgray{\sR_{\sX, 1}^{(t)}}} \dots \ket*{x_t}_{\darkgray{\sR_{\sX, t}^{(t)}}} \right) \otimes \left( S_{\pi} \ket*{y_1}_{\darkgray{\sR_{\sY, 1}^{(t)}}} \dots \ket*{y_t}_{\darkgray{\sR_{\sY, t}^{(t)}}} \right),
\end{align}
\end{fact}

\begin{proof}[Proof of~\cref{lem:cho-transfer}]
Our proof will use the following trivial identities for registers $\sA$ and $(\sR_{\sX, i}^{(t)})_{i \in [N]}$:
\begin{align}
 \sum_{z \in [N]} \ketbra*{z}_{\gsA}
 &= \Id_{\gsA}, \label{eq:zz-ID-A} \\
 \sum_{z \in [N]} \ketbra*{z}_{\darkgray{\sR_{\sX, i}^{(t)}}} &= \Id_{\darkgray{\sR_{\sX, i}^{(t)}}}. \label{eq:zz-ID-Xi}
\end{align}
For any $n$-qubit unitary $G$ and $x, z \in [N]$, we have
\begin{equation} \label{eq:xGz-ID}
\bra*{x}_{\gsA} G_{\gsA} \ket*{z}_{\gsA}
= \bra*{x}_{\darkgray{\sR_{\sX, i}^{(t)}}} G_{\darkgray{\sR_{\sX, i}^{(t)}}} \ket*{z}_{\darkgray{\sR_{\sX, i}^{(t)}}}.  
\end{equation}
Therefore, we have
\begin{align}
    \sum_{x \in [N]} \ket*{x}_{\darkgray{\sR_{\sX, i}^{(t)}}} \otimes \bra*{x}_{\gsA} G_{\gsA} &= \sum_{x, z \in [N]} \ket*{x}_{\darkgray{\sR_{\sX, i}^{(t)}}} \otimes \left(\bra*{x}_{\gsA} G_{\gsA} \ket*{z}_{\gsA}\right) \bra*{z}_{\gsA} \tag{Using Eq.~\eqref{eq:zz-ID-A}} \\
    &= \sum_{x, z \in [N]} \ket*{x}_{\darkgray{\sR_{\sX, i}^{(t)}}} \otimes \left(\bra*{x}_{\darkgray{\sR_{\sX, i}^{(t)}}} G_{\darkgray{\sR_{\sX, i}^{(t)}}} \ket*{z}_{\darkgray{\sR_{\sX, i}^{(t)}}}\right) \bra*{z}_{\gsA} \tag{Using Eq.~\eqref{eq:xGz-ID}} \\
    &= \sum_{x, z \in [N]} \left( \ketbra*{x}_{\darkgray{\sR_{\sX, i}^{(t)}}} G_{\darkgray{\sR_{\sX, i}^{(t)}}} \ket*{z}_{\darkgray{\sR_{\sX, i}^{(t)}}}\right) \otimes \bra*{z}_{\gsA} \\
    &= \sum_{z \in [N]} G_{\darkgray{\sR_{\sX, i}^{(t)}}} \ket*{z}_{\darkgray{\sR_{\sX, i}^{(t)}}} \otimes \bra*{z}_{\gsA} \tag{Using Eq.~\eqref{eq:zz-ID-Xi}} \\
    &= \sum_{x \in [N]} G_{\darkgray{\sR_{\sX, i}^{(t)}}} \ket*{x}_{\darkgray{\sR_{\sX, i}^{(t)}}} \otimes \bra*{x}_{\gsA}. \tag{Relabeling $z$ with $x$}
\end{align}
Applying the above identity to registers $\darkgray{\sR_{\sX, 1}^{(t)}}, \ldots, \darkgray{\sR_{\sX, t}^{(t)}}$ to \cref{fact:expli-form-cho-state} yields
\begin{align}
\ket*{\Adv^{\pr \cdot G}_t} &= \sqrt{\frac{(N-t)!}{N!}} \sum_{\substack{(x_1, \ldots, x_t) \in [N]^t \\ (y_1, \ldots, y_t) \in [N]_{\dist}^t}}
\left[ \, \prod_{i = 1}^t \Big( \ket*{y_i}_{\gsA} \otimes \bra*{x_i}_{\gsA} G_{\gsA} \cdot A_{i, \gsA \gsB} \Big) \ket*{0}_{\gsA \gsB} \, \right] \otimes\\
&\frac{1}{\sqrt{t!}} \sum_{\pi \in \sSym_t} \left( S_{\pi} \ket*{x_1}_{\darkgray{\sR_{\sX, 1}^{(t)}}} \dots \ket*{x_t}_{\darkgray{\sR_{\sX, t}^{(t)}}} \right) \otimes \left( S_{\pi} \ket*{y_1}_{\darkgray{\sR_{\sY, 1}^{(t)}}} \dots \ket*{y_t}_{\darkgray{\sR_{\sY, t}^{(t)}}} \right) \\
&= \sqrt{\frac{(N-t)!}{N!}} \sum_{\substack{(x_1, \ldots, x_t) \in [N]^t \\ (y_1, \ldots, y_t) \in [N]_{\dist}^t}}
\left[ \, \prod_{i = 1}^t \Big( \ketbra*{y_i}{x_i}_{\gsA} \cdot A_{i, \gsA \gsB} \Big) \ket*{0}_{\gsA \gsB} \, \right] \otimes\\
&\frac{1}{\sqrt{t!}} \sum_{\pi \in \sSym_t} \left( S_{\pi} \, G_{\darkgray{\sR_{\sX, 1}^{(t)}}} \ket*{x_1}_{\darkgray{\sR_{\sX, 1}^{(t)}}} \dots 
G_{\darkgray{\sR_{\sX, t}^{(t)}}} \ket*{x_t}_{\darkgray{\sR_{\sX, t}^{(t)}}} \right) \otimes \left( S_{\pi} \ket*{y_1}_{\darkgray{\sR_{\sY, 1}^{(t)}}} \dots \ket*{y_t}_{\darkgray{\sR_{\sY, t}^{(t)}}} \right) \\
&= (G_{\darkgray{\sR_{\sX, 1}^{(t)}}} \otimes \ldots \otimes G_{\darkgray{\sR_{\sX, t}^{(t)}}}) \ket*{\Adv^{\pr}_t}
\end{align}
The last line follows from the fact that $(G_{\darkgray{\sR_{\sX, 1}^{(t)}}} \otimes \ldots \otimes G_{\darkgray{\sR_{\sX, t}^{(t)}}})$ acts identically on all $t$ registers, so
\begin{equation}
    S_{\pi} \cdot (G_{\darkgray{\sR_{\sX, 1}^{(t)}}} \otimes \ldots \otimes G_{\darkgray{\sR_{\sX, t}^{(t)}}}) = (G_{\darkgray{\sR_{\sX, 1}^{(t)}}} \otimes \ldots \otimes G_{\darkgray{\sR_{\sX, t}^{(t)}}}) \cdot S_{\pi}.
\end{equation}
This concludes the proof.
\end{proof}

\begin{corollary}[Trace distance between original state and the projected state] \label{corollary:distinct-X}
    \begin{equation}
\norm{\Tr_{\sR}\left(\Pi^{\dist}_{\gsR_{\darkgray{X}}^{\darkgray{(t)}}} \cdot \E_{ C \gets \frakD} \ketbra*{\Adv^{\pr \cdot C}_t}_{\gsA \gsB \gsR} \cdot \Pi^{\dist}_{\gsR_{\darkgray{X}}^{\darkgray{(t)}}}  \right) - \Tr_{\sR}\left(\E_{ C \gets \frakD} \ketbra*{\Adv^{\pr \cdot C}_t}_{\gsA \gsB \gsR} \right)}_1 \leq \frac{t(t-1)}{N+1}.
    \end{equation}
\end{corollary}
\begin{proof} The trace distance can be bounded as follows,
\begin{align}
    &\norm{\Tr_{\sR}\left(\Pi^{\dist}_{\gsR_{\darkgray{X}}^{\darkgray{(t)}}} \cdot \E_{ C \gets \frakD} \ketbra*{\Adv^{\pr \cdot C}_t}_{\gsA \gsB \gsR} \cdot \Pi^{\dist}_{\gsR_{\darkgray{X}}^{\darkgray{(t)}}}  \right) - \Tr_{\sR}\left(\E_{ C \gets \frakD} \ketbra*{\Adv^{\pr \cdot C}_t}_{\gsA \gsB \gsR} \right)}_1\\
    &= 1 - \Tr\left(\E_{ C \gets \frakD} \Pi^{\dist}_{\gsR_{\darkgray{X}}^{\darkgray{(t)}}} \cdot \ketbra*{\Adv^{\pr \cdot C}_t}_{\gsA \gsB \gsR} \cdot \Pi^{\dist}_{\gsR_{\darkgray{X}}^{\darkgray{(t)}}}\right) \tag{\cref{lemma:gentle}}\\
    &= 1 - \Tr\left( \E_{ C \gets \frakD} \Pi^{\dist}_{\darkgray{\sR_{\sX}^{(t)}}} \cdot C^{\otimes t}_{\darkgray{\sR^{(t)}_{\sX}}} \cdot \ketbra*{\Adv^{\pr}_t}_{\gsA \gsB \gsP \gsF} \cdot C^{\otimes t}_{\darkgray{\sR^{(t)}_{\sX}}} \cdot \Pi^{\dist}_{\gsR_{\darkgray{X}}^{\darkgray{(t)}}} \right) \tag{By~\cref{lem:cho-transfer}}\\
    &\leq \frac{t(t-1)}{N+1}, \tag{By~\cref{lem:almost-distinct-on-X}}
\end{align}
which completes the proof of this corollary.
\end{proof}

\subsection{Relating $V$ to $\pfo$} 
\label{subsec:relation-cho-pfo}

We now connect the path-recording oracle $\pr$ to the $\pfo$ oracle defined previously. We begin by defining the $\pfo$ analog of $\ket*{\Adv^{\pr \cdot G}_t}_{\gsA \gsB \gsR}$.

\begin{definition}
    Given an $n$-qubit unitary $G$ and a $t$-query adversary $\Adv$ specified by a $t$-tuple of unitaries $(A_{1, \gsA \gsB},\dots,A_{t, \gsA \gsB})$, define
    \begin{align}
        \ket*{\Adv^{\pfo \cdot G}_t}_{\gsA \gsB \gsP \gsF} \coloneqq \prod_{i = 1}^t \Big( \pfo \cdot G_{\gsA} \cdot A_{i,\gsA \gsB} \Big) \ket*{0}_{\gsA \gsB} \ket*{\pf_{\varnothing}}_{\gsP \gsF}.
    \end{align}
    Recall that
    \begin{align}
    \ket*{\pf_{\varnothing}}_{\gsP \gsF} \coloneqq \frac{1}{\sqrt{N!}} \sum_{\pi \in \sSym_N} \ket*{\pi}_{\gsP} \otimes \frac{1}{\sqrt{2^N}} \sum_{f \in \{0,1\}^N} \ket*{f}_{\gsF}.
\end{align}
\end{definition}

We can expand the definition of $\ket*{\Adv^{\pfo \cdot G}_t}_{\gsA \gsB \gsP \gsF}$ to obtain the following.

\begin{fact}[Explicit form of $\ket*{\Adv^{\pfo \cdot G}_t}_{\gsA \gsB \gsP \gsF}$]
\label{fact:expli-form-pfo-state}
    \begin{align}
        \ket*{\Adv^{\pfo \cdot G}_t}_{\gsA \gsB \gsP \gsF} &= \sqrt{\frac{(N-t)!}{N!}} \sum_{\substack{(x_1, \ldots, x_t) \in [N]^t \\ (y_1, \ldots, y_t) \in [N]^t}}
\left[ \, \prod_{i = 1}^t \Big( \ketbra*{y_i}{x_i}_{\gsA} \cdot G_{\gsA} \cdot A_{i,\gsA \gsB} \Big) \ket*{0}_{\gsA \gsB} \, \right]  \otimes \ket*{\pf_{\{(x_i,y_i)\}_{i = 1}^t}}_{\gsP \gsF}.
    \end{align}
\end{fact}

While the state $\ket*{\pf_{\{(x_i,y_i)\}_{i = 1}^t}}$ is supported on an exponential number of qubits, we can compress the environment using the following linear operator $\Compress$.
By~\cref{claim:phi-S-orthogonal}, $\Compress$ is a partial isometry. Intuitively, $\Compress$ ``compresses'' the state $\ket*{\pf_R}$, which requires an exponential number of qubits $n$, to $\ket*{R}$, which is only as big as the size of the relation. 

\begin{definition}
    Define $\Compress: \calH_{\sP} \otimes \calH_{\sF} \rightarrow \calH_{\sR}$ to be
    \begin{align}
        \Compress \coloneqq \sum_{R \in\calR^{\bij}} \ketbra*{R}{\pf_R}
    \end{align}
\end{definition}

Next, we will use $\Compress$ to relate the path-recording oracle $\pr$ to the purified permutation-function oracle. To do so, we will need to define the following projectors.

\begin{definition}[Distinct subspace projector]
Given $0 \leq t \leq N$. Let
\begin{align}
    \Pi^{\dist}_{\darkgray{\sR_{\sX}^{(t)}}} &\coloneq \sum_{(x_1, \ldots, x_t) \in [N]_{\dist}^t} \ketbra*{x_1}_{\darkgray{\sR_{\sX, 1}^{(t)}}} \otimes \ldots \otimes \ketbra*{x_t}_{\darkgray{\sR_{\sX, t}^{(t)}}}.
\end{align}
\end{definition}

\begin{definition}[Distinct subspace projector for $\pf$-relation states]
    Let
    \begin{align}
    \widetilde{\Pi}^{\dist}_{\gsP \gsF} \coloneqq \sum_{\substack{R \in\calR^{\bij},\\ \abs{R} = t}} \ketbra*{\pf_R}.
    \end{align}
\end{definition}

\begin{lemma}[Relating $\pr$ and $\pfo$ states]
\label{lemma:compress-cho-pfo-state}
    For all $n$-qubit unitaries $G$,
    \begin{align}
        \Compress \cdot \widetilde{\Pi}^{\dist}_{\gsP \gsF} \cdot \ket*{\Adv^{\pfo \cdot G}_t}_{\gsA \gsB \gsP \gsF} = \Pi^{\dist}_{\darkgray{\sR_{\sX}^{(t)}}} \cdot \ket*{\Adv^{\pr \cdot G}_t}
    \end{align}
\end{lemma}
\begin{proof}
    By~\cref{fact:expli-form-cho-state}, we have
    \begin{align}
        \ket*{\Adv^{\pr \cdot G}_t} &= \sqrt{\frac{(N-t)!}{N!}} \sum_{\substack{(x_1, \ldots, x_t) \in [N]^t \\ (y_1, \ldots, y_t) \in [N]_{\dist}^t}}
\left[ \, \prod_{i = 1}^t \Big( \ketbra*{y_i}{x_i}_{\gsA} \cdot G_{\gsA} \cdot A_{i, \gsA \gsB} \Big) \ket*{0}_{\gsA \gsB} \, \right] \otimes \ket*{\{(x_i,y_i)\}_{i = 1}^t}_{\gsR}.
    \end{align}
    Applying $\Pi^{\dist}_{\darkgray{\sR_{\sX}^{(t)}}}$ to this state selects the terms corresponding to $(x_1,\dots,x_t) \in [N]^t_{\dist}$: 
    \begin{align}
        \Pi^{\dist}_{\darkgray{\sR_{\sX}^{(t)}}} \cdot \ket*{\Adv^{\pr \cdot G}_t} &= \sqrt{\frac{(N-t)!}{N!}} \sum_{\substack{(x_1, \ldots, x_t) \in [N]_{\dist}^t \\ (y_1, \ldots, y_t) \in [N]_{\dist}^t}}
\left[ \, \prod_{i = 1}^t \Big( \ketbra*{y_i}{x_i}_{\gsA} \cdot G_{\gsA} \cdot A_{i, \gsA \gsB} \Big) \ket*{0}_{\gsA \gsB} \, \right] \otimes \ket*{\{(x_i,y_i)\}_{i = 1}^t}_{\gsR}. \label{eq:pi-psi-t}
    \end{align}
    By~\cref{fact:expli-form-pfo-state},
    \begin{align}
        \ket*{\Adv^{\pfo \cdot G}_t}_{\gsA \gsB \gsP \gsF} &= \sqrt{\frac{(N-t)!}{N!}} \sum_{\substack{(x_1, \ldots, x_t) \in [N]^t \\ (y_1, \ldots, y_t) \in [N]^t}}
\left[ \, \prod_{i = 1}^t \Big( \ketbra*{y_i}{x_i}_{\gsA} \cdot G_{\gsA} \cdot A_{i, \gsA \gsB} \Big) \ket*{0}_{\gsA \gsB} \, \right]  \otimes \ket*{\pf_{\{(x_i,y_i)\}_{i = 1}^t}}_{\gsP \gsF}.
    \end{align}
    Applying $\widetilde{\Pi}^{\dist}_{\gsP \gsF}$ selects the terms corresponding to $(x_1,\dots,x_t) \in [N]^t_{\dist}$ and $(y_1,\dots,y_t) \in [N]^t_{\dist}$:
    \begin{align}
    &\widetilde{\Pi}^{\dist}_{\gsP \gsF} \cdot \ket*{\Adv^{\pfo \cdot G}_t}_{\gsA \gsB \gsP \gsF}\\
    &= \sqrt{\frac{(N-t)!}{N!}} \sum_{\substack{(x_1, \ldots, x_t) \in [N]_{\dist}^t \\ (y_1, \ldots, y_t) \in [N]_{\dist}^t}}
    \left[ \, \prod_{i = 1}^t \Big( \ketbra*{y_i}{x_i}_{\gsA} \cdot G_{\gsA} \cdot A_{i, \gsA \gsB} \Big) \ket*{0}_{\gsA \gsB} \, \right]  \otimes \ket*{\pf_{\{(x_i,y_i)\}_{i = 1}^t}}_{\gsP \gsF}. \label{eq:tilde-pi-phi-t}
    \end{align}
    Since $\Compress$ maps $\ket*{\pf_R}$ to $\ket*{R}$ for all $R \in \calR^{\bij}$, applying $\Compress$ to the right-hand side of~\cref{eq:tilde-pi-phi-t} yields the right-hand side of~\cref{eq:pi-psi-t}, which proves the claim. 
    \end{proof}

\begin{corollary}[Trace distance between original state and the projected state]
\label{corollary:phi-projector-small-dist}
\begin{align}
    \norm{\Tr_{\sP \sF}\left(\E_{ C \gets \frakD} \ketbra*{\Adv^{\pfo \cdot C}_t}_{\gsA \gsB \gsP \gsF}\right) - \Tr_{\sP \sF}\left(\widetilde{\Pi}^{\dist}_{\gsP \gsF} \cdot \E_{ C \gets \frakD} \ketbra*{\Adv^{\pfo \cdot C}_t}_{\gsA \gsB \gsP \gsF} \cdot \widetilde{\Pi}^{\dist}_{\gsP \gsF} \right)}_1 \leq \frac{t(t-1)}{N+1}.
\end{align}
\end{corollary}

\begin{proof}
    By~\cref{lemma:gentle}, we have
    \begin{align}
        &\norm{\Tr_{\sP \sF}\left(\E_{ C \gets \frakD} \ketbra*{\Adv^{\pfo \cdot C}_t}_{\gsA \gsB \gsP \gsF}\right) - \Tr_{\sP \sF}\left(\widetilde{\Pi}^{\dist}_{\gsP \gsF} \cdot \E_{ C \gets \frakD} \ketbra*{\Adv^{\pfo \cdot C}_t}_{\gsA \gsB \gsP \gsF} \cdot \widetilde{\Pi}^{\dist}_{\gsP \gsF} \right)}_1 \nonumber \\
        &= 1 - \Tr\left(\widetilde{\Pi}^{\dist}_{\gsP \gsF} \cdot \E_{ C \gets \frakD} \ketbra*{\Adv^{\pfo \cdot C}_t}_{\gsA \gsB \gsP \gsF} \cdot \widetilde{\Pi}^{\dist}_{\gsP \gsF} \right) \label{eq:pre-compress}
    \end{align}
    Next, observe that $\widetilde{\Pi}^{\dist}_{\gsP \gsF} = \Compress^\dagger \cdot \Compress \cdot \widetilde{\Pi}^{\dist}_{\gsP \gsF}$ since
    \begin{align}
        &\Compress^\dagger \cdot \Compress \cdot \widetilde{\Pi}^{\dist}_{\gsP \gsF}\\
        &= \Big( \sum_{R \in\calR^{\bij}} \ketbra*{\pf_R}{R} \Big) \cdot \Big( \sum_{R \in\calR^{\bij}} \ketbra*{R}{\pf_R} \Big) \cdot \Big( \sum_{\substack{R \in\calR^{\bij},\\ \abs{R} = t}} \ketbra*{\pf_R} \Big) \\
        &= \sum_{\substack{R \in\calR^{\bij},\\ \abs{R} = t}} \ketbra*{\pf_R} = \widetilde{\Pi}^{\dist}_{\gsP \gsF}.
    \end{align}
    By plugging this identity into (\ref{eq:pre-compress}), we get
    \begin{align}
        (\ref{eq:pre-compress}) &= 1 - \Tr\left(\Compress^\dagger \cdot \Compress \cdot \widetilde{\Pi}^{\dist}_{\gsP \gsF} \cdot \E_{ C \gets \frakD} \ketbra*{\Adv^{\pfo \cdot C}_t}_{\gsA \gsB \gsP \gsF} \cdot \widetilde{\Pi}^{\dist}_{\gsP \gsF} \right) \\
        &= 1 - \Tr\left(\Compress \cdot \widetilde{\Pi}^{\dist}_{\gsP \gsF} \cdot \E_{ C \gets \frakD} \ketbra*{\Adv^{\pfo \cdot C}_t}_{\gsA \gsB \gsP \gsF} \cdot \widetilde{\Pi}^{\dist}_{\gsP \gsF} \cdot \Compress^\dagger \right) \\
        &= 1 - \Tr\left( \E_{ C \gets \frakD} \Pi^{\dist}_{\darkgray{\sR_{\sX}^{(t)}}} \cdot \ketbra*{\Adv^{\pr \cdot C}_t}_{\gsA \gsB \gsR} \cdot \Pi^{\dist}_{\gsR_{\darkgray{X}}^{\darkgray{(t)}}} \right) \tag{By~\cref{lemma:compress-cho-pfo-state}}\\
        &\leq \frac{t(t-1)}{N+1} \tag{By~\cref{corollary:distinct-X}},
    \end{align}
    which completes the proof.
\end{proof}

\section{The PRU proof}
\label{sec:hybrids-full-list}

\subsection{Setup}

We define a distribution over $n$-qubit unitaries parameterized by any $n$-qubit unitary $2$-design $\frakD$.

\begin{definition}[$\mathsf{PRU}(\frakD)$ distribution]
    Let $\frakD$ be a distribution supported on $\calU(N)$. The distribution $\mathsf{PF}({\frakD})$ is defined as follows:
    \begin{enumerate}
        \item Sample a uniformly random permutation $\pi \gets \sSym_{N}$, a uniformly random $f \gets \{0,1\}^N$, and a uniformly random $n$-qubit unitary $C \gets \frakD$.
        \item Output the unitary $\calO \coloneqq P_\pi \cdot F_f \cdot C$.
    \end{enumerate} 
\end{definition}

The goal of this section is to prove the following theorem.

\begin{theorem}[$\mathsf{PF}(\frakD)$ is indistinguishable from Haar-random]\label{thm:statistical-PRU} Let $\Adv$ be a $t$-query oracle adversary that only makes forward queries, and let $\frakD$ be an exact unitary $2$-design. Then
\begin{align}
    \TD\left(\E_{\calO \gets \mathsf{PF}(\frakD)} \ketbra*{\Adv_t^{\calO}}, \E_{\calO \gets \mu_{\mathsf{Haar}}} \ketbra*{\Adv_t^{\calO}} \right) \leq \frac{4t(t-1)}{N+1}
\end{align}
\end{theorem}

Since quantum-secure pseudorandom permutations and pseudorandom functions exist assuming one-way functions~\cite{zhandry2016note,zhandry2021construct}, the existence of computationally-secure PRU follows immediately from~\cref{thm:statistical-PRU}.

\begin{theorem}
    If quantum-secure one-way functions exist, then pseudorandom unitaries exist.
\end{theorem}

The main technical component of the proof of~\cref{thm:statistical-PRU} is the following lemma.

\begin{lemma}[$\pru(\frakD)$ is indistinguishable from $\pr$]\label{lemma:pfd-cho}
     Let $\Adv$ be a $t$-query oracle adversary and let $\frakD$ be an exact unitary $2$-design. Then
     \begin{align}
        \TD\left(\E_{\calO \gets \mathsf{PF}(\frakD)} \ketbra*{\Adv_t^{\calO}}, \,\,\, \Tr_{\sR}\left( \ketbra*{\Adv^{\pr}_t}_{\gsA \gsB \gsR} \right) \right) \leq \frac{2t(t-1)}{N+1} \label{eq:intermediate-step-main-thm}
    \end{align}
\end{lemma}

\paragraph{\cref{lemma:pfd-cho} implies~\cref{thm:statistical-PRU}.} \cref{lemma:pfd-cho} implies~\cref{thm:statistical-PRU} by the following argument. We can instantiate $\frakD = \mu_{\mathsf{Haar}}$, i.e., $\frakD$ outputs a Haar-random $n$-qubit unitary. Then the output of $\pru(\mathfrak{D}) = \pru(\mu_{\mathsf{Haar}})$ is $ P_{\pi} \cdot F_f \cdot C$ for random $\pi,f$ and Haar-random $C$. By invariance of the Haar measure, this is exactly the same as outputting a Haar-random unitary. Thus, we have the following corollary of~\cref{lemma:pfd-cho}.

\begin{theorem}[$\pr$ is indistinguishable from Haar random]\label{theorem:haar-cho}
     Let $\Adv$ be a $t$-query oracle adversary. Then
     \begin{align}
        \TD\left(\E_{\calO \sim \mu_{\mathsf{Haar}}} \ketbra*{\Adv_t^{\calO}}, \Tr_{\sR}\left( \ketbra*{\Adv^{\pr}_t}_{\gsA \gsB \gsR} \right) \right) \leq \frac{2t(t-1)}{N+1}
    \end{align}
\end{theorem}

\cref{thm:statistical-PRU} follows from combining~\cref{lemma:pfd-cho,theorem:haar-cho} using the triangle inequality. It remains to prove~\cref{lemma:pfd-cho}.

\subsection{Proof of~\cref{lemma:pfd-cho}}

\begin{proof}[Proof of~\cref{lemma:pfd-cho}]
    We will use a hybrid argument. Define the mixed states
\begin{align}
    \rho^{(\frakD)}_0 &\coloneqq \E_{\calO \gets \mathsf{PF}(\frakD)} \ketbra*{\Adv_t^{\calO}}\\
    \rho^{(\frakD)}_1 &\coloneqq \Tr_{\sP \sF}\left(\E_{ C \gets \frakD} \ketbra*{\Adv^{\pfo \cdot C}_t}_{\gsA \gsB \gsP \gsF}\right)\\
    \rho^{(\frakD)}_2 &\coloneqq \Tr_{\sP \sF}\left(\widetilde{\Pi}^{\dist}_{\gsP \gsF} \cdot \E_{ C \gets \frakD} \ketbra*{\Adv^{\pfo \cdot C}_t}_{\gsA \gsB \gsP \gsF} \cdot \widetilde{\Pi}^{\dist}_{\gsP \gsF} \right)\\
    \rho^{(\frakD)}_3 &\coloneqq \Tr_{\sR}\left(\Pi^{\dist}_{\gsR_{\darkgray{X}}^{\darkgray{(t)}}} \cdot \E_{ C \gets \frakD} \ketbra*{\Adv^{\pr \cdot C}_t}_{\gsA \gsB \gsR} \cdot \Pi^{\dist}_{\gsR_{\darkgray{X}}^{\darkgray{(t)}}}  \right)\\
    \rho^{(\frakD)}_4 &\coloneqq \Tr_{\sR}\left(\E_{ C \gets \frakD} \ketbra*{\Adv^{\pr \cdot C}_t}_{\gsA \gsB \gsR} \right)\\
    \rho_5 &\coloneqq \Tr_{\sR}\left( \ketbra*{\Adv^{\pr}_t}_{\gsA \gsB \gsR} \right).
\end{align}
We argue indistinguishability between each consecutive pair of mixed states:
\begin{itemize}
    \item $\rho^{(\frakD)}_0 = \rho^{(\frakD)}_1$ by~\cref{claim:purified-vs-standard-PFO}.
    \item $\norm*{\rho^{(\frakD)}_1 - \rho^{(\frakD)}_2}_1 \leq t(t-1)/(N+1)$ by~\cref{corollary:phi-projector-small-dist}.
    \item $\rho^{(\frakD)}_2 = \rho^{(\frakD)}_3$, since by~\cref{lemma:compress-cho-pfo-state}, these are two mixed states whose purifications are related by the $\Compress$ isometry, which only acts on the purifying register. 
    \item $\norm*{\rho^{(\frakD)}_3 - \rho^{(\frakD)}_4}_1 \leq t(t-1)/(N+1)$ by~\cref{corollary:distinct-X}.
    \item $\rho^{(\frakD)}_4 = \rho^{(\frakD)}_5$ since
    \begin{align}
        \rho^{(\frakD)}_4 &= \E_{ C \gets \frakD} \Tr_{\sR}\left( \ketbra*{\Adv^{\pr \cdot C}_t}_{\gsA \gsB \gsR} \right) = \E_{ C \gets \frakD} \Tr_{\sR}\left( \ketbra*{\Adv^{\pr}_t}_{\gsA \gsB \gsR} \right) = \rho_5,
    \end{align}
    where the second equality follows from~\cref{lem:cho-transfer}, which states that for any $C$, $\ketbra*{\Adv^{\pr \cdot C}_t}$ and $\ketbra*{\Adv^{\pr}_t}$ are related by a unitary on the purifying register.
\end{itemize}
Using the triangle inequality, we obtain~\cref{eq:intermediate-step-main-thm}, which completes the proof.
\end{proof}

\newpage
\part{Strong PRUs}
\label{part:strong}

The goal of~\cref{part:strong} is to construct strong PRUs, which are secure against adversaries that make both forward and inverse queries to the unitary oracle. It is important to note that several operators that were defined in~\cref{part:standard}, including $\pfo, \Compress$ and $V$, will be have new definitions in~\cref{part:strong}.

\section{The purified permutation-function oracle}
\label{sec:PF3-oracle}

In this section, we analyze the view of an adversary that makes queries to an oracle $P_{\pi} \cdot F_{f}$, for uniformly random $\pi \gets \sSym_N$ and a random \textbf{ternary} function $f \gets \{0,1,2\}^N$. We will do this by analyzing the \emph{purified permutation-function permutation} oracle, which uses a purification of $\pi$ and $f$.

\begin{definition}[Purified permutation-function oracle] 

    The purified permutation-function oracle $\spfo$ is a unitary acting on registers $\sA,\sP,\sF$, where
    \begin{itemize}
        \item $\sP$ is a register associated with the Hilbert space $\calH_{\sP}$,  defined to be the span of the orthonormal states $\ket*{\pi}$ for all $\pi \in \sSym_N$. 
        \item $\sF$ is a register associated with the Hilbert space $\calH_{\sF}$, defined to be the span of the orthonormal states $\ket*{f}$ for all $f \in \{0,1,2\}^N$. 
    \end{itemize}
    The unitary $\spfo$ is defined to act as follows:
    \begin{align}
        \spfo_{\gsA\gsP\gsF} \ket*{x}_{\gsA} \ket*{\pi}_{\gsP} \ket*{f}_{\gsF} & \coloneqq \omega_3^{f(x)} \ket*{\pi(x)}_{\gsA} \ket*{\pi}_{\gsP} \ket*{f}_{\gsF}, \label{eq:pfo-definition-strong}\\
        &= \sum_{y \in [N]} \ket*{y}_{\gsA} \delta_{\pi(x) = y} \ket*{\pi} \omega_3^{f(x)} \ket*{f},
    \end{align}
    for all $x \in [N], \pi \in \sSym_N,$ and $f \in \{0, 1, 2\}^N$.
    Here, $\omega_3 = \exp( 2 \pi i / 3)$.
\end{definition}

The action of $\spfo^\dagger$ is
\begin{align}
    \spfo^\dagger \ket*{y}_{\gsA} \ket*{\pi} \ket*{f} &=  \sum_{x \in [N]} \ket*{x}_{\gsA} \delta_{\pi(x) = y} \ket*{\pi} \omega_3^{-f(x)} \ket*{f}.
\end{align}

The view of an adversary that queries the purified oracle is equivalent to the view of an adversary that queries the standard oracle $P_\pi \cdot F_f$, for uniformly random $\pi \gets \sSym_N$ and $f \gets \{0,1,2\}^N$. 

\begin{claim}[Equivalence of purified and standard oracles] \label{claim:purified-vs-standard-ternary-PFO}
    For any oracle adversary $\Adv$, the following oracle instantiations are perfectly indistinguishable:
    \begin{itemize}
        \item (Queries to a random $P_\pi \cdot F_f$) Sample a uniformly random $\pi \gets \sSym_N, f \gets \{0,1, 2\}^N$. On each query, apply $P_\pi \cdot F_f$ to register $\sA$.
        \item (Queries to $\spfo$) Initialize registers $\sP,\sF$ to $\frac{1}{\sqrt{N!}} \sum_{\pi \in \sSym_N} \ket*{\pi}_{\gsP} \otimes \frac{1}{\sqrt{2^N}} \sum_{f \in \{0,1, 2\}^N} \ket*{f}_{\gsF}$. At each query, apply $\spfo$ to registers $\sA,\sP,\sF$.
    \end{itemize}
\end{claim}

\noindent The proof is the same as the proof of~\cref{claim:purified-vs-standard-PFO} in~\cref{part:standard}.

Next, we define the following states on the $\sP, \sF$ registers.

\begin{definition}[$\mathsf{pf}$-relation state]
    For $L = \{(x_1, y_1), \dots, (x_\ell, y_\ell)\} \in \calR_\ell$ and $R = \{(x'_1,y'_1),\dots,(x'_r,y'_r)\} \in \calR_r$, where $\ell$ and $r$ are non-negative integers such that $\ell + r \leq N$, let
    \begin{align}
        \ket*{\pf_{L,R}} \coloneqq \frac{1}{\sqrt{(N-\ell-r)!}} \sum_{\pi \in \sSym_N} \delta_{\pi,L \cup R} \ket*{\pi} \frac{1}{\sqrt{3^N}} \sum_{f \in \{0,1,2\}^N} \omega_3^{f(x_1) + \cdots + f(x_\ell) - (f(x'_1) + \cdots + f(x'_r))} \ket*{f},
    \end{align}
    where $\delta_{\pi,L \cup R}$ is an indicator variable that equals $1$ if $\pi(x) = y$ for all $(x,y) \in L \cup R$, and is $0$ otherwise.
\end{definition}

Note that when $\ell = r= 0$, i.e., $L = R = \varnothing$ are both the empty relation, the $\pf$-relation state $\ket*{\pf_{\varnothing,\varnothing}}_{\gsP \gsF}$ is the uniform superposition over all permutations $\pi \in \sSym_N$ and all ternary functions $f \in \{0, 1, 2\}^N$,
    \begin{align}
        \ket*{\pf_{\varnothing,\varnothing}}_{\gsP \gsF} \coloneqq \frac{1}{\sqrt{N!}} \sum_{\pi \in \sSym_N} \ket*{\pi}_{\gsP} \otimes \frac{1}{\sqrt{3^N}} \sum_{f \in \{0,1\}^N} \ket*{f}_{\gsF}.
    \end{align}

Recall that a relation $R$ is \emph{bijective} if and only if $\abs{\Im(R)} = \abs{\Dom(R)} = \abs{R}$. Equivalently, writing $R = \{(x_1,y_1),\dots,(x_t,y_t)\}$, $R$ is bijective if $x_1,\dots,x_t$ are all distinct, and $y_1,\dots,y_t$ are also all distinct.

\begin{definition}
    Let $\calR^{2,\dist}$ be the set of all ordered pairs of relations $(L,R) \in \calR^2$ where $L \cup R$ is a bijective relation.
\end{definition}

\subsection{Orthonormality of the $\pf$-relation states}

\begin{claim}[Orthonormality of $\mathsf{pf}$-relation states]
\label{claim:phi-L-R-orthogonal}
    $\{\ket*{\pf_{L,R}}\}_{(L,R) \in \calR^{2,\dist}}$ is an orthonormal set of vectors. 
\end{claim}

\begin{proof}[Proof of~\cref{claim:phi-L-R-orthogonal}]
    For $x \in [N]$, let $e_x \in \{0,1,2\}^N$ denote the $N$-dimensional vector that has a $1$ in the $x$-th position, and is $0$ everywhere else. Then by writing $f(x)$ as $f(x) = f \cdot e_x$, we get
    \begin{align}
        &\frac{1}{\sqrt{3^N}} \sum_{f \in \{0,1,2\}^N} \omega_3^{f(x_1) + \cdots + f(x_\ell) - (f(x_1') + \cdots f(x_r'))} \ket*{f}_{\gsF} \\
        &= \frac{1}{\sqrt{3^N}} \sum_{f \in \{0,1,2\}^N} \omega_3^{f \cdot (e_{x_1} + \cdots + e_{x_\ell}) - f \cdot (e_{x_1'} + \cdots + e_{x_r'})} \ket*{f}_{\gsF} \\
        &= \mathsf{QFT}_3^{\otimes N} \ket*{(e_{x_1} + \cdots + e_{x_\ell}) - (e_{x'_1} + \cdots + e_{x'_r})  \ (\mathrm{mod} \ 3)}_{\gsF},
    \end{align}
    where $\mathsf{QFT}_3$ denotes the $3$-ary quantum Fourier transform. When $\{x_1,\dots,x_\ell,x'_1,\dots,x'_r\}$ are all distinct, there is a bijection between $(e_{x_1} + \cdots + e_{x_\ell}) - (e_{x'_1} + \cdots + e_{x'_r})$ and the sets $\{x_1,\dots,x_\ell\},\{x'_1,\dots,x'_r\}$: the first set corresponds to the indices where the vector is $1$, and the second set is the indices where the vector is $-1 \equiv 2 \ (\mathrm{mod} \ 3)$. Thus, there is an isometry that maps
    \begin{align}
        \frac{1}{\sqrt{3^N}} \sum_{f \in \{0,1,2\}^N} \omega_3^{f(x_1) + \cdots + f(x_\ell) - (f(x_1') + \cdots f(x_r'))} \ket*{f}_{\gsF} \mapsto \ket*{\{x_1,\dots,x_\ell\}}\ket*{\{x_1',\dots,x_r'\}},
    \end{align}
    whenever $\{x_1,\dots,x_\ell,x'_1,\dots,x'_r\}$ are all distinct. Thus, for any $L = \{(x_1,y_1),\dots,(x_\ell,y_\ell)\}$, $R = \{(x_1',y_1'),\dots,(x_r',y_r')\}$ where $L \cup R$ is a bijective relation, applying this isometry to the $\sF$ register of $\ket*{\pf_{L,R}}$ yields
    \begin{align}
        \ket*{\pf_{L,R}} \mapsto \frac{1}{\sqrt{(N-\ell-r)!}} \sum_{\pi \in \sSym_N} \delta_{\pi,L \cup R} \ket*{\pi}_{\gsP} \otimes \ket*{\{x_1,\dots,x_\ell\}} \ket*{\{x_1',\dots,x_r'\}}.
    \end{align}
    Next, we can apply an isometry that, controlled on $\ket*{\pi}$, sends each $x_i$ to the tuple $(x_i,\pi(x_i)) = (x_i,y_i)$. The result is
    \begin{align}
        \frac{1}{\sqrt{(N-\ell-r)!}} \sum_{\pi \in \sSym_N} \delta_{\pi,L \cup R} \ket*{\pi}_{\gsP} \otimes \ket*{\{(x_1,y_1),\dots,(x_\ell,y_\ell)\}} \ket*{\{(x_1',y_1'),\dots,(x_r',y_r')\}}.
    \end{align}
    Finally, controlled on the last two registers, we can uncompute the superposition on the $\sP$ register. The result is 
    \begin{align}
        \ket*{\{(x_1,y_1),\dots,(x_\ell,y_\ell)\}} \ket*{\{(x_1',y_1'),\dots,(x_r',y_r')\}} = \ket*{L} \ket*{R}.
    \end{align}
    This completes the proof.
\end{proof}

\begin{definition}
    Define the partial isometry $\Compress: \calH_{\sP} \otimes \calH_{\sF} \rightarrow \calH_{\sL} \otimes \calH_{\sR}$ to be
    \begin{align}
        \Compress \coloneqq \sum_{(L,R) \in \calR^{2,\dist}} \ket*{L}_{\gsL} \otimes \ket*{R}_{\gsR} \cdot \bra{\pf_R}_{\gsP \gsF}
    \end{align}
\end{definition}

Here, $\sL$ and $\sR$ are variable-length registers as defined in~\cref{subsec:relation-states}. Note that $\Compress$ is a partial isometry by~\cref{claim:phi-L-R-orthogonal}.

\subsection{How $\spfo$ acts on the $\pf$-relation states}

\begin{claim}[Action of $\spfo$]
\label{claim:ternary-pfo-action}
    For any $(L,R) \in \calR^{2,\dist}$ and $x \in [N]$ such that $x \not\in \Dom(L \cup R)$, we have
    \begin{align}
        \spfo \ket*{x}_{\gsA} \ket*{\pf_{L,R}}_{\gsP \gsF} = \frac{1}{\sqrt{N-\abs{L \cup R}}} \sum_{\substack{y \in [N]:\\ y\not\in \Im(L \cup R)}} \ket*{y}_{\gsA} \ket*{\pf_{L \cup \{(x,y)\},R}}_{\gsP \gsF}. \label{eq:tpfo-map}
    \end{align}
    Similarly, for any $(L,R) \in \calR^{2,\dist}$ and $y \in [N]$ such that $y \not\in \Im(L \cup R)$, we have
    \begin{align}
        \spfo^\dagger \ket*{y}_{\gsA} \ket*{\pf_{L,R}}_{\gsP \gsF} = \frac{1}{\sqrt{N-\abs{L \cup R}}} \sum_{\substack{x \in [N]:\\ x\not\in \Dom(L \cup R)}} \ket*{x}_{\gsA} \ket*{\pf_{L,R \cup \{(x,y)\}}}_{\gsP \gsF}. \label{eq:tpfo-inverse-map}
    \end{align}
\end{claim}

\begin{proof}[Proof of~\cref{claim:ternary-pfo-action}]
    Recall that 
    \begin{align}
        \spfo_{\gsA\gsP\gsF} \ket*{x}_{\gsA} \ket*{\pi}_{\gsP} \ket*{f}_{\gsF} = \sum_{y \in [N]} \ket*{y}_{\gsA} \delta_{\pi(x) = y} \ket*{\pi} \omega_3^{f(x)} \ket*{f}.
    \end{align}
    Let us write $L = \{(x_1,y_1),\dots,(x_\ell,y_\ell)\}$ and $R = \{(x'_1,y'_1),\dots,(x'_r,y'_r)\}$. Then
    \begin{align}
        \ket*{\pf_{L,R}}_{\gsP \gsF} = \frac{1}{\sqrt{(N-\ell-r)!}} \sum_{\pi \in \sSym_N} \delta_{\pi,L \cup R} \ket*{\pi}_{\gsP} \frac{1}{\sqrt{3^N}} \sum_{f \in \{0,1,2\}^N} \omega_3^{f(x_1) + \cdots + f(x_\ell) - f(x'_1) - \cdots - f(x'_r))} \ket*{f}_{\gsF},
    \end{align}
    Thus, we have
    \begin{align}
        &\spfo_{\gsA \gsP \gsF} \ket*{x}_{\gsA} \ket*{\pf_{L,R}}_{\gsP \gsF} \\
        &= \sum_{y \in [N]} \ket*{y}_{\gsA} \frac{1}{\sqrt{(N-\ell-r)!}} \sum_{\pi \in \sSym_N} \textcolor{red}{\delta_{\pi(x) = y}} \cdot \delta_{\pi,L \cup R} \ket*{\pi}_{\gsP} \nonumber \\
        & \quad \frac{1}{\sqrt{3^N}} \sum_{f \in \{0,1,2\}^N} \omega_3^{f(x_1) + \cdots + f(x_\ell) + \textcolor{red}{f(x)} - f(x'_1) - \cdots - f(x'_r))} \ket*{f}_{\gsF}.
    \end{align}
    In this sum, $\ket*{y}$ has a coefficient of $0$ whenever $y \in \Im(L \cup R)$, since in that case the constraints that $\delta_{\pi, L \cup R}$ and $\delta_{\pi(x) = y}$ are impossible to satisfy since $x\not\in \Dom(L \cup R)$, and thus satisfying both constraints would require $y$ to have two different preimages under the permutation $\pi$. We can therefore rewrite the above sum as
    \begin{align}
        &\spfo_{\gsA \gsP \gsF} \ket*{x}_{\gsA} \ket*{\pf_{L,R}}_{\gsP \gsF} \\
        &= \frac{1}{\sqrt{N - \ell - r}} \sum_{\substack{y \in [N]:\\ y\not\in \Im(L \cup R)}} \ket*{y}_{\gsA} \frac{1}{\sqrt{(N-\ell-1-r)!}} \sum_{\pi \in \sSym_N} \delta_{\pi,L \cup \{(x,y)\} \cup R} \ket*{\pi}_{\gsP} \nonumber \\
        & \quad \frac{1}{\sqrt{3^N}} \sum_{f \in \{0,1,2\}^N} \omega_3^{f(x_1) + \cdots + f(x_\ell) + f(x) - f(x'_1) - \cdots - f(x'_r))} \ket*{f}_{\gsF}\\
        &= \frac{1}{\sqrt{N - \ell - r}} \sum_{y\in [N]} \ket*{y}_{\gsA} \ket*{\pf_{L \cup \{(x,y)\},R}}_{\gsP \gsF}.
    \end{align}
    This completes the proof of~\cref{eq:tpfo-map}. Since $\spfo^\dagger$ applies the map
    \begin{align}
        \spfo^\dagger \ket*{y}_{\gsA} \ket*{\pi} \ket*{f} &=  \sum_{x \in [N]} \ket*{x}_{\gsA} \delta_{\pi(x) = y} \ket*{\pi} \omega_3^{-f(x)} \ket*{f},
    \end{align}
    the proof for~\cref{eq:tpfo-inverse-map} follows by a symmetric argument. 
\end{proof}

\section{The partial path-recording oracle $W$}
 
In the previous section, we proved~\cref{claim:ternary-pfo-action}, which partially characterizes how the unitaries $\spfo$ and $\spfo^\dagger$ act in terms of states $\ket*{x}_{\gsA} \ket*{\pf_{L,R}}_{\gsP \gsF}$. We also proved that there exists an isometry $\Compress$ that maps $\ket*{\pf_{L,R}}_{\gsP \gsF}$ to $\ket*{L}_{\gsL} \ket*{R}_{\gsR}$ for all pairs of relations $L,R$ such that their union $L \cup R$ is a bijective relation.
In this section, we will define a linear operator $W$ that we call the \emph{partial path recording oracle}. This $W$ operator, up to isometry, implements a restricted version of the $\spfo$ operator. In particular, we have the following.
\begin{itemize}
    \item On states of the form $\ket*{x}_{\gsA} \ket*{L}_{\gsL} \ket*{R}_{\gsR}$ such that $L \cup R$ is a bijection and $x\not\in \Dom(L \cup R)$, the linear map $W$ performs exactly the same map as $\spfo$ (up to isometry).
    \item On states of the form $\ket*{y}_{\gsA} \ket*{L}_{\gsL} \ket*{R}_{\gsR}$ such that $L \cup R$ is a bijection and $y\not\in \Im(L \cup R)$, the linear map  $W^\dagger$ performs exactly the same map as $\spfo^\dagger$ (up to isometry).
\end{itemize}
In the above, ``up to isometry'' refers to the isometry $\Compress$ that maps $\ket*{\pf_{L,R}}_{\gsP \gsF}$ to $\ket*{L}_{\gsL} \ket*{R}_{\gsR}$. Formally, the registers $\sL$ and $\sR$ are both \emph{variable-length} registers that store the two relations $L$ and $R$. We refer the reader to~\cref{prelim:variable-length-registers,subsec:relation-states} in the Preliminaries section for our definitions of variable-length registers, relations, and relation states.

\paragraph{The role of the $W$ operator in our proof.} Looking ahead to our main proof, we will show that if $C,D$ are sampled from any $n$-qubit $2$-design, then an adversary (making both forward and inverse queries) cannot distinguish between an oracle that implements $D_{\gsA} \cdot \spfo \cdot C_{\gsA}$ and an oracle that implements $D_{\gsA} \cdot W \cdot C_{\gsA}$, except with negligible advantage. Thus, even though $W$ only behaves like (a compressed version of) $\spfo$ on a restricted subspace, we will show that the twirling of $C, D$ prevents the adversary from detecting the difference.

In the next section, we will show that the $W$ operator can also be seen as a restricted version of another linear operator $V$ that we call the \emph{path-recording oracle}. The connection between $W$ and $V$ plays a crucial role in our proof; see~\cref{sec:symmetric-V} for further discussion.

\subsection{Defining $W^L$ and $W^R$}

Before we define $W$, we will first define helper operators $W^L$ and $W^R$. The $W^L$ operator is defined to capture the (partial) characterization of $\spfo$ given in~\cref{eq:tpfo-map}, while $W^R$ is defined to capture the (partial) characterization of $\spfo^\dagger$ given in~\cref{eq:tpfo-inverse-map}.

\begin{definition}[$W^L$ and $W^R$]
\label{def:ternary-pfo-action}
    Define $W^L$ to be the linear map such that for any $(L,R) \in \calR^{2,\dist}$ and $x \in [N]$ such that $x\not\in \Dom(L \cup R)$,
    \begin{align}
        W^L \cdot \ket*{x}_{\gsA} \ket*{L}_{\gsL} \ket*{R}_{\gsR} \coloneqq \frac{1}{\sqrt{N-\abs{L \cup R}}} \sum_{\substack{y \in [N]:\\ y\not\in \Im(L \cup R)}} \ket*{y}_{\gsA} \ket*{L \cup \{(x,y)\}}_{\gsL} \ket*{R}_{\gsR}. \label{eq:WL-def}
    \end{align}
    Similarly, define $W^R$ be the linear map such that for any $(L,R) \in \calR^{2,\dist}$ and $y \in [N]$ such that $y\not\in \Im(L \cup R)$,
    \begin{align}
        W^R \cdot \ket*{y}_{\gsA} \ket*{L}_{\gsL} \ket*{R}_{\gsR} \coloneqq \frac{1}{\sqrt{N-\abs{L \cup R}}} \sum_{\substack{x \in [N]:\\ x\not\in \Dom(L \cup R)}} \ket*{x}_{\gsA} \ket*{L}_{\gsL} \ket*{R \cup \{(x,y)\}}_{\gsR}. \label{eq:WR-def}
    \end{align}
\end{definition}

It is useful to define the following projectors to describe the actions of $W^L, W^R$.

\begin{definition}[Bijective-relation projectors] \label{def:bij-proj}
    Define the projectors
    \begin{align}
        \Pi^{\bij}_{\gsL \gsR} \coloneq \sum_{(L,R) \in \mathcal{R}^{2,\dist}} \ketbra*{L}_{\gsL} \otimes \ketbra*{R}_{\gsR}, \quad\quad \Pi^{\bij}_{\leq t, \gsL \gsR} \coloneq \Pi^{\bij}_{\gsL \gsR} \cdot \Pi_{\leq t, \gsL \gsR} = \Pi_{\leq t, \gsL \gsR} \cdot \Pi^{\bij}_{\gsL \gsR},
    \end{align}
    where the projector $\Pi_{\leq t, \gsL \gsR}$ is the maximum-length projector defined in \cref{notation:pi-leq-t}.
\end{definition}

By the definition of $W^L$ and $W^R$, we have the following fact about the action of $W^L$ and $W^R$ on states with a bounded length.

\begin{fact} \label{fact:WLWR-space-leqi}
    For any integer $i \geq 0$, $W^L, W^R$ map states in the subspace associated to the projector $\Id_{\gsA} \otimes \Pi^{\bij}_{\leq i, \gsL \gsR}$ into the subspace associated with the projector $\Id_{\gsA} \otimes \Pi^{\bij}_{\leq i+1, \gsL \gsR}$.
\end{fact}

The following property follows from the relation between $W^L, W^R$ and $\spfo, \spfo^\dagger$.

\begin{claim}
    $W^L$ and $W^R$ are both partial isometries.
\end{claim}

\begin{proof}
    Since $\spfo$ is a unitary operator, the operator obtained  by restricting the domain of $\spfo$ to the span of the states $\ket*{x} \ket*{\pf_{L,R}}$ is a partial isometry. Up to relabeling $\ket*{\pf_{L,R}}$ as $\ket*{L,R}$ (i.e., applying the partial isometry $\Compress$), this is $W^L$. Similarly, $\spfo^\dagger$ is a unitary, and the operator obtained by restricting $\spfo^\dagger$ to the span of states $\ket*{y} \ket*{\pf_{L,R}}$ is a partial isometry. Up to relabeling $\ket*{\pf_{L,R}}$ as $\ket*{L,R}$, this is $W^R$.
\end{proof}

\begin{notation}
    For a partial isometry $G$, let $\calD(G)$ and $\calI(G)$ denote its domain and image. Let $\Pi^{\calD(G)} = G^\dagger \cdot G$ and $\Pi^{\calI(G)} = G \cdot G^\dagger$ denote the orthogonal projectors onto $\calD(G)$ and $\calI(G)$.
\end{notation}

\begin{claim} \label{claim:pileqt-commutes-with-WLWRDomIm}
For all integers $t \geq 0$,
    $\Pi_{\leq t}$ commutes with $\Pi^{\calD(W^L)}$, $\Pi^{\calI(W^L)}$, $\Pi^{\calD(W^R)}$, and $\Pi^{\calI(W^R)}$.
\end{claim}
\begin{proof}
    By~\cref{fact:WLWR-space-leqi}, $\Pi^{\calD(W^L)} = W^{L,\dagger} \cdot W^{L}$ maps states from $\Id_{\gsA} \otimes \Pi^{\bij}_{\leq t,\gsL \gsR}$ to $\Id_{\gsA} \otimes \Pi^{\bij}_{\leq t,\gsL \gsR}$ for $t \geq 0$. This implies that $\Pi^{\calD(W^L)}$ commutes with $\Pi_{\leq t}$ for all $t \geq 0$.
    By~\cref{fact:WLWR-space-leqi}, $\Pi^{\calI(W^L)} = W^{L} \cdot W^{L,\dagger}$ maps states from $\Id_{\gsA} \otimes \Pi^{\bij}_{\leq t+1,\gsL \gsR}$ to $\Id_{\gsA} \otimes \Pi^{\bij}_{\leq t+1,\gsL \gsR}$ for $t+1 \geq 0$. This implies that $\Pi^{\calD(W^L)}$ commutes with $\Pi_{\leq t}$ for all $t \geq 1$. Additionally, $\Pi^{\calI(W^L)} = W^{L} \cdot W^{L,\dagger}$ has no support on $\Pi_{\leq 0}$, and thus it commutes with $\Pi_{\leq 0}$.
    By symmetric arguments, we obtain the analogous statements for $W^R$.
\end{proof}

It will be useful to state the connection between the $W^L, W^R$, and $\spfo$ more formally. 

\begin{fact} We have
    \begin{align}
    W^L &= \Compress \cdot \spfo \cdot \Compress^\dagger \cdot \Pi^{\calD(W^L)} =
    \Pi^{\calI(W^L)} \cdot \Compress \cdot \spfo \cdot \Compress^\dagger, \label{eq:WL-relation-pfo}\\
    W^R &= \Compress \cdot \spfo^\dagger \cdot \Compress^\dagger \cdot \Pi^{\calD(W^R)} = 
    \Pi^{\calI(W^R)} \cdot \Compress \cdot \spfo^\dagger \cdot \Compress^\dagger. \label{eq:WR-relation-pfo}
\end{align}
\end{fact}

\subsection{Defining $W$}

We now use $W^L$ and $W^R$ to define the partial path-recording oracle $W$.

\begin{definition} \label{def:symmetric-W}
    The partial path-recording oracle is the operator $W$ defined as
    \begin{equation}
        W \coloneqq W^L + W^{R,\dagger}.
    \end{equation}
\end{definition}

From \cref{fact:WLWR-space-leqi}, we immediately obtain the following fact.

\begin{fact} \label{fact:W-space-leqi}
    $\calD(W)$, $\calI(W)$ are subspaces of the image of $\Id_{\gsA} \otimes \Pi^{\bij}_{\gsL \gsR}$. Moreover, for any integer $i \geq 0$, $W$ and $W^\dagger$ map states in the subspace associated to the projector $\Id_{\gsA} \otimes \Pi^{\bij}_{\leq i, \gsL \gsR}$ into the subspace associated with the projector $\Id_{\gsA} \otimes \Pi^{\bij}_{\leq i+1, \gsL \gsR}$.
\end{fact}

\begin{claim}\label{claim:W-partial-isometry}
    $W$ is a partial isometry.
\end{claim}

\begin{proof}[Proof of~\cref{claim:W-partial-isometry}]
    Since $W^L$ and $W^R$ (and hence $W^{R,\dagger}$) are partial isometries, the operator $W = W^L + W^{R,\dagger}$ is a partial isometry as long as both of the following are true:
    \begin{itemize}
        \item The subspaces $\calD(W^L)$ and $\calD(W^{R,\dagger}) = \calI(W^R)$ are orthogonal, i.e., $W$ is a sum of two partial isometries with orthogonal domains.
        \item The subspaces $\calI(W^L)$ and $\calI(W^{R,\dagger}) = \calD(W^R)$ are orthogonal, i.e., $W$ is a sum of two partial isometries with orthogonal images. 
    \end{itemize}
    $\calD(W^L)$ and $\calI(W^R)$ are orthogonal because $\calD(W^L)$ is only supported on states $\ket*{x} \ket*{L} \ket*{R}$ where $x\not\in \Dom(L \cup R)$, while $\calI(W^R)$ is only supported on states $\ket*{x} \ket*{L} \ket*{R}$ where $x \in \Dom(L \cup R)$ (this can be seen by inspecting the right-hand-side of~\cref{eq:WR-def}). A symmetric argument shows that $\calD(W^R)$ and $\calI(W^L)$ are also orthogonal, which completes the proof.
\end{proof}

In fact, our proof of~\cref{claim:W-partial-isometry} establishes the following relationship between the domain and image of $W$ and the domain and image of $W^L$ and $W^R$.

\begin{fact}
\label{fact:domain-image-W-WL-WR}
    The domain and image of $W$ are given by
    \begin{align}
        \Pi^{\calD(W)} &= \Pi^{\calD(W^L)} + \Pi^{\calI(W^R)}, \label{eq:expand-DW}\\
        \Pi^{\calI(W)} &= \Pi^{\calD(W^R)} + \Pi^{\calI(W^L)}. \label{eq:expand-IW}
    \end{align}
\end{fact}

\begin{claim} \label{claim:pileqt-commutes-with-manything}
For all integers $t \geq 0$,
    $\Pi_{\leq t}$ commutes with $\Pi^{\calD(W)}$ and $\Pi^{\calI(W)}$.
\end{claim}
\begin{proof}
This follows immediately from \cref{claim:pileqt-commutes-with-WLWRDomIm}, which states that the projector $\Pi_{\leq t}$ commutes with the projectors $\Pi^{\calD(W^L)}$, $\Pi^{\calI(W^L)}$, $\Pi^{\calD(W^R)}$, $\Pi^{\calI(W^R)}$.
\end{proof}

\begin{corollary} \label{cor:PiLRR2D-PiDW-projector-lesseq-t}
    For all integers $t \geq 0$, the image of $\Pi^{\calD(W)}_{\leq t, \gsA \gsL \gsR}$ is a subspace of the image of $\Id_{\gsA} \otimes \Pi^{\bij}_{\leq t, \gsL \gsR}$. Similarly, the image of $\Pi^{\calI(W)}_{\leq t, \gsA \gsL \gsR}$ is a subspace of the image of $\Id_{\gsA} \otimes \Pi^{\bij}_{\leq t, \gsL \gsR}$.
\end{corollary}

Using~\cref{fact:domain-image-W-WL-WR}, we can now establish the following relationship between $W$ and $\spfo$. 

\begin{claim}[$W$ is a restriction of $\spfo$ up to isometry]
\label{claim:relate-W-and-spfo}
We have
    \begin{align}
        W &= \Compress \cdot \spfo \cdot \Compress^\dagger \cdot \Pi^{\calD(W)},\label{eq:compress-proof-spfo-goal-1}\\
        W^\dagger &= \Compress \cdot \spfo^\dagger \cdot \Compress^\dagger \cdot \Pi^{\calI(W)}. \label{eq:compress-proof-spfo-goal-2}
    \end{align}
\end{claim}

In words,~\cref{claim:relate-W-and-spfo} says that for any state in $\calD(W)$, the domain of $W$, the action of $W$ is the same as $\spfo$ up to isometry. Additionally, it says that for any state in the image in $\calI(W)$, the image of $W$, the action of $W^\dagger$ is the same as $\spfo^\dagger$ up to isometry.

\begin{proof}[Proof of~\cref{claim:relate-W-and-spfo}]
    We will prove the first equality, \cref{eq:compress-proof-spfo-goal-1}; the second equality, \cref{eq:compress-proof-spfo-goal-2}, follows from a symmetric argument. From \cref{eq:WL-relation-pfo} and \cref{eq:WR-relation-pfo}, we have
    \begin{align}
        \Compress \cdot \spfo \cdot \Compress^\dagger \cdot \Pi^{\calD(W^L)} &= W^L, \label{eq:compress-proof-goal-1-spfo}\\
        \Compress \cdot \spfo \cdot \Compress^\dagger \cdot \Pi^{\calI(W^R)} &= W^{R,\dagger}.\label{eq:compress-proof-goal-2-spfo}
    \end{align}
    Summing~\cref{eq:compress-proof-goal-1-spfo,eq:compress-proof-goal-2-spfo} yields
    \begin{align}
        \Compress \cdot \spfo \cdot \Compress^\dagger \cdot (\Pi^{\calD(W^L)} + \Pi^{\calI(W^R)}) = W^L + W^{R,\dagger},
    \end{align}
    and plugging in $\Pi^{\calD(W)} = \Pi^{\calD(W^L)} + \Pi^{\calI(W^R)}$ from~\cref{eq:expand-DW} and $W = W^L + W^{R,\dagger}$ yields~\cref{eq:compress-proof-spfo-goal-1}.
\end{proof}

\section{The path-recording oracle $V$}
\label{sec:symmetric-V}

In the previous section, we defined a linear operator $W$ and showed that $W$ acts as a restricted version of $\spfo$, up to an application of the $\Compress$ isometry. In this section, we will introduce a second linear operator $V$, which will satisfy a number of key properties that will be crucial for our proof. We will show that $V$ satisfies the following properties:
\begin{itemize}
    \item \textbf{$V$ is indistinguishable from $W$ under twirling}, i.e., for $C,D$ sampled from any $n$-qubit unitary $2$-design $\frakD$, an adversary making forward and inverse queries cannot distinguish between queries to $D_{\gsA} \cdot V \cdot C_{\gsA}$ and queries to $D_{\gsA} \cdot W \cdot C_{\gsA}$.
    \item \textbf{$V$ satisfies approximate unitary invariance}, which we will use to conclude the following: an adversary making forward and inverse queries cannot distinguish between queries to $D_{\gsA} \cdot V \cdot C_{\gsA}$ for $C,D$ sampled from any $n$-qubit unitary $2$-design $\frakD$, and plain queries to $V$.\footnote{For technical reasons, our main proof will handle both of these bullets in one argument.}
\end{itemize}

We will refer to $V$ as the \emph{path-recording oracle}. We remark that this definition of $V$ is different from the one given in~\cref{part:standard}, as this $V$ will need to be designed to handle forward and inverse queries. In~\cref{subsec:imp-forward-inverse-q} we describe how to implement $V$ efficiently.

\subsection{Defining $V^L$ and $V^R$}

To define $V$, we first introduce helper operators $V^L$ and $V^R$.

\begin{definition}[left and right partial isometries] \label{def:V-sym-PRO}
    Let $V^L$ be the linear operator that acts as follows. For $x \in [N]$ and $(L,R) \in \mathcal{R}^{2,\leq N-1}$,
    \begin{align}
        V^L \cdot \ket*{x}_{\gsA} \ket*{L}_{\gsL} \ket*{R}_{\gsR} \coloneqq \sum_{\substack{y \in [N]:\\ y\not\in \Im(L \cup R)}} \frac{1}{\sqrt{N - \abs{\Im(L \cup R)}}} \ket*{y}_{\gsA} \ket*{L \cup \{(x,y)\}}_{\gsL} \ket*{R}_{\gsR}.
    \end{align}
    Define $V^R$ to be the linear operator such that for all $y \in [N]$ and $(L,R) \in \mathcal{R}^{2,\leq N-1}$,
    \begin{align}
        V^R \cdot \ket*{y}_{\gsA} \ket*{L}_{\gsL} \ket*{R}_{\gsR} \coloneqq \sum_{\substack{x \in [N]:\\ x\not\in \Dom(L \cup R)}} \frac{1}{\sqrt{N - \abs{\Dom(L \cup R)}}} \ket*{x}_{\gsA} \ket*{L}_{\gsL} \ket*{R \cup \{(x, y)\} }_{\gsR}.
    \end{align}
    By construction, $V^L$ and $V^R$ take states in $\Id_{\gsA} \otimes \Pi^{\calR^2}_{\leq i, \gsL \gsR}$ to $\Id_{\gsA} \otimes \Pi^{\calR^2}_{\leq i+1, \gsL \gsR}$.
\end{definition}

\paragraph{Why these definitions of $V^L$ and $V^R$?} On states of the form $\ket*{x} \ket*{L} \ket*{R}$ within the domain of $W^L$, the operators $W^L$ and $V^L$ act in the same way. However, the domain of $W^L$ is limited to states $\ket*{x} \ket*{L} \ket*{R}$ where $L \cup R$ forms a bijection and $x \notin \Dom(L \cup R)$ (which also implies that $\abs{L \cup R} \leq N-1$). On the other hand, the definition of $V^L$ extends $W^L$ so that it acts on all $\ket*{x} \ket*{L} \ket*{R}$ satisfying $\abs{L \cup R} \leq N-1$. In particular, we have dropped the requirement that $L \cup R$ is a bijection and that $x\not\in \Dom(L \cup R)$. An analogous relationship holds between $V^R$ and $W^R$. We define these extended operators, $V^L$ and $V^R$, to establish a property known as (approximate) unitary invariance (see~\cref{claim:VL-VR-approx-invariance}). Importantly, this property holds only for the extended operators $V^L$ and $V^R$, and not for the original $W^L$ and $W^R$ operators.

\begin{claim}
\label{claim:VL-VR-isometry}
    $V^L$ and $V^R$ are partial isometries.
\end{claim}

\begin{proof}
    We will give the proof for $V^L$; the proof for $V^R$ follows by a symmetric argument. $V^L$ is a partial isometry if and only if $V^L \cdot V^{L,\dagger}$ is the orthogonal projector onto $\calD(V^L)$. From the definition of $V^L$, we can see that its domain is
    \begin{align}
        \calD(V^L) = \mathrm{span}\{\ket*{x}_{\gsA} \ket*{L}_{\gsL} \ket*{R}_{\gsR}: x \in [N], (L, R) \in \mathcal{R}^{2,\leq N-1}\}.
    \end{align}
    It suffices to show that for all $x,x' \in [N]$, and $(L,R) \in \mathcal{R}^{2,\leq N-1}$ and $(L',R') \in \calR^{2,\leq N-1}$ that
    \begin{align}
        \bra*{x'}_{\gsA} \bra*{L'}_{\gsL} \bra*{R'}_{\gsR} \cdot V^{L,\dagger} \cdot V^L \cdot \ket*{x}_{\gsA} \ket*{L}_{\gsL} \ket*{R}_{\gsR} = \braket*{x'}{x}_{\gsA} \braket*{L'}{L}_{\gsL} \braket*{R'}{R}_{\gsR}.
    \end{align}
    We can expand out the LHS as
    \begin{align}
        \Big(\sum_{y' \not\in \Im(L' \cup R')} \frac{\bra*{y'}_{\gsA} \bra*{L' \cup x'y'}_{\gsL} \bra*{R'}_{\gsR}}{\sqrt{N - \abs{\Im(L' \cup R')}}} \Big)\cdot \Big(\sum_{y \not\in \Im(L \cup R)} \frac{\ket*{y}_{\gsA} \ket*{L \cup \{(x,y)\}}_{\gsL} \ket*{R}_{\gsR}}{\sqrt{N - \abs{\Im(L \cup R)}}} \Big)
    \end{align}
    The summand is zero unless $y' = y$, $L' \cup x'y' = L \cup \{(x,y)\}$, and $R' = R$. Combining the first two constraints, we have $L' \cup x'y = L \cup \{(x,y)\}$. Since $y$ does not appear in either $\Im(L')$ or $\Im(L)$, this implies $x ' = x$ and $L' = L$. This means that the sum is $0$ unless $x = x'$, $L = L'$ and $R = R'$. When these constraints are satisfied, the sum becomes $\sum_{y \not\in \Im(L \cup R)} 1/(N - \abs{\Im(L \cup R)}) = 1$. This completes the proof that $V^L$ is a partial isometry.
\end{proof}

\subsection{Defining $V$}

\begin{definition}
\label{def:symmetric-V}
    The path-recording oracle is the operator $V$ defined as
    \begin{align}
        V &= V^L \cdot (\Id - V^R \cdot V^{R,\dagger}) + (\Id - V^L \cdot V^{L,\dagger}) \cdot V^{R,\dagger}.
    \end{align}
    By construction, $V$ and $V^\dagger$ take states in $\Id_{\gsA} \otimes \Pi^{\calR^2}_{\leq i, \gsL \gsR}$ to $\Id_{\gsA} \otimes \Pi^{\calR^2}_{\leq i+1, \gsL \gsR}$ for any integer $i \geq 0$.
\end{definition}

\paragraph{Why this definition of $V$?} Recall that since we defined $W \coloneqq W^L + W^{R,\dagger}$, it might seem natural to define $V \coloneqq V^L + V^{R,\dagger}$. However, if we defined $V$ this way, it would \emph{not} be a partial isometry. As we showed in the proof of~\cref{claim:W-partial-isometry}, $W^L + W^{R,\dagger}$ is a partial isometry because $W^L$ and $W^{R,\dagger}$ do not ``overlap'', i.e., they are partial isometries with orthogonal domains and orthogonal images . On the other hand, this is not true for $V^L$ and $V^{R,\dagger}$. Thus, in order to ensure that $V$ is a partial isometry, we need to ``project out'' the overlap between $V^L$ and $V^{R,\dagger}$. 

\begin{claim} \label{claim:V-partial-iso}
    $V$ is a partial isometry.
\end{claim}

\begin{proof}
    We will first show that $V^L \cdot (\Id - V^R \cdot V^{R,\dagger})$ is a partial isometry. This is true if and only if $(\Id - V^R \cdot V^{R,\dagger}) \cdot V^{L,\dagger} \cdot V^L \cdot (\Id - V^R \cdot V^{R,\dagger})$ is a projector. To show that this operator is a projector, it suffices to show that $\Pi^{\calD(V^L)} = V^{L,\dagger} \cdot V^L $ and $\Pi^{\calI(V^R)} = V^R \cdot V^{R,\dagger}$ commute. From the definition of $V^L$, its domain is the image of the projector $\Id_{\gsA} \otimes \Pi^{\calR^2}_{\leq N-1, \gsL \gsR}$. Since $V^R$ takes states in $\Pi^{\calR^2}_{\leq i, \gsL \gsR}$ to $\Pi^{\calR^2}_{\leq i+1, \gsL \gsR}$ (for $0 \leq i \leq N-1$), it follows that $V^R \cdot V^{R,\dagger}$ takes states in $\Pi^{\calR^2}_{\leq i+1, \gsL \gsR}$ to $\Pi^{\calR^2}_{\leq i+1, \gsL \gsR}$ (for $0 \leq i \leq N-1$). In particular, this means it commutes with $\Id_{\gsA} \otimes \Pi^{\calR^2}_{\leq N-1}$. Using a symmetric argument, we can conclude that $(\Id - V^L \cdot V^{L,\dagger}) \cdot V^{R,\dagger}$ is also a partial isometry.

    Now, we just need to show that the sum of these two partial isometries is a partial isometry. It suffices to show that their domains are orthogonal and their images are orthogonal. To see that their domains are orthogonal, note that the domain of $V^L \cdot (\Id - V^R \cdot V^{R,\dagger})$ is a subspace of $\Id - \Pi^{\calI(V^R)}$, while the domain of $(\Id - V^L \cdot V^{L,\dagger}) \cdot V^{R,\dagger}$ is a subspace of $\Pi^{\calI(V^R)}$, and hence they are orthogonal. A symmetric argument shows their images are orthogonal. This completes the proof.
\end{proof}

$V$ being a partial isometry implies that any state generated by an adversary that queries $V$ and $V^\dagger$ will have a norm at most $1$.
This is an important property that will be central to our strong PRU proof.
Recall that in the standard PRU proof of \cref{part:standard}, the path-recording oracle acts as an isometry on all states that can be generated by querying the path-recording oracle.
This first property of $V$ being a partial isometry is a relaxation of the isometric property of the standard path-recording oracle.
While $V$ is a partial isometry, we will later show that the state generated by an adversary that queries $V$ and $V^\dagger$ will have a norm close to one for subexponential number of queries.

\subsection{Two-sided unitary invariance}

The path-recording oracle $V$ satisfies an (approximate) two-sided unitary invariance property, which we state below.

\begin{definition} \label{def:multi-rot-Q}
    For any $n$-qubit unitary $C,D$, define
    \begin{align}
        Q[C,D] &\coloneqq (C \otimes D^T)^{\otimes *}_{\gsL} \otimes (\overline{C} \otimes D^{\dagger})^{\otimes *}_{\gsR}.
    \end{align}
\end{definition}

\begin{claim}[two-sided unitary invariance]
\label{claim:two-sided-invariance}
    For any integer $0 \leq t \leq N-1$ and any pair of $n$-qubit unitaries $C,D$,
    \begin{align}
        \norm{D_{\gsA} \cdot V_{\leq t} \cdot C_{\gsA} \otimes Q[C,D]_{\gsL \gsR}
        - Q[C,D]_{\gsL \gsR} \cdot V_{\leq t}}_{\opnorm} & \leq 16\sqrt{\frac{2t(t+1)}{N}},\\
        \norm{C_{\gsA}^\dagger \cdot (V^\dagger)_{\leq t} \cdot D_{\gsA}^\dagger \otimes Q[C,D]_{\gsL \gsR}
        - Q[C,D]_{\gsL \gsR} \cdot (V^\dagger)_{\leq t}}_{\opnorm} & \leq 16\sqrt{\frac{2t(t+1)}{N}},
    \end{align}
\end{claim}

\cref{claim:two-sided-invariance} is proven in~\cref{sec:symmetric-V-approx-invariance}.
The two-sided unitary invariance of $V$ allows us to move the random unitaries $C$ and $D$ acting on system register $\sA$ to the purifying registers $\sL, \sR$.

\subsection{$W$ is a restriction of $V$}

We now show that $W$ is a restriction of $V$. First, we need the following basic facts relating $W^L, W^R, V^L$, and $V^R$ that follow immediately from the definitions of these operators.

\begin{fact}
    We have
    \begin{itemize}
        \item $W^L$ is a restriction of $V^L$ and $W^R$ is a restriction of $V^R$:
        \begin{align}
        W^L &= V^L \cdot \Pi^{\calD(W^L)} = \Pi^{\calI(W^L)} \cdot V^L \label{eq:WL-VL-relation} \\
        W^R &= V^R \cdot \Pi^{\calD(W^R)} = \Pi^{\calI(W^R)} \cdot V^R \label{eq:WR-VR-relation}
        \end{align}
        \item The image of $V^R$ is in the kernel of $W^L$, and the image of $V^L$ is in the kernel of $W^R$, i.e.,
        \begin{align}
        W^L \cdot V^R = W^R \cdot V^L &= 0 \label{eq:im-VR-kernel-WL},
    \end{align}
    \end{itemize}
\end{fact}

\begin{lemma} \label{lem:P1P2P1-prop}
    If $\Pi_1$ and $\Pi_2$ are projectors, and $\Pi_1 = \Pi_1 \Pi_2 \Pi_1$ then $\Pi_1$ is a subspace of $\Pi_2$.
\end{lemma}
\begin{proof}
    Consider any normalized state $\ket*{\psi} \in \Pi_1$, i.e., $\Pi_1 \ket*{\psi} = \ket*{\psi}$. We have the following identity,
    \begin{equation}
        1 = \bra*{\psi} \Pi_1 \ket*{\psi} = \bra*{\psi} \Pi_1 \Pi_2 \Pi_1 \ket*{\psi} = \bra*{\psi} \Pi_2 \ket*{\psi}.
    \end{equation}
    Because $\Pi_2$ is a projector and $\bra*{\psi} \Pi_2 \ket*{\psi} = 1$, we have $\ket*{\psi} \in \Pi_2$.
\end{proof}

\begin{lemma} \label{lem:partial-iso-V1V2-prop}
    Consider any partial isometries $V_1, V_2$. If $V_2 = V_1 \cdot \Pi^{\calD(V_2)}$, then $\calD(V_2)$ is a subspace of $\calD(V_1)$. And if $ V_2 = \Pi^{\calI(V_2)} \cdot V_1$, then $\calI(V_2)$ is a subspace of $\calI(V_1)$.
\end{lemma}
\begin{proof}
    From $V_2 = V_1 \cdot \Pi^{\calD(V_2)}$, we have
    \begin{equation}
        \Pi^{\calD(V_2)} = V_2^\dagger \cdot V_2 = \Pi^{\calD(V_2)} \cdot V_1^\dagger \cdot V_1 \cdot \Pi^{\calD(V_2)} = \Pi^{\calD(V_2)} \cdot \Pi^{\calD(V_1)} \cdot \Pi^{\calD(V_2)}.
    \end{equation}
    Hence from \cref{lem:P1P2P1-prop}, we have $\calD(V_2)$ is a subspace of $\calD(V_1)$.

    From $V_2 = \Pi^{\calI(V_2)} \cdot V_1$, we have
    \begin{equation}
        \Pi^{\calI(V_2)} = V_2 \cdot V_2^\dagger = \Pi^{\calI(V_2)} \cdot V_1 \cdot V_1^\dagger \cdot \Pi^{\calI(V_2)} = \Pi^{\calI(V_2)} \cdot \Pi^{\calI(V_1)} \cdot \Pi^{\calI(V_2)}.
    \end{equation}
    Hence from \cref{lem:P1P2P1-prop}, we have $\calI(V_2)$ is a subspace of $\calI(V_1)$.
\end{proof}

\begin{corollary} \label{cor:image-WL-VL-WR-VR}
    $\calI(W^L)$ is a subspace of $\calI(V^L)$. And $\calI(W^R)$ is a subspace of $\calI(V^R)$.
\end{corollary}
\begin{proof}
    This follows immediately from \cref{eq:WL-VL-relation}, \cref{eq:WR-VR-relation}, and \cref{lem:partial-iso-V1V2-prop}.
\end{proof}
    
\begin{claim}[$W$ is a restriction of $V$]
\label{claim:relate-W-and-V}
We have
    \begin{align}
        W &= V \cdot \Pi^{\calD(W)}, \label{eq:relate-W-and-V-goal1}\\
        W^\dagger &= V^\dagger \cdot \Pi^{\calI(W)}. \label{eq:relate-W-and-V-goal2}
    \end{align}
\end{claim}

In words,~\cref{claim:relate-W-and-V} says that for any state in $\calD(W)$, the domain of $W$, the action of $W$ is the same as $V$. Additionally, it says that for any state in the image in $\calI(W)$, the image of $W$, the action of $W^\dagger$ is the same as $V^\dagger$.

\begin{proof}[Proof of~\cref{claim:relate-W-and-V}]
    To prove~\cref{eq:relate-W-and-V-goal1}, it suffices to show that
    \begin{align}
        V \cdot \Pi^{\calD(W^L)} &= W^L \label{eq:VW-proof-goal-1},\\
        V \cdot \Pi^{\calI(W^R)} &= W^{R,\dagger}. \label{eq:VW-proof-goal-2}
    \end{align}
    This is because summing these two equations gives
    \begin{align}
        V \cdot (\Pi^{\calD(W^L)} + \Pi^{\calI(W^R)}) = W^L + W^{R,\dagger},
    \end{align}
    and plugging in $\Pi^{\calD(W)} = \Pi^{\calD(W^L)} + \Pi^{\calI(W^R)}$ from~\cref{eq:expand-DW} and $W = W^L + W^{R,\dagger}$ yields~\cref{eq:relate-W-and-V-goal1}. It remains to prove~\cref{eq:VW-proof-goal-1,eq:VW-proof-goal-2}.
    \begin{itemize}
        \item \textbf{Proof of~\cref{eq:VW-proof-goal-1}}. By the definition of $V$, we have
    \begin{align}
        V \cdot \Pi^{\calD(W^L)} = \Big(V^L \cdot (\Id - V^R \cdot V^{R,\dagger}) + (\Id - V^L \cdot V^{L,\dagger}) \cdot V^{R,\dagger} \Big) \cdot \Pi^{\calD(W^L)}.
    \end{align}
    Note that $V^{R,\dagger} \cdot \Pi^{D({W^L})} = V^{R,\dagger} \cdot W^{L,\dagger} \cdot W^L = (W^L \cdot V^R)^{\dagger} \cdot W^L = 0$, where the final equality uses~\cref{eq:im-VR-kernel-WL}. Thus, 
    \begin{align}
        V \cdot \Pi^{\calD(W^L)} = V^L \cdot \Pi^{\calD(W^L)} = W^L,
    \end{align}
    where the second equality follows from~\cref{eq:WL-VL-relation}.
    \item \textbf{Proof of~\cref{eq:VW-proof-goal-2}}. By the definition of $V$, \begin{align}
        V \cdot \Pi^{\calI(W^R)} = \Big(V^L \cdot (\Id - V^R \cdot V^{R,\dagger}) + (\Id - V^L \cdot V^{L,\dagger}) \cdot V^{R,\dagger} \Big) \cdot \Pi^{\calI(W^R)}.
    \end{align}
    Since $\calI(W^R)$ is a subspace of $\calI(V^R)$ by \cref{cor:image-WL-VL-WR-VR}, we have $V^L \cdot (\Id - V^R \cdot V^{R,\dagger}) \cdot \Pi^{\calI(W^R)} = 0$. Next, we have $V^{R,\dagger} \cdot \Pi^{\calI(W^R)} = (\Pi^{\calI(W^R)} \cdot V^R)^{\dagger} = W^{R,\dagger}$ by ~\cref{eq:WR-VR-relation}. Thus, we have
    \begin{align}
        V \cdot \Pi^{\calI(W^R)} &= (\Id - V^L \cdot V^{L,\dagger}) \cdot W^{R,\dagger}\\
        &= W^{R,\dagger} - V^L \cdot V^{L,\dagger} \cdot W^{R,\dagger} \\
        &= W^{R,\dagger},
    \end{align}
    where the last equality uses the fact that $V^{L,\dagger} \cdot W^{R,\dagger} = (W^R \cdot V^L)^{\dagger} = 0$ from~\cref{eq:im-VR-kernel-WL}.
    \end{itemize}
    This completes the proof of~\cref{eq:relate-W-and-V-goal1}. The proof of \cref{eq:relate-W-and-V-goal2} follows by a symmetric argument.
\end{proof}

\begin{corollary} \label{cor:subspace-relation-V-W}
    $\Pi^{\calD(W)}$ is a subspace of $\Pi^{\calD(V)}$.
    And $\Pi^{\calI(W)}$ is a subspace of $\Pi^{\calI(V)}$.
\end{corollary}
\begin{proof}
    This follows immediately from \cref{claim:relate-W-and-V} and \cref{lem:partial-iso-V1V2-prop}.
\end{proof}

\begin{corollary} \label{claim:relate-W-V-projector}
We have
    \begin{align}
        W^\dagger \cdot V &= \Pi^{\calD(W)}\\
        W \cdot V^\dagger &= \Pi^{\calI(W)}.
    \end{align}
\end{corollary}
\begin{proof}
    From $W = V \cdot \Pi^{\calD(W)}$, we can multiply $V^\dagger$ on the left of both sides to obtain
    \begin{equation}
        V^\dagger \cdot W = V^\dagger \cdot V \cdot \Pi^{\calD(W)}.
    \end{equation}
    Using $V^\dagger \cdot V = \Pi^{\calD(V)}$, we have
    \begin{equation}
        V^\dagger \cdot W = \Pi^{\calD(V)} \cdot \Pi^{\calD(W)} = \Pi^{\calD(W)},
    \end{equation}
    since $\Pi^{\calD(W)}$ is a subspace of $\Pi^{\calD(V)}$ from \cref{cor:subspace-relation-V-W}.
    Taking dagger yields $W^\dagger \cdot V = \Pi^{\calD(W)}$.
    
    From $W^\dagger = V^\dagger \cdot \Pi^{\calI(W)}$, we can multiply $V$ on the left of both sides to obtain
    \begin{equation}
        V \cdot W^\dagger = V \cdot V^\dagger \cdot \Pi^{\calI(W)}.
    \end{equation}
    Using $V \cdot V^\dagger = \Pi^{\calI(V)}$, we have
    \begin{equation}
        V \cdot W^\dagger = \Pi^{\calI(V)} \cdot \Pi^{\calI(W)} = \Pi^{\calI(W)},
    \end{equation}
    since $\Pi^{\calI(W)}$ is a subspace of $\Pi^{\calI(V)}$ from \cref{cor:subspace-relation-V-W}.
    Taking dagger yields $W \cdot V^\dagger = \Pi^{\calI(W)}$.
\end{proof}

\section{The strong PRU proof}

\subsection{Setup}

We define a distribution over $n$-qubit unitaries parameterized by any $n$-qubit unitary $2$-design $\frakD$.

\begin{definition}[$\mathsf{sPRU}(\mathfrak{D})$ distribution]
    For any distribution $\mathfrak{D}$ supported on $\calU(N)$, define the distribution $\mathsf{sPRU}({\frakD})$ as follows:
    \begin{enumerate}
        \item Sample a uniformly random permutation $\pi \gets \sSym_{N}$, a uniformly random $f \gets \{0,1, 2\}^N$, and two independently sampled $n$-qubit unitaries $C, D \gets \mathfrak{D}$. Following the definitions in~\cref{sec:PF3-oracle}, 
        \begin{align}
            F_f \coloneqq \sum_{x \in [N]} e^{2 \pi \cdot f(x) \cdot i/3} \ketbra*{x} \quad \text{and} \quad P_{\pi} \coloneqq \sum_{x \in [N]} \ketbra*{\pi(x)}{x}.
        \end{align}
        \item Output the $n$-qubit unitary $\calO \coloneqq D \cdot P_\pi \cdot F_f \cdot C$.
    \end{enumerate} 
\end{definition}

The goal of this section is to prove the following theorem.

\begin{theorem}[$\mathsf{sPRU}(\frakD)$ is a statistical strong PRU]\label{thm:statistical-PRU-strong} Let $\Adv$ be a $t$-query oracle adversary that can perform forward and inverse queries and let $\frakD$ be an exact unitary $2$-design. Then
\begin{align}
    \TD\left(\E_{\calO \gets \mathsf{sPRU}(\frakD)} \ketbra*{\Adv_t^{\calO}}_{\gsA \gsB}, \E_{\calO \gets \mu_{\mathsf{Haar}}} \ketbra*{\Adv_t^{\calO}}_{\gsA \gsB} \right) \leq \frac{18 t(t+1)}{N^{1/8}}
\end{align}
\end{theorem}

Since quantum-secure pseudorandom permutations and pseudorandom functions exist assuming one-way functions by~\cite{zhandry2016note,zhandry2021construct}, the existence of computationally-secure strong PRUs follows immediately from \cref{thm:statistical-PRU-strong}.

\begin{theorem}
    If quantum-secure one-way functions exist, then strong pseudorandom unitaries exist.
\end{theorem}

The main technical component of the proof of~\cref{thm:statistical-PRU-strong} is~\cref{lemma:pfd-cho-strong}, which relates the PRU adversary to an adversary that queries the path-recording oracle $V$, defined previously in~\cref{sec:symmetric-V}. Recall that $V$ is a partial isometry that acts on registers $(\sA,\sL,\sR)$, where $\sL$ and $\sR$ are variable-length registers. Initially, $\sL$ and $\sR$ are both initialized to the length-$0$ state $\ket*{\varnothing}$. To state~\cref{lemma:pfd-cho-strong}, we will need the following definition.

\begin{definition}[the global state after queries to $V$]
    For a $t$-query oracle adversary $\Adv$ that can perform forward and inverse queries and any $0 \leq i \leq t$, let
    \begin{align}
        \ket*{\calA_i^V}_{\gsA \gsB \gsL \gsR} \coloneqq \prod_{i = 1}^t \Bigg( \Big( (1-b_i) \cdot V_{\gsA \gsL \gsR} + b_i \cdot V_{\gsA \gsL \gsR}^\dagger \Big) \cdot A_{i,\gsA \gsB} \Bigg) \ket*{0^{n+m}}_{\gsA \gsB} \otimes \ket*{\varnothing}_{\gsL} \ket*{\varnothing}_{\gsR}
    \end{align}
    denote the global state on registers $\sA,\sB,\sL,\sR$ after $\calA$ makes $i$ queries to $V$.
\end{definition}

\begin{lemma}[$\mathsf{sPRU}(\mathfrak{D})$ is indistinguishable from $V$]\label{lemma:pfd-cho-strong}
     Let $\mathfrak{D}$ be any exact unitary $2$-design. For any $t$-query oracle adversary $\calA$,
     \begin{align}
        \TD\left(\E_{\calO \gets \mathsf{sPRU}(\frakD)} \ketbra*{\Adv_t^{\calO}}_{\gsA \gsB}, \,\,\, \Tr_{\sL \sR}\left( \ketbra*{\Adv^{V}_t}_{\gsA \gsB \gsL \gsR} \right) \right) \leq \frac{9t(t+1)}{N^{1/8}} \label{eq:intermediate-step-strong-main-thm}
    \end{align}
\end{lemma}

\paragraph{\cref{lemma:pfd-cho-strong} implies~\cref{thm:statistical-PRU-strong}.} \cref{lemma:pfd-cho-strong} implies \cref{thm:statistical-PRU-strong} by the following argument. We can instantiate $\frakD = \mu_{\mathsf{Haar}}$, i.e., $\frakD$ outputs a Haar-random $n$-qubit unitary. Then the output of $\spru(\mathfrak{D}) = \spru(\mu_{\mathsf{Haar}})$ is $D \cdot P_{\pi} \cdot F_f \cdot C$ for random $\pi,f$ and Haar-random $D$ and $C$. By invariance of the Haar measure, this is exactly the same as outputting a Haar-random unitary. Thus, we have the following corollary of~\cref{lemma:pfd-cho-strong}.
\begin{theorem}[$V$ is indistinguishable from a Haar-random unitary]\label{theorem:haar-cho-strong}
     Let $\Adv$ be a $t$-query oracle adversary that can perform forward and inverse queries. Then
     \begin{align}
        \TD\left(\E_{\calO \gets \mu_{\mathsf{Haar}}} \ketbra*{\Adv_t^{\calO}}_{\gsA \gsB}, \,\,\, \Tr_{\sL \sR}\left( \ketbra*{\Adv^{V}_t}_{\gsA \gsB \gsL \gsR} \right) \right) \leq \frac{9t(t+1)}{N^{1/8}}.
    \end{align}
\end{theorem}
\cref{thm:statistical-PRU-strong} follows from combining \cref{lemma:pfd-cho-strong,theorem:haar-cho-strong} using the triangle inequality. The remainder of this section is devoted to proving~\cref{lemma:pfd-cho-strong}.

\subsection{$V$ is indistinguishable from twirled $W$}

Our first step towards proving~\cref{lemma:pfd-cho-strong} is to prove that an oracle adversary $\Adv$ that makes both forward and inverse queries cannot distinguish whether its query is implemented by the path-recording oracle $V$ (\cref{def:symmetric-V}), or as $D \cdot W \cdot C$ where $C,D \gets \mathfrak{D}$ are sampled from a $2$-design, and $W$ is the \emph{partial} path-recording oracle (\cref{def:symmetric-W}). 

We will require the following definitions. Let $\sC$ and $\sD$ be a pair of registers that  each contain the \emph{description} of an $n$-qubit unitary. These registers will be part of the purification and will not be in the adversary's view.

\begin{definition}
\label{def:init-D-state}
    For any distribution $\frakD$ over $n$-qubit unitaries, define the state
    \begin{align}
    \ket*{\init(\mathfrak{D})}_{\gsC \gsD} \coloneq \int_{C,D} \sqrt{ d\mu_{\mathfrak{D}}(C) d\mu_{\mathfrak{D}}(D)} \ket*{C}_{\gsC} \otimes \ket*{D}_{\gsD},
\end{align}
where $\mu_{\mathfrak{D}}(C)$ is the probability measure for which $C$ is sampled from $\mathfrak{D}$.
\end{definition}

Recall from~\cref{def:multi-rot-Q} that for any pair of $n$-qubit unitaries $C,D$, the operator $Q[C,D]_{\gsL \gsR}$ is defined as
\begin{align}
    Q[C,D] &\coloneqq (C \otimes D^T)^{\otimes *}_{\gsL} \otimes (\overline{C} \otimes D^{\dagger})^{\otimes *}_{\gsR}.
\end{align}

\begin{definition}[Controlled $C,D$ and $Q$]
\label{def:controlled-CDQ}
    Define the following operators
    \begin{align}
        &\mathsf{cC} \coloneq \int_{C} C_{\gsA} \otimes \ketbra*{C}_{\gsC}, \quad \mathsf{cD} \coloneq \int_{D} D_{\gsA} \otimes \ketbra*{D}_{\gsD},\\
    &\mathsf{cQ} \coloneq \int_{C, D} Q[C, D]_{\gsL \darkgray{,} \gsR} \otimes \ketbra*{C}_{\gsC} \otimes \ketbra*{D}_{\gsD}.
\end{align}
\end{definition}

We now state a key lemma that we will need for our proof.

\begin{lemma}[Twirling] \label{lem:twirling-strongPRU}
For any unitary $2$-design $\frakD$, and any integer $0 \leq t \leq N-1$, we have
\begin{align}
    \norm{  \E_{C,D \gets \frakD} (C_{\gsA} \otimes Q[C,D]_{\gsL \gsR})^\dagger \cdot \Big( \Pi^{\bij}_{\leq t, \gsL \gsR} - \Pi^{\calD(W)}_{\leq t, \gsA \gsL \gsR}\Big) \cdot (C_{\gsA} \otimes Q[C,D]_{\gsL \gsR}) }_{\opnorm} &\leq 6t \sqrt{\frac{t}{N}},\\
    \norm{  \E_{C,D \gets \frakD} (D^\dagger_{\gsA} \otimes Q[C,D]_{\gsL \gsR})^\dagger \cdot \Big( \Pi^{\bij}_{\leq t, \gsL \gsR} - \Pi^{\calI(W)}_{\leq t, \gsA \gsL \gsR}\Big) \cdot (D^\dagger_{\gsA} \otimes Q[C,D]_{\gsL \gsR}) }_{\opnorm} &\leq 6t \sqrt{\frac{t}{N}},
\end{align}
\end{lemma}
Note that in the statement of~\cref{lem:twirling-strongPRU}, $\Pi^{\bij}_{\leq t, \gsL \gsR}$ is shorthand for $\Id_{\gsA} \otimes \Pi^{\bij}_{\leq t, \gsL \gsR}$, and thus the operators inside the $\norm{\cdot}_{\opnorm}$ act on $\sA,\sL,\sR$. We prove~\cref{lem:twirling-strongPRU} in~\cref{sec:twirling-strongPRU}.

Next, we define the following adversary states.

\begin{definition}[Twirled-$W$ purification]
\label{def:twirled-W-state}
    Define the states $\ket*{\Adv_i^{W,\frakD}}_{\gsA \gsB \gsL \gsR \gsC \gsD}$ as follows:
    \begin{align}
        \ket*{\Adv_0^{W,\frakD}} &\coloneqq \ket*{0^n}_{\gsA} \ket*{0^m}_{\gsB} \ket*{\varnothing}_{\gsL} \ket*{\varnothing}_{\gsR} \ket*{\mathsf{init}(\mathfrak{D})}_{\gsC \gsD},\\
        \mathrm{For} \quad i =1,\dots,t:  \quad \ket*{\Adv_i^{W,\frakD}} &\coloneqq \Big( (1-b_i) \cdot (\scD \cdot W \cdot \scC) + b_i \cdot (\scD \cdot W \cdot \scC)^\dagger \Big) \cdot A_i \cdot \ket*{\Adv^{W,\frakD}_{i-1}}.
    \end{align}
\end{definition}

For contrast, let us recall the definition of $\ket*{\Adv_i^V}$.

\begin{definition}[$V$ purification]
\label{def:plain-V-0}
    Define the states $\ket*{\Adv_i^V}_{\gsA \gsB \gsL \gsR}$ for $0 \leq i \leq t$ as follows:
    \begin{align}
        \ket*{\Adv_{0}^{V}} &\coloneqq \ket*{0^n}_{\gsA} \ket*{\varnothing}_{\gsL} \ket*{\varnothing}_{\gsR},\\
        \mathrm{For} \quad i =1,\dots,t:  \quad \ket*{\Adv_i^{V}} &\coloneqq \Big( (1-b_i) \cdot V + b_i \cdot V^\dagger \Big) \cdot A_i \cdot \ket*{\Adv^{V}_{i-1}}.
    \end{align}
\end{definition}

Note that because $b_i \in \{0, 1\}$, in the construction of these purified states, one either queries $V$, $\mathsf{cD} \cdot W \cdot \mathsf{cC}$ for $b_i = 0$ or $V^\dagger$, $(\mathsf{cD} \cdot W \cdot \mathsf{cC})^\dagger$ for $b_i = 1$.
Because $W$ and $V$ are partial isometries from \cref{claim:W-partial-isometry} and \cref{claim:V-partial-iso},
$W, W^\dagger, V, V^\dagger$ are all equal to applying a projector followed by a unitary. Hence, $\ket*{\Adv_t^V}$, $\ket*{\Adv_t^{W,\frakD}}$ are both states with norm at most $1$.

\begin{fact}[Norm of the purified states]
\label{fact:properties-of-Adv-t-states-norm}
    For any $t \geq 0$, $\ket*{\Adv_t^V}$, $\ket*{\Adv_t^{W,\frakD}}$ both have norm at most $1$.
\end{fact}

Furthermore, from \cref{def:symmetric-V}, $V$ and $V^\dagger$ take states in the subspace associated with the projector $\Id_{\gsA} \otimes \Pi^{\calR^2}_{\leq i, \gsL \gsR}$ to the the subspace associated with the projector $\Id_{\gsA} \otimes \Pi^{\calR^2}_{\leq i+1, \gsL \gsR}$.
Hence, after $t$ queries in total to $V$ and $V^\dagger$, we have $\ket*{\Adv_t^V}$ is in the image of $\Pi^{\calR^2}_{\leq t}$.
Similarly, from \cref{fact:W-space-leqi}, $W$ and $W^\dagger$ map states in $\Id_{\gsA} \otimes \Pi^{\bij}_{\leq i, \gsL \gsR}$ to $\Id_{\gsA} \otimes \Pi^{\bij}_{\leq i+1, \gsL \gsR}$. Hence, after $t$ queries to $W$ and $W^\dagger$, we have $\ket*{\Adv_t^{W,\frakD}}$ is in the image of $\Pi^{\bij}_{\leq t}$.
We collect these two basic properties in \cref{fact:properties-of-Adv-t-states-spaces}.

\begin{fact}[Spaces that the purified states are in]
\label{fact:properties-of-Adv-t-states-spaces}
    For any $t \geq 0$, we have the following guarantees:
    \begin{itemize}
        \item $\ket*{\Adv_t^V}$ is in the image of $\Pi^{\calR^2}_{\leq t}$.
        \item $\ket*{\Adv_t^{W,\frakD}}$ is in the image of $\Pi^{\bij}_{\leq t}$.
    \end{itemize}
\end{fact}

The main technical claim of this subsection is the following.

\begin{claim}
\label{claim:overlap-twirled-W-Q-plain-V}
    For any integer $t \geq 0$,
    \begin{align}
        \Re \left[ \bra*{\Adv_t^{W,\frakD}}_{\gsA \gsB \gsL \gsR \gsC \gsD} \cdot \scQ_{\gsL \gsR \gsC \gsD} \cdot \Big(\ket*{\Adv_t^V}_{\gsA \gsB \gsL \gsR} \ket*{\init(\frakD)}_{\gsC \gsD}\Big) \right] \geq 1 - \frac{35 t^2}{N^{1/4}} \label{eq:overlap-twirled-W-Q-plain-V}
    \end{align}
\end{claim}

\begin{proof}[Proof of~\cref{claim:overlap-twirled-W-Q-plain-V}]
    We prove this claim by induction. When $t = 0$, we have
    \begin{align}
        \scQ_{\gsL \gsR \gsC \gsD} \cdot \Big(\ket*{\Adv_0^V}_{\gsA \gsB \gsL \gsR} \ket*{\init(\frakD)}_{\gsC \gsD} \Big) &= \scQ_{\gsL \gsR \gsC \gsD} \cdot \Big( \ket*{0^n}_{\gsA} \ket*{\varnothing}_{\gsL} \ket*{\varnothing}_{\gsR} \ket*{\init(\frakD)}_{\gsC \gsD}\Big)\\
        &= \ket*{0^n}_{\gsA} \ket*{\varnothing}_{\gsL} \ket*{\varnothing}_{\gsR} \ket*{\init(\frakD)}_{\gsC \gsD}\\
        &= \ket*{A_0^{W,\frakD}}_{\gsA \gsB \gsL \gsR \gsC \gsD}, 
    \end{align}
    where the first equality is by the definition of $\ket*{\Adv_0^V}$ (\cref{def:plain-V-0}), the second is because $\scQ$ acts as identity on $\ket*{\varnothing}_{\gsL} \ket*{\varnothing}_{\gsR} \ket*{\init(\frakD)}_{\gsL \gsR}$, and the third equality is the definition of $\ket*{\Adv_0^{W,\frakD}}$ (\cref{def:twirled-W-state}). This implies that
    \begin{align}
        \Re \left[ \bra*{\Adv_0^{W,\frakD}}_{\gsA \gsB \gsL \gsR \gsC \gsD} \cdot \scQ_{\gsL \gsR \gsC \gsD} \cdot \Big(\ket*{\Adv_0^V}_{\gsA \gsB \gsL \gsR} \ket*{\init(\frakD)}_{\gsC \gsD}\Big)\right] = 1,
    \end{align}
    so the base case holds.

    For the inductive step, assume that
    \begin{align}
        \Re \left[ \bra*{\Adv_t^{W,\frakD}}_{\gsA \gsB \gsL \gsR \gsC \gsD} \cdot \scQ_{\gsL \gsR \gsC \gsD} \cdot \Big(\ket*{\Adv_t^V}_{\gsA \gsB \gsL \gsR} \ket*{\init(\frakD)}_{\gsC \gsD}\Big) \right] \geq 1 - \frac{35t^2}{N^{1/4}}
    \end{align}
    for some integer $t \geq 0$. We will prove that the claim holds for $t + 1$. To simplify notation, let us assume that the adversary makes a forward query at step $t+1$, i.e., $b_{t+1} = 0$; this is without loss of generality because the argument is symmetric if the adversary makes an inverse query at step $t+1$. We have
    \begin{align}
        \ket*{\Adv_{t+1}^{W,\frakD}} &= \scD \cdot W \cdot \scC \cdot A_{t+1} \cdot \ket*{\Adv_t^{W,\frakD}}, \\
        \scQ \cdot \Big(\ket*{\Adv_{t+1}^{V}}\ket*{\init(\frakD)} \Big)&= \scQ \cdot V \cdot A_{t+1} \cdot \ket*{\Adv_t^{V}} \ket*{\init(\frakD)}
    \end{align}
    and thus
    \begin{align}
        & \Re \left[ \bra*{\Adv_{t+1}^{W,\frakD}} \cdot \scQ \cdot \Big(\ket*{\Adv_{t+1}^V} \ket*{\init(\frakD)} \Big) \right] \label{eq:overlap-twirled-W-Q-plain-V-expand-0}\\
        &= \Re \left[ \bra*{\Adv_t^{W,\frakD}} \cdot A_{t+1}^\dagger \cdot \scC^\dagger \cdot W^\dagger \cdot \scD^\dagger \cdot \scQ \cdot V \cdot A_{t+1} \cdot \ket*{\Adv_t^{V}} \ket*{\init(\frakD)} \right] \label{eq:overlap-twirled-W-Q-plain-V-expand-1}
    \end{align}
    By~\cref{fact:properties-of-Adv-t-states-spaces}, the states $\ket*{\Adv_t^{W,\frakD}}$ and $\ket*{\Adv_t^{V}}$ are both in the image of $\Pi_{\leq t}$. Following~\cref{notation:length-restricted-ops}, we write $W_{\leq t} = W \cdot \Pi_{\leq t}$ and $V_{\leq t} = V \cdot \Pi_{\leq t}$. We can then rewrite (\ref{eq:overlap-twirled-W-Q-plain-V-expand-1}) as
    \begin{align}
        (\ref{eq:overlap-twirled-W-Q-plain-V-expand-1}) &= \Re \left[ \bra*{\Adv_t^{W,\frakD}} \cdot A_{t+1}^\dagger \cdot \scC^\dagger \cdot W_{\leq t}^\dagger \cdot \scD^\dagger \cdot \scQ \cdot V_{\leq t} \cdot A_{t+1} \cdot \ket*{\Adv_t^{V}} \ket*{\init(\frakD)} \right] \label{eq:overlap-twirled-W-Q-plain-V-expand-2}
    \end{align}
    Next, we will write $\scQ \cdot V_{\leq t}$ as
    \begin{align}
        \scQ \cdot V_{\leq t} = \scD \cdot V_{\leq t} \cdot \scC \cdot \scQ. + \Big(\scQ \cdot V_{\leq t} - \scD \cdot V_{\leq t} \cdot \scC \cdot \scQ\Big)
    \end{align}
    This allows us to rewrite (\ref{eq:overlap-twirled-W-Q-plain-V-expand-2}) as
    \begin{align}
        &\Re \left[ \bra*{\Adv_t^{W,\frakD}} \cdot A_{t+1}^\dagger \cdot \scC^\dagger \cdot W_{\leq t}^\dagger \cdot V_{\leq t} \cdot \mathsf{cC} \cdot \mathsf{cQ} \cdot A_{t+1} \cdot \ket*{\Adv_t^{V}} \ket*{\init(\frakD)} \right] \nonumber\\
        & \quad + \Re \left[ \bra*{\Adv_t^{W,\frakD}} \cdot A_{t+1}^\dagger \cdot \scC^\dagger \cdot W_{\leq t}^\dagger \cdot \scD^\dagger \cdot \Big(\scQ \cdot V_{\leq t} - \scD \cdot V_{\leq t} \cdot \scC \cdot \scQ\Big) \cdot A_{t+1} \cdot \ket*{\Adv_t^{V}} \ket*{\init(\frakD)} \right] \label{eq:overlap-twirled-W-Q-plain-V-expand-3}.
    \end{align}
    We can lower bound the second term in the sum as follows. We know that $A_{t+1} \cdot \ket*{\Adv_t^{V}} \ket*{\init(\frakD)}$ and $\bra*{\Adv_t^{W,\frakD}} \cdot A_{t+1}^\dagger \cdot \scC^\dagger \cdot W_{\leq t}^\dagger \cdot \scD^\dagger$ have at most unit norm by~\cref{fact:properties-of-Adv-t-states-norm} and the fact that $A_{t+1}, \scC, \scD, W_{\leq t}^\dagger$ all have operator norm at most $1$ (since $W_{\leq t}^\dagger = (W \cdot \Pi_{\leq t})^\dagger$ and $W$ is a partial isometry by~\cref{claim:W-partial-isometry}). Then by~\cref{claim:two-sided-invariance}, the second term can be lower bounded by
    \begin{align}
        &-\norm{\Big( \mathsf{cD} \cdot V_{\leq t} \cdot \mathsf{cC} \cdot \mathsf{cQ} - \mathsf{cQ} \cdot V_{\leq t} \Big)}_{\opnorm} \\
        &-\norm{\sum_{C,D} \Big( D_{\gsA} \cdot V_{\leq t} \cdot C_{\gsA} \otimes Q[C,D]_{\gsL \gsR} - Q[C,D]_{\gsL \gsR} \cdot V_{\leq t} \Big) \otimes \ketbra*{C,D}}_{\opnorm} \\
        &- \max_{C,D} \norm{ D_{\gsA} \cdot V_{\leq t} \cdot C_{\gsA} \otimes Q[C,D]_{\gsL \gsR} - Q[C,D]_{\gsL \gsR} \cdot V_{\leq t}}_{\opnorm} \\
        &\geq -16 \sqrt{\frac{2 t (t+1)}{N}} \tag{by~\cref{claim:two-sided-invariance}}.\label{eq:invoke-op-norm-bound-invariance}
    \end{align}
    Combining this bound with the sequence of equalities $(\ref{eq:overlap-twirled-W-Q-plain-V-expand-0}) = (\ref{eq:overlap-twirled-W-Q-plain-V-expand-1}) = (\ref{eq:overlap-twirled-W-Q-plain-V-expand-2}) = (\ref{eq:overlap-twirled-W-Q-plain-V-expand-3})$, we get
    \begin{align}
        &\Re \left[ \bra*{\Adv_{t+1}^{W,\frakD}} \cdot \scQ \cdot \Big(\ket*{\Adv_{t+1}^V} \ket*{\init(\frakD)} \Big) \right] \\
        &\geq \underbrace{\Re \left[ \bra*{\Adv_t^{W,\frakD}} \cdot A_{t+1}^\dagger \cdot \scC^\dagger \cdot W_{\leq t}^\dagger \cdot V_{\leq t} \cdot \mathsf{cC} \cdot \mathsf{cQ} \cdot A_{t+1} \cdot \ket*{\Adv_t^{V}} \ket*{\init(\frakD)} \right]}_{\coloneqq \gamma_t} - 16 \sqrt{\frac{2 t (t+1)}{N}} \label{eq:overlap-twirled-W-Q-plain-V-expand-4}.
    \end{align}
    Next we can use properties of the $W$ and $V$ operators to rewrite 
    \begin{align}
        W_{\leq t}^\dagger \cdot V_{\leq t} &= \Big( W \cdot \Pi_{\leq t} \Big)^\dagger \cdot V \cdot \Pi_{\leq t}\\
        &= \Pi_{\leq t} \cdot W^\dagger \cdot V \cdot \Pi_{\leq t}\\
        &= \Pi_{\leq t} \cdot \Pi^{\calD(W)}  \cdot \Pi_{\leq t} \tag{by~\cref{claim:relate-W-V-projector}}\\
        & = \Pi_{\leq t} \cdot \Big(  \Pi^{\bij} - (\Pi^{\bij} - \Pi^{\calD(W)} ) \Big) \cdot \Pi_{\leq t} \\
        &= \Pi^{\bij}_{\leq t} - \Big( \Pi^{\bij}_{\leq t} - \Pi^{\calD(W)}_{\leq t}\Big) \tag{\cref{def:bij-proj} and \cref{claim:pileqt-commutes-with-manything}}
    \end{align}
    Plugging this into $\gamma_t$, we get
    \begin{align}
        \gamma_t &= \underbrace{\Re \left[ \bra*{\Adv_t^{W,\frakD}} \cdot A_{t+1}^\dagger \cdot \scC^\dagger \cdot \Pi^{\bij}_{\leq t} \cdot \mathsf{cC} \cdot \mathsf{cQ} \cdot A_{t+1} \cdot \ket*{\Adv_t^{V}} \ket*{\init(\frakD)} \right]}_{\coloneqq \alpha_t}\\
        & \quad - \underbrace{\Re \left[ \bra*{\Adv_t^{W,\frakD}} \cdot A_{t+1}^\dagger \cdot \scC^\dagger \cdot \Big( \Pi^{\bij}_{\leq t} - \Pi^{\calD(W)}_{\leq t}\Big) \cdot \mathsf{cC} \cdot \mathsf{cQ} \cdot A_{t+1} \cdot \ket*{\Adv_t^{V}} \ket*{\init(\frakD)} \right]}_{\coloneqq \beta_t}
    \end{align}
    \paragraph{Bounding $\alpha_t$.} Observe that 
    \begin{align}
        &\bra*{\Adv_t^{W,\frakD}} \cdot A_{t+1}^\dagger \cdot \scC^\dagger \cdot \Pi^{\bij}_{\leq t}\\
        &= \bra*{\Adv_t^{W,\frakD}} \cdot \Pi^{\bij}_{\leq t} \cdot A_{t+1}^\dagger \cdot \scC^\dagger \tag{ $\Pi^{\bij}_{\leq t,\gsL \gsR}$ commutes with $(A_{t+1}^\dagger \cdot \scC^\dagger)_{\gsA}$}\\
        &= \bra*{\Adv_t^{W,\frakD}} \cdot A_{t+1}^\dagger \cdot \scC^\dagger \tag{by \cref{fact:properties-of-Adv-t-states-spaces}}.
    \end{align}
    Thus,
    \begin{align}
        \alpha_t = \Re \left[ \bra*{\Adv_t^{W,\frakD}} \cdot \mathsf{cQ} \cdot \ket*{\Adv_t^{V}} \ket*{\init(\frakD)} \right] \geq 1 - \frac{35 t^2}{N^{1/4}},
    \end{align}
    by the inductive hypothesis.
    \paragraph{Bounding $\beta_t$.} We will lower bound $- \beta_t$ by upper bounding $\beta_t$:
    \begin{align}
        \beta_t &\leq \abs{\bra*{\Adv_t^{W,\frakD}} \cdot A_{t+1}^\dagger \cdot \scC^\dagger \cdot \Big( \Pi^{\bij}_{\leq t} - \Pi^{\calD(W)}_{\leq t}\Big) \cdot \mathsf{cC} \cdot \mathsf{cQ} \cdot A_{t+1} \cdot \ket*{\Adv_t^{V}} \ket*{\init(\frakD)}}\\
        & \leq \max_{\substack{\ket*{u}\in \calH_{\gsA \gsB \gsL \gsR \gsC \gsD}: \norm{\ket*{u}}_2 \leq 1 \\ \ket*{v}\in \calH_{\gsA \gsB \gsL \gsR}: \norm{\ket*{v}}_2 \leq 1}} \abs{\bra*{u} \cdot \Big( \Pi^{\bij}_{\leq t} - \Pi^{\calD(W)}_{\leq t}\Big) \cdot \mathsf{cC} \cdot \mathsf{cQ} \cdot \ket*{v} \ket*{\init(\frakD)}},\\
        &= \Bigg(\max_{\substack{\ket*{v}\in \calH_{\gsA \gsB \gsL \gsR}: \\ \norm{\ket*{v}}_2 \leq 1}} \bra*{v} \bra*{\init(\frakD)} \cdot \scQ^\dagger \cdot \scC^\dagger \cdot \Big( \Pi^{\bij}_{\leq t} - \Pi^{\calD(W)}_{\leq t}\Big) \cdot \mathsf{cC} \cdot \mathsf{cQ} \cdot \ket*{v} \ket*{\init(\frakD)} \Bigg)^{1/2}\\
        &= \norm{  \E_{C,D \gets \frakD} (C_{\gsA} \otimes Q[C,D]_{\gsL \gsR})^\dagger \cdot \Big(\Pi^{\bij}_{\leq t} - \Pi^{\calD(W)}_{\leq t}\Big) \cdot (C_{\gsA} \otimes Q[C,D]_{\gsL \gsR}) }_{\opnorm}^{1/2}\\
        &\leq \Big(6t \sqrt{\frac{t}{N}}\Big)^{1/2} \leq \frac{3 t^{3/4}}{N^{1/4}}
    \end{align}
    where:
    \begin{itemize}
        \item the first inequality uses the fact that $\Re(z) \leq \abs{z}$, 
        \item the second inequality holds because $\scC \cdot A_{t+1} \cdot \ket*{\Adv_t^{W,\frakD}} \in \calH_{\gsA \gsB \gsL \gsR \gsC \gsD}$ and $A_{t+1} \cdot \ket*{\Adv_t^{V}} \in \calH_{\gsA \gsB \gsL \gsR}$ both have at most unit norm,
        \item the third line uses the fact that
        \begin{align}
            \max_{\substack{\ket*{u}: \norm{\ket*{u}}_2 \leq 1, \\ \ket*{v}: \norm{\ket*{v}}_2 \leq 1}} \abs{\bra*{u} \cdot M \cdot \ket*{v}} = \Big( \max_{\ket*{v}: \norm{\ket*{v}}_2 \leq 1} \bra*{v}\cdot M^\dagger \cdot M \cdot \ket*{v} \Big)^{1/2},
        \end{align}
    and the fact that $\Big(\Pi^{\bij}_{\leq t} - \Pi^{\calD(W)}_{\leq t}\Big)^\dagger \cdot \Big(\Pi^{\bij}_{\leq t} - \Pi^{\calD(W)}_{\leq t}\Big) = \Pi^{\bij}_{\leq t} - \Pi^{\calD(W)}_{\leq t}$, since $\Pi^{\bij}_{\leq t} - \Pi^{\calD(W)}_{\leq t}$ is a projector.\footnote{By~\cref{fact:W-space-leqi}, $\Pi^{\bij} - \Pi^{\calD(W)}$ is a projector. By~\cref{claim:pileqt-commutes-with-manything}, $\Pi^{\calD(W)}$ commutes with $\Pi_{\leq t}$ and by~\cref{claim:pileqt-commutes-with-manything}, $\Pi^{\bij}$ commutes with $\Pi_{\leq t}$ by~\cref{def:bij-proj}. Recall the fact that if $\Pi_1$ and $\Pi_2$ are projectors such that $[\Pi_1, \Pi_2] = 0$, then $\Pi_1 \cdot \Pi_2$ is a projector. Thus, since $\Pi^{\bij}_{\leq t} - \Pi^{\calD(W)}_{\leq t} = (\Pi^{\bij} - \Pi^{\calD(W)}) \cdot \Pi_{\leq t}$, we have that $\Pi^{\bij}_{\leq t} - \Pi^{\calD(W)}_{\leq t}$ is a projector.}
    \item the fourth line follows from the definitions of $\ket*{\init(\frakD)}, \scC, \scQ$ (\cref{def:init-D-state,def:controlled-CDQ}), 
    \item and the last line follows from~\cref{lem:twirling-strongPRU}.
    \end{itemize} 
    Note that in the fourth line, we can drop the $\sB$ register since the operator inside the $\norm{\cdot}_{\opnorm}$ acts as identity on $\sB$.
    Putting everything together, we have
    \begin{align}
        \Re \left[ \bra*{\Adv_{t+1}^{W,\frakD}} \cdot \scQ \cdot \Big(\ket*{\Adv_{t+1}^V} \ket*{\init(\frakD)} \Big) \right] &\geq \alpha_t - \beta_t - 16 \sqrt{\frac{2t(t+1)}{N}}\\
        &\geq 1- \frac{35t^2}{N^{1/4}} - \frac{3 t^{3/4}}{N^{1/4}} - 16\sqrt{\frac{2t(t+1)}{N}}\\
        &\geq 1 - \frac{1}{N^{1/4}} \cdot \Big(35 t^2 + 3t^{3/4} + 16\cdot \frac{2t}{N^{1/4}} \Big)\\
        &\geq 1 - \frac{1}{N^{1/4}} \cdot \Big(35t^2 + 35t\Big) \\
        & \geq 1 - \frac{35 (t+1)^2}{N^{1/4}},
    \end{align}
    which establishes the claim for $t +1$. This concludes the proof.
\end{proof}

\begin{lemma} \label{lem:closeness-AWD-and-PhiVt}
For any $0 \leq t <N$ and any unitary $2$-design $\frakD$, we have 
    \begin{align}
        \TD( \Tr_{- \sA \sB} \ketbra*{\Adv^{W, \mathfrak{D}}_t}_{\gsA \gsB \gsL \gsR \gsC \gsD}, \Tr_{- \sA \sB} \ketbra*{\Adv^{V}_t}_{\gsA \gsB \gsL \gsR} ) \leq \frac{9t}{N^{1/8}}. \label{eq:TD-phi-psi}
    \end{align}
\end{lemma}

\begin{proof}
Using the fact that $\ket*{\Adv^{W, \mathfrak{D}}_t}_{\gsA \gsB \gsL \gsR \gsC \gsD}$ and $\ket*{\Adv_t^V}_{\gsA \gsB \gsL \gsR}$ are subnormalized states from \cref{fact:properties-of-Adv-t-states-norm} and that $\TD( \ketbra*{u}, \ketbra*{v} ) \leq \norm{ \ket*{u} - \ket*{v} }_2$ for subnormalized states $\ket*{u}, \ket*{v}$, we have
\begin{align}
    & \TD\left( \ketbra*{\Adv^{W, \mathfrak{D}}_t}_{\gsA \gsB \gsL \gsR \gsC \gsD}, \scQ_{\gsL \gsR \gsC \gsD} \cdot \Big(\ketbra*{\Adv_t^V}_{\gsA \gsB \gsL \gsR} \otimes \ketbra*{\init(\frakD)}_{\gsC \gsD}\Big) \cdot \scQ_{\gsL \gsR \gsC \gsD}^\dagger \right)^2 \\
    &\leq \norm{\ket*{\Adv^{W, \mathfrak{D}}_t}_{\gsA \gsB \gsL \gsR} - \scQ_{\gsL \gsR \gsC \gsD} \cdot \Big(\ket*{\Adv_t^V}_{\gsA \gsB \gsL \gsR} \otimes \ket*{\init(\frakD)}_{\gsC \gsD}\Big) }_2^2\\
    &= \braket*{\Adv^{W, \mathfrak{D}}_t} + \braket*{\Adv^V_t} - 2 \Re\left[ \bra*{\Adv^{W, \mathfrak{D}}_t}_{\gsA \gsB \gsL \gsR \gsC \gsD} \cdot \scQ_{\gsL \gsR \gsC \gsD} \cdot \Big(\ket*{\Adv_t^V}_{\gsA \gsB \gsL \gsR} \otimes \ket*{\init(\frakD)}_{\gsC \gsD}\Big) \right]\\
    &\leq 2 - 2 \cdot (1 - \frac{35 t^2}{N^{1/4}}) = \frac{70 t^2}{N^{1/4}}. \tag{Using \cref{claim:overlap-twirled-W-Q-plain-V}}
\end{align}
Therefore, using the fact that $\mathsf{cQ}_{\gsL \gsR \gsC \gsD}$ acts only on $\sL, \sR, \sC, \sD$ and $\mathsf{cQ}_{\gsL \gsR \gsC \gsD}$ is a unitary, we obtain
\begin{align}
    &\TD( \Tr_{- \sA \sB} \ketbra*{\Adv^{W, \mathfrak{D}}_t}_{\gsA \gsB \gsL \gsR \gsC \gsD}, \Tr_{- \sA \sB} \ketbra*{\Adv^{V}_t}_{\gsA \gsB \gsL \gsR} )\\
    &= \TD\Big( \Tr_{-\sA \sB} \ketbra*{\Adv^{W, \mathfrak{D}}_t}_{\gsA \gsB \gsL \gsR \gsC \gsD},\nonumber \\
    &\quad \quad \quad \quad \quad \Tr_{-\sA \sB} \left[ \scQ_{\gsL \gsR \gsC \gsD} \cdot \Big(\ketbra*{\Adv_t^V}_{\gsA \gsB \gsL \gsR} \otimes \ketbra*{\init(\frakD)}_{\gsC \gsD}\Big) \cdot \scQ_{\gsL \gsR \gsC \gsD}^\dagger \right] \Big)\\
    &\leq \TD \left( \ketbra*{\Adv^{W, \mathfrak{D}}_t}_{\gsA \gsB \gsL \gsR \gsC \gsD}, \scQ_{\gsL \gsR \gsC \gsD} \cdot \Big(\ketbra*{\Adv_t^V}_{\gsA \gsB \gsL \gsR} \otimes \ketbra*{\init(\frakD)}_{\gsC \gsD}\Big) \cdot \scQ_{\gsL \gsR \gsC \gsD}^\dagger \right)\\
    &\leq \sqrt{\frac{70 t^2}{N^{1/4}}} \leq \frac{9t}{N^{1/8}}.
\end{align}
This completes the proof.
\end{proof}

\subsection{Twirled $W$ and twirled $\spfo$ are indistinguishable}

Let $\ket*{+_{N!}}_{\gsP}$ and $\ket*{+_{3^N}}_{\gsF}$ denote the uniform superposition over all permutations and functions, respectively. 
We define the follow state obtained by querying twirled $\spfo$.

\begin{definition}[Twirled $\spfo$ purification]
    Let
    \begin{align}
        \ket*{\Adv_{0}^{\spfo,\frakD}}_{\gsA \gsB \gsP \gsF \gsC \gsD} \coloneqq \ket*{0^n}_{\gsA} \ket*{0^m}_{\gsB} \ket*{+_{N!}}_{\gsP} \ket*{+_{3^N}}_{\gsF} \ket*{\mathsf{init}(\mathfrak{D})}_{\gsC \gsD},
    \end{align}
    For $1\leq i \leq t$, define
    \begin{align}
        \ket*{\Adv_i^{\spfo,\frakD}} \coloneqq \Big( (1-b_i) \cdot (\scD \cdot \spfo \cdot \scC) + b_i \cdot (\scD \cdot \spfo \cdot \scC)^\dagger \Big) \cdot A_i \cdot \ket*{\Adv^{\spfo,\frakD}_{i-1}}.
    \end{align}
\end{definition}

To connect twirled $W$ and twirled $\spfo$, we need to define the following projections.

\begin{definition}
Define the projectors
\begin{align}
    \widetilde{\Pi}^{\calD(W)} &\coloneqq \Compress^\dagger \cdot \Pi^{\calD(W)} \cdot \Compress,\\
    \widetilde{\Pi}^{\calI(W)} &\coloneqq \Compress^\dagger \cdot \Pi^{\calI(W)} \cdot \Compress.
\end{align}
\end{definition}

We define the following state obtained by querying twirled $\spfo$, but depending on whether forward or inverse query (determined by $b_i$) is made, we will add a projector.

\begin{definition}[Twirled projected $\spfo$ purification]
    Let $\ket*{\Adv_0^{\widetilde{\spfo},\frakD}} \coloneqq \ket*{\Adv_0^{\spfo,\frakD}}$. For $1 \leq i \leq t$, define
    \begin{align}
        \ket*{\Adv_i^{\widetilde{\spfo},\frakD}} \coloneqq \Big( (1-b_i) \cdot (\scD \cdot \spfo \cdot \widetilde{\Pi}^{\calD(W)} \cdot \scC) + b_i \cdot (\scC^\dagger \cdot \spfo^\dagger \cdot \widetilde{\Pi}^{\calI(W)} \cdot \scD^\dagger) \Big) \cdot A_i \cdot \ket*{\Adv^{\widetilde{\spfo},\frakD}_{i-1}}.
    \end{align}
\end{definition}

\begin{claim}
\label{claim:compress-main-proof}
    For all integers $0 \leq t \leq N$,
    \begin{align}
        \ket*{\Adv_t^{W,\frakD}}_{\gsA \gsB \gsL \gsR \gsC \gsD} = \Compress \cdot \ket*{\Adv_t^{\widetilde{\spfo},\frakD}}_{\gsA \gsB \gsP \gsF \gsC \gsD}.
    \end{align}
\end{claim}
\begin{proof}
    We prove this using induction. The base case $t = 0$ follows from the fact that
    \begin{equation}
        \Compress \cdot \ket*{+_{N!}}_{\gsP} \ket*{+_{3^N}}_{\gsF} = \ket*{\varnothing}_{\gsL} \ket*{\varnothing}_{\gsR}.
    \end{equation}
    If $\ket*{\Adv_t^{W,\frakD}}_{\gsA \gsB \gsL \gsR \gsC \gsD} = \Compress \cdot \ket*{\Adv_t^{\widetilde{\spfo},\frakD}}_{\gsA \gsB \gsP \gsF \gsC \gsD}$ for $t > 0$, then we have
    \begin{align}
        &\ket*{\Adv_{t+1}^{W,\frakD}}_{\gsA \gsB \gsL \gsR \gsC \gsD} \\
        &= \Big( (1-b_i) \cdot (\scD \cdot W \cdot \scC) + b_i \cdot (\scD \cdot W \cdot \scC)^\dagger \Big) \cdot A_i \cdot \ket*{\Adv^{W,\frakD}_{t}}_{\gsA \gsB \gsL \gsR \gsC \gsD} \\
        &= \Big( (1-b_i) \cdot (\scD \cdot \Compress \cdot \spfo \cdot \Compress^\dagger \cdot \Pi^{\calD(W)} \cdot \scC) \nonumber \\
        &\,\,\,+ b_i \cdot (\scC^\dagger \cdot \Compress \cdot \spfo^\dagger \cdot \Compress^\dagger \cdot \Pi^{\calI(W)} \cdot \scD^\dagger) \Big) \cdot A_i \cdot \ket*{\Adv^{W,\frakD}_{t}}_{\gsA \gsB \gsL \gsR \gsC \gsD} \tag{Using \cref{claim:relate-W-and-spfo}}\\
        &= \Compress \cdot \Big( (1-b_i) \cdot (\scD \cdot \spfo \cdot \widetilde{\Pi}^{\calD(W)} \cdot \scC) \nonumber \\
        &\quad\,\,\,+ b_i \cdot \scC^\dagger \cdot\spfo^\dagger \cdot \widetilde{\Pi}^{\calI(W)} \cdot \scD^\dagger \Big) \cdot A_i \cdot \Compress^\dagger \cdot \ket*{\Adv^{W,\frakD}_{t}}_{\gsA \gsB \gsL \gsR \gsC \gsD}\\
        & = \Compress \cdot \Big( (1-b_i) \cdot (\scD \cdot \spfo \cdot \widetilde{\Pi}^{\calD(W)} \cdot \scC) \nonumber \\
        &\quad \,\,\, + b_i \cdot \scC^\dagger \cdot\spfo^\dagger \cdot \widetilde{\Pi}^{\calI(W)} \cdot \scD^\dagger \Big) \cdot A_i \cdot \ket*{\Adv^{\widetilde{\spfo},\frakD}_{t}}_{\gsA \gsB \gsP \gsF \gsC \gsD} \tag{inductive hypothesis} \\
        &= \Compress \cdot \ket*{\Adv^{\widetilde{\spfo},\frakD}_{t+1}}_{\gsA \gsB \gsP \gsF \gsC \gsD}.
    \end{align}
    This concludes the proof.
\end{proof}

\begin{lemma}[Norm bound] \label{lem:norm-close-TD}
For any $0 \leq t <N$ and any unitary $2$-design $\frakD$, we have 
\begin{align}
    1 \geq \braket*{\Adv^{W, \mathfrak{D}}_t}_{\gsA \gsB \gsL \gsR \gsC \gsD} \geq 1 - \frac{70t^2}{N^{1/4}}.
\end{align}
\end{lemma}
\begin{proof}
We can utilize the following bounds,
\begin{align}
    &\braket*{\Adv^{W, \mathfrak{D}}_t}_{\gsA \gsB \gsL \gsR \gsC \gsD}\\ 
    &\geq \braket*{\Adv^{W, \mathfrak{D}}_t}_{\gsA \gsB \gsL \gsR \gsC \gsD} \cdot \braket*{\Adv_t^V}_{\gsA \gsB \gsL \gsR} \tag{$ \braket*{\Adv_t^V}_{\gsA \gsB \gsL \gsR} \leq 1$ from \cref{fact:properties-of-Adv-t-states-norm}} \\
    &= \braket*{\Adv^{W, \mathfrak{D}}_t}_{\gsA \gsB \gsL \gsR \gsC \gsD} \cdot \Big(\bra*{\Adv_t^V}_{\gsA \gsB \gsL \gsR} \bra*{\init(\frakD)}_{\gsC \gsD}\Big) \cdot \scQ^\dagger \cdot \scQ_{\gsL \gsR \gsC \gsD} \cdot \Big(\ket*{\Adv_t^V}_{\gsA \gsB \gsL \gsR} \ket*{\init(\frakD)}_{\gsC \gsD}\Big) \\
    &\geq \left| \bra*{\Adv^{W, \mathfrak{D}}_t}_{\gsA \gsB \gsL \gsR \gsC \gsD} \cdot \scQ_{\gsL \gsR \gsC \gsD} \cdot \Big(\ket*{\Adv_t^V}_{\gsA \gsB \gsL \gsR} \ket*{\init(\frakD)}_{\gsC \gsD}\Big) \right|^2 \tag{Cauchy-Schwarz inequality}\\
    &\geq \Re\left[ \bra*{\Adv^{W, \mathfrak{D}}_t}_{\gsA \gsB \gsL \gsR \gsC \gsD} \cdot \scQ_{\gsL \gsR \gsC \gsD} \cdot \Big(\ket*{\Adv_t^V}_{\gsA \gsB \gsL \gsR} \ket*{\init(\frakD)}_{\gsC \gsD}\Big) \right]^2\\
    &\geq (1- \frac{35t^2}{N^{1/4}})^2 \geq 1 - \frac{70t^2}{N^{1/4}} \tag{Using \cref{claim:overlap-twirled-W-Q-plain-V}},
\end{align}
which completes the proof.
\end{proof}

\begin{lemma} \label{lem:pfo-W-closeness}
For all integers $0 \leq t \leq N$,
    \begin{equation}
        \TD\left( \Tr_{-\sA \sB} \ketbra*{\Adv_t^{\spfo,\frakD}}_{\gsA \gsB \gsP \gsF \gsC \gsD}, \Tr_{-\sA \sB} \ketbra*{\Adv_t^{W,\frakD}}_{\gsA \gsB \gsL \gsR \gsC \gsD} \right) \leq \frac{9t^2}{N^{1/8}}.
    \end{equation}
\end{lemma}
\begin{proof}
    Because $\Compress$ acts on registers $\sP, \sF$ and maps to $\sL, \sR$, we have
    \begin{equation}
        \Tr_{-\sA \sB} \ketbra*{\Adv_t^{W,\frakD}}_{\gsA \gsB \gsL \gsR \gsC \gsD} = \Tr_{-\sA \sB} \ketbra*{\Adv_t^{\widetilde{\spfo},\frakD}}_{\gsA \gsB \gsP \gsF \gsC \gsD}.
    \end{equation}
    Because $\spfo$ is an isometry, $\ket*{\Adv_t^{\spfo,\frakD}}_{\gsA \gsB \gsP \gsF \gsC \gsD}$ has norm $1$.
    Furthermore, from \cref{claim:compress-main-proof}, because $\Compress$ is an isometry, we have
    \begin{equation}
        \braket*{\Adv_t^{\widetilde{\spfo},\frakD}}_{\gsA \gsB \gsP \gsF \gsC \gsD} = \braket{\Adv^{W, \mathfrak{D}}_t}_{\gsA \gsB \gsL \gsR \gsC \gsD} \leq 1.
    \end{equation}
    Together, we can obtain the following,
    \begin{align}
        &\TD\left( \Tr_{-\sA \sB} \ketbra*{\Adv_t^{\spfo,\frakD}}_{\gsA \gsB \gsP \gsF \gsC \gsD}, \Tr_{-\sA \sB} \ketbra*{\Adv_t^{W,\frakD}}_{\gsA \gsB \gsL \gsR \gsC \gsD} \right)\\
        &= \TD\left( \Tr_{-\sA \sB} \ketbra*{\Adv_t^{\spfo,\frakD}}_{\gsA \gsB \gsP \gsF \gsC \gsD}, \Tr_{-\sA \sB} \ketbra*{\Adv_t^{\widetilde{\spfo},\frakD}}_{\gsA \gsB \gsP \gsF \gsC \gsD} \right)\\
        &\leq \TD\left( \ketbra*{\Adv_t^{\spfo,\frakD}}_{\gsA \gsB \gsP \gsF \gsC \gsD}, \ketbra*{\Adv_t^{\widetilde{\spfo},\frakD}}_{\gsA \gsB \gsP \gsF \gsC \gsD} \right)\\
        &\leq \norm{\ket*{\Adv_t^{\spfo,\frakD}}_{\gsA \gsB \gsP \gsF \gsC \gsD} - \ket*{\Adv_t^{\widetilde{\spfo},\frakD}}_{\gsA \gsB \gsP \gsF \gsC \gsD}}_2 \tag{$\frac{1}{2}\norm{u u^\dagger - v v^\dagger}_{\mathrm{tr}} \leq \norm{u - v}_{2}$ if $\norm{u}_2, \norm{v}_2 \leq 1$}\\
        &\leq t \cdot \sqrt{1 - \braket{\Adv_t^{\widetilde{\spfo},\frakD}}_{\gsA \gsB \gsP \gsF \gsC \gsD}} \tag{\cref{lem:seq-gentleM-pure} on sequential gentle measurement}\\
        &= t \cdot \sqrt{1 - \braket{\Adv_t^{W,\frakD}}_{\gsA \gsB \gsL \gsR \gsC \gsD}} \tag{\cref{claim:compress-main-proof}}\\
        &\leq t \cdot \sqrt{\frac{70t^2}{N^{1/4}}} \leq \frac{9t^2}{N^{1/8}} \tag{\cref{lem:norm-close-TD}}.
    \end{align}
    This concludes the proof.
\end{proof}

\subsection{Proof of~\cref{lemma:pfd-cho-strong}}

From \cref{claim:purified-vs-standard-ternary-PFO}, we have
\begin{equation}
    \Tr_{\sP \sF \sC \sD} \ketbra*{\Adv_t^{\spfo,\frakD}}_{\gsA \gsB \gsP \gsF \gsC \gsD} = \E_{\calO \gets \mathsf{sPRU}(\frakD)} \ketbra*{\Adv_t^{\calO}}_{\gsA \gsB}.
\end{equation}
From \cref{lem:closeness-AWD-and-PhiVt}, we have
\begin{align}
    \TD\left( \Tr_{\sL \sR \sC \sD} \ketbra*{\Adv^{W, \mathfrak{D}}_t}_{\gsA \gsB \gsL \gsR \gsC \gsD}, \Tr_{\sL \sR} \ketbra*{\Adv^{V}_t}_{\gsA \gsB \gsL \gsR} \right) \leq \frac{9t}{N^{1/8}}.
\end{align}
From \cref{lem:pfo-W-closeness}, we have
\begin{equation}
    \TD\left( \Tr_{\sP \sF \sC \sD} \ketbra*{\Adv_t^{\spfo,\frakD}}_{\gsA \gsB \gsP \gsF \gsC \gsD}, \Tr_{\sL \sR \sC \sD} \ketbra*{\Adv_t^{W,\frakD}}_{\gsA \gsB \gsL \gsR \gsC \gsD} \right) \leq \frac{9t^2}{N^{1/8}}.
\end{equation}
By triangle inequality, we have
\begin{align}
    &\TD\left( \E_{\calO \gets \mathsf{sPRU}(\frakD)} \ketbra*{\Adv_t^{\calO}}_{\gsA \gsB}, \Tr_{\sL \sR} \ketbra*{\Adv_t^{V}}_{\gsA \gsB \gsL \gsR} \right)\\
    &= \TD\left( \Tr_{\sP \sF \sC \sD} \ketbra*{\Adv_t^{\spfo,\frakD}}_{\gsA \gsB \gsP \gsF \gsC \gsD}, \Tr_{\sL \sR} \ketbra*{\Adv_t^{V}}_{\gsA \gsB \gsL \gsR} \right)\\
    &\leq \TD\left( \Tr_{\sP \sF \sC \sD} \ketbra*{\Adv_t^{\spfo,\frakD}}_{\gsA \gsB \gsP \gsF \gsC \gsD}, \Tr_{\sL \sR \sC \sD} \ketbra*{\Adv_t^{W,\frakD}}_{\gsA \gsB \gsL \gsR \gsC \gsD} \right) \nonumber\\
    &+ \TD\left( \Tr_{\sL \sR \sC \sD} \ketbra*{\Adv^{W, \mathfrak{D}}_t}_{\gsA \gsB \gsL \gsR \gsC \gsD}, \Tr_{\sL \sR} \ketbra*{\Adv^{V}_t}_{\gsA \gsB \gsL \gsR} \right)\\
    & \leq \frac{9t(t+1)}{N^{1/8}}.
\end{align}
This completes the proof of \cref{lemma:pfd-cho-strong}.

\section{Proof of~\cref{claim:two-sided-invariance}}
\label{sec:symmetric-V-approx-invariance}

In this section, we prove~\cref{claim:two-sided-invariance}, which states that the symmetric path recording oracle $V$ is approximately unitary invariant. For convenience, we restate the lemma below:

\begin{lemma}[\cref{claim:two-sided-invariance}, restated] For any $0 \leq t < N$, and any pair of $n$-qubit unitaries $C,D$, we have
\begin{align}
        \norm{D_{\gsA} \cdot V_{\leq t} \cdot C_{\gsA} \otimes Q[C,D]_{\gsL \gsR}
        - Q[C,D]_{\gsL \gsR} \cdot V_{\leq t}}_{\opnorm} & \leq 16\sqrt{\frac{2t(t+1)}{N}},\\
        \norm{C_{\gsA}^\dagger \cdot (V^\dagger)_{\leq t} \cdot D_{\gsA}^\dagger \otimes Q[C,D]_{\gsL \gsR}
        - Q[C,D]_{\gsL \gsR} \cdot (V^\dagger)_{\leq t}}_{\opnorm} & \leq 16\sqrt{\frac{2t(t+1)}{N}},
    \end{align}        
\end{lemma}

To prove this lemma, we will define a pair of operators $E^L$ and $E^R$ that satisfy \emph{exact} unitary invariance. We will then prove that $E^L$ is close in operator norm to $V^L$, and that $E^R$ is close in operator norm to $E^R$. By combining these guarantees, we will show that $V^L$ and $V^R$ satisfy approximate unitary invariance, which we will use to prove that $V$ satisfies approximate unitary invariance.

\subsection{Defining $E^L$ and $E^R$}

\begin{definition} \label{def:ELER-definition}
Define the operator $E^L$ and $E^R$ that act on registers $\sA,\sL,\sR$ as follows:
\begin{align}
    E^L \coloneqq \frac{1}{\sqrt{N}} \sum_{x,y \in [N]} \ketbra*{y}{x}_{\gsA} \otimes \sum_{L \in \calR} \sqrt{\num(L,(x,y)) + 1} \cdot \ketbra*{L \cup \{(x,y)\}}{L}_{\gsL} \otimes \sum_{R \in \calR} \ketbra*{R}_{\gsR}. \label{def:EL-operator-num-form}\\
    E^R \coloneqq \frac{1}{\sqrt{N}} \sum_{x,y \in [N]} \ketbra*{x}{y}_{\gsA} \otimes \sum_{L \in \calR} \ketbra*{L}_{\gsL} \otimes \sum_{R \in \calR} \sqrt{\num(R,(x,y)) + 1} \cdot \ketbra*{R \cup \{(x,y)\}}{R}_{\gsR}. \label{def:ER-operator-num-form}
\end{align}
\end{definition}

We will show that $E^L$ and $E^R$ satisfies the following unitary invariance property. To state the property, recall that we define the operator $Q[C,D]$ as follows:

\begin{definition}[\cref{def:multi-rot-Q}, restated]
    For any pair of $n$-qubit unitaries $C,D$, define
    \begin{align}
        Q[C,D] &\coloneqq (C \otimes D^T)^{\otimes *}_{\gsL} \otimes (\overline{C} \otimes D^{\dagger})^{\otimes *}_{\gsR}.
    \end{align}
\end{definition}

\begin{claim}[Exact unitary invariance of $E^L$ and $E^R$]
\label{claim:E-invariance}
    For any pair of $n$ qubit unitaries $C,D$, we have
    \begin{align}
        D_{\gsA} \cdot E^L_{\gsA \gsL \gsR} \cdot C_{\gsA} &= Q[C,D]_{\gsL \gsR} \cdot E^L_{\gsA \gsL \gsR} \cdot Q[C,D]_{\gsL \gsR}^\dagger \label{eq:EL-invariance},\\
        C^\dagger_{\gsA} \cdot E^{R}_{\gsA \gsL \gsR} \cdot D^\dagger_{\gsA} &= Q[C,D]_{\gsL \gsR} \cdot E^{R}_{\gsA \gsL \gsR} \cdot Q[C,D]_{\gsL \gsR}^\dagger \label{eq:ER-invariance},
    \end{align}
\end{claim}

To prove~\cref{claim:E-invariance}, it will be useful to have the following alternative expressions for $E^L$ and $E^R$.

\begin{claim}[Alternative form of $E^L$ and $E^R$]
\label{claim:alt-form-of-E}
    The $E^L$ operator can also be written as
    \begin{align}
        E^L &= \frac{1}{\sqrt{N}} \sum_{x,y \in [N]} \ketbra*{y}{x}_{\gsA} \otimes \sum_{\ell \geq 0} \Pi^{\calR}_{\ell+1, \gsL} \cdot \Big(\sqrt{\ell+1} \cdot \ket*{x,y} \otimes \Pi_{\ell} \Big)_{\gsL} \otimes \Pi^{\calR}_{\gsR}. \label{def:EL-operator-sym-form}\\
        E^R &= \frac{1}{\sqrt{N}} \sum_{x,y \in [N]} \ketbra*{x}{y}_{\gsA} \otimes \Pi^{\calR}_{\gsL} \otimes \sum_{r \geq 0} \Pi^{\calR}_{r+1, \gsR} \cdot \Big(\sqrt{r+1} \cdot \ket*{x,y} \otimes \Pi_{r} \Big)_{\gsR}. \label{def:ER-operator-sym-form}
    \end{align}
    Here $\Pi_\ell$ denotes the projector onto the span of length-$\ell$ states $\ket*{x_1,y_1,\dots,x_\ell,y_\ell}$, and $(\ket*{x,y} \otimes \Pi_\ell)_{\gsL}$ is the linear operator that maps
    \begin{align}
        (\ket*{x,y} \otimes \Pi_{\ell})_{\gsL} \cdot \ket*{x_1,y_1,\dots,x_\ell,y_\ell} = \ket*{x,y,x_1,y_1,\dots,x_\ell,y_\ell}.
    \end{align}
\end{claim}

\begin{proof}
    We will prove the statement for $E^L$, and the proof for $E^R$ will be symmetric. To establish $(\ref{def:EL-operator-num-form}) = (\ref{def:EL-operator-sym-form})$, we need to prove that for all $(x,y) \in [N]^2$ and $\ell \geq 0$,
    \begin{align}
        \sum_{L \in \calR_\ell} \sqrt{\num(L,(x,y)) +1} \cdot \ketbra*{L \cup \{(x,y)\}}{L}_{\gsL} = \Pi^{\calR}_{\ell+1,\gsL} \cdot \Big(\sqrt{\ell+1} \cdot \ket*{x,y} \otimes \Pi_\ell \Big)_{\gsL}. \label{eq:insert-sym-subspace-relation}
    \end{align}
    Since $\Pi^{\calR}_{\ell+1} = \sum_{R \in \calR_{\ell+1}} \ketbra*{R}$ (\cref{not:pi-R-t-all-t}), we can write the right-hand side of~\cref{eq:insert-sym-subspace-relation} as
    \begin{align}
        &\sum_{L \in \calR_{\ell+1}} \ketbra*{L} \cdot \Big(\sqrt{\ell+1} \cdot \ket*{x,y} \otimes \Pi_{\ell} \Big)_{\gsR}\\
        &= \sum_{L \in \calR_{\ell}} \ketbra*{L \cup \{(x,y)\}} \cdot \Big(\sqrt{\ell+1} \cdot \ket*{x,y} \otimes \Pi_\ell \Big)_{\gsR}.\label{eq:insert-sym-subspace-relation-2}
    \end{align}
    Therefore, we need to prove that for all $\ell \geq 0$ and $L \in \calR_{\ell}$ that
    \begin{align}
        \bra*{L \cup \{(x,y)\}} \cdot \Big(\sqrt{\ell+1} \cdot \ket*{x,y} \otimes \Pi_\ell \Big) = \sqrt{\num(L,(x,y))} \cdot \bra*{L}. \label{eq-bra-R-xy-R} 
    \end{align}
    To see this, note that $\bra*{L \cup \{(x,y)\}}$ is a superposition over all permutations of the elements of $L \cup \{(x,y)\}$, and thus when we right multiply by $\Big(\sqrt{\ell+1} \cdot \ket*{x,y} \otimes \Pi_\ell \Big)$, the resulting state is proportional to $\bra*{L}$. To compute the proportionality constant, note that a
    \begin{align}
        \frac{\binom{\ell-1}{\num(L,(x,y))-1}}{\binom{\ell}{\num(L,(x,y))}} = \frac{\num(L,(x,y))}{\ell}
    \end{align}
    fraction of the permutations of the elements of $L \cup \{(x,y)\}$ will have $(x,y)$ in the left-most slot. Thus, 
    \begin{align}
        \bra*{L \cup \{(x,y)\}} \cdot \Big(\ket*{x,y} \otimes \Pi_\ell \Big) = \frac{\sqrt{\num(L,(x,y))}}{\sqrt{\ell+1}} \cdot \bra*{L},
    \end{align}
    which gives~\cref{eq-bra-R-xy-R} when we multiply by $\sqrt{\ell+1}$.
\end{proof}

We can use~\cref{claim:alt-form-of-E} to prove exact unitary invariance of $E^L$ and $E^R$ (\cref{claim:E-invariance}).

\begin{proof}[Proof of~\cref{claim:E-invariance}]
    To prove \cref{eq:EL-invariance}, it suffices to prove that
    \begin{align}
        D_{\gsA} \cdot E^L_{\gsA \gsL \gsR} \cdot C_{\gsA} \otimes Q[C,D]_{\gsL \gsR} = Q[C,D]_{\gsL \gsR} \cdot E^L_{\gsA \gsL \gsR}.
    \end{align}
    Recall that
    \begin{align}
        Q[C,D]_{\gsL \gsR} = (C \otimes D^T)^{\otimes *}_{\gsL} \otimes (\overline{C} \otimes D^{\dagger})^{\otimes *}_{\gsR}.
    \end{align}
    Expanding the left-hand-side using the definition of $Q[C,D]$ and the expression for $E$ given by~\cref{claim:alt-form-of-E}, we have
    \begin{align}
        &D_{\gsA} \cdot E^L_{\gsA \gsL \gsR} \cdot C_{\gsA} \otimes Q[C,D]_{\gsL \gsR} \\
        &= \frac{1}{\sqrt{N}} \sum_{x,y \in [N]} D_{\gsA} \cdot \ketbra*{y}{x}_{\gsA} \cdot C_{\gsA} \otimes \sum_{\ell \geq 0} \Pi^{\calR}_{\ell+1, \gsL} \cdot ( \sqrt{\ell+1} \cdot \ket*{x,y} \otimes \Pi^{\calR}_{\ell, \gsL}) \cdot (C \otimes D^T)^{\otimes \ell}_{\gsL}\\
        & \quad \otimes \Pi^{\calR}_{\gsR} \cdot (\overline{C} \otimes D^\dagger)^{\otimes *}_{\gsR}\\
        &= \frac{1}{\sqrt{N}} \sum_{x,y \in [N]} \ketbra*{y}{x}_{\gsA} \otimes \sum_{\ell \geq 0} \Pi^{\calR}_{\ell+1, \gsL} \cdot (C \otimes D^T)^{\otimes \ell+1}_{\gsL} \cdot ( \sqrt{\ell+1} \cdot \ket*{x,y} \otimes \Pi^{\calR}_{\ell, \gsL}) \\
        & \quad \otimes (\overline{C} \otimes D^\dagger)^{\otimes *}_{\gsR} \cdot \Pi^{\calR}_{\gsR} \\
        &= Q[C,D]_{\gsL \gsR} \cdot E^L_{\gsA \gsL \gsR}
    \end{align}
    A similar argument works for $E^{R}_{\gsA \gsL \gsR}$ to establish \cref{eq:ER-invariance}.
\end{proof}

\subsection{Approximate unitary invariance of $V^L$ and $V^R$}

We now prove approximate unitary invariance of the operators $V^L$ and $V^R$. The key step is the following lemma, which relates these operators to $E^L$ and $E^R$.

Recall that for an operator $M$ acting on registers $\sL, \sR$, the notation $M_{\leq t} = M \cdot \Pi_{\leq t, \gsL \gsR}$ refers to the restriction of the operator $M$ to states where the combined length of the $\sL$ and $\sR$ components is at most $t$.

\begin{claim}
\label{claim:VL-EL-close}
    For any positive integer $t$, 
    \begin{align}
        \norm{V^L_{\leq t} - E^L_{\leq t}}_{\opnorm} \leq \sqrt{\frac{2t(t+1)}{N}} \quad \text{and} \quad 
        \norm{V^R_{\leq t} - E^R_{\leq t}}_{\opnorm} \leq \sqrt{\frac{2t(t+1)}{N}}.
    \end{align}
\end{claim}

\begin{proof}
    We will only prove this for $V^L_{\leq t}$, as the proof for $V^R_{\leq t}$ is analogous. Let $\ket*{\psi}_{\gsA \gsL \gsR}$ be an arbitrary unit-norm state in the image of $\Id_{\gsA} \otimes \Pi_{\leq t, \gsL \gsR}$. In particular, 
    \begin{align}
        \ket*{\psi}_{\gsA \gsL \gsR} = \sum_{\substack{x \in [N],\\ (L,R) \in \calR^{2}}} \alpha_{x,L,R} \ket*{x}_{\gsA}\ket*{L}_{\gsL} \ket*{R}_{\gsR}.
    \end{align}
    where $\alpha_{x,L,R}$ is zero whenever $\abs{L \cup R} > t$. It suffices to show that for any such $\ket*{\psi}$,
    \begin{align}
        \norm{V^L \ket*{\psi} - E^L_{\gsA \gsL \gsR} \ket*{\psi}}_{\opnorm} \leq \sqrt{\frac{2t(t+1)}{N}}.
    \end{align}
    Expanding out $V^L \ket*{\psi}$, we get
    \begin{align}
        V^L \ket*{\psi} &= \sum_{\substack{x \in [N],\\ (L,R) \in \calR^{2}}} \frac{\alpha_{x,L,R} }{\sqrt{N - \abs{\Im(L \cup R)}}} \sum_{\substack{y\in [N]: \\ y \not\in \Im(L \cup R)}} \ket*{y} \ket*{L \cup \{(x,y)\}} \ket*{R}.
    \end{align}
    Expanding out $E^L_{\gsA \gsL \gsR} \ket*{\psi}_{\gsA \gsL \gsR}$, we get
    \begin{align}
        E^L_{\gsA \gsL \gsR} \ket*{\psi}
        &= \sum_{\substack{x \in [N],\\(L,R) \in \calR^{2}}} \frac{\alpha_{x,L,R}}{\sqrt{N}}  \sum_{y \in [N]} \ket*{y} \sqrt{\num(L,(x,y)) + 1} \cdot \ket*{L \cup \{(x,y)\}} \ket*{R}
    \end{align}
Then we have
\begin{align}
    &V^L_{\gsA \gsL \gsR} \ket*{\psi} - E^L_{\gsA \gsL \gsR} \ket*{\psi}\\
    &= \sum_{\substack{x \in [N],\\(L,R) \in \calR^{2}}} \alpha_{x,L,R} \sum_{y\in [N]} \ket*{y} \ket*{L \cup \{(x,y)\}} \ket*{R} \Bigg( \frac{\delta_{y \not\in \Im(L \cup R)}}{\sqrt{N - \abs{L \cup R}}}  - \frac{ \sqrt{\num(L,(x,y))+1}}{\sqrt{N}} \Bigg)\\
    &= \underbrace{\sum_{\substack{x \in [N],\\(L,R) \in \calR^{2}}} \alpha_{x,L,R} \sum_{\substack{y \in [N]:\\ y\not\in \Im(L \cup R)}} \ket*{y} \ket*{L \cup \{(x,y)\}} \ket*{R} \Bigg( \frac{1}{\sqrt{N - \abs{\Im(L \cup R)}}}  - \frac{ 1}{\sqrt{N}} \Bigg)}_{\coloneqq \ket*{v}} \nonumber \\
    &\quad + \underbrace{\sum_{\substack{x \in [N],\\(L,R) \in \calR^{2}}} \alpha_{x,L,R} \sum_{y \in \Im(L \cup R)} \ket*{y} \ket*{L \cup \{(x,y)\}} \ket*{R} \Bigg( - \frac{ \sqrt{\num(L,(x,y))+1}}{\sqrt{N}} \Bigg)}_{\coloneqq \ket*{w}}.
\end{align}
Note that $\ket*{v}$ and $\ket*{w}$ are orthogonal, since $\ket*{v}$ is a superposition of states $\ket*{y} \ket*{L'} \ket*{R}$ where $y$ is in $\Im(L' \cup R)$ exactly once, while $\ket*{w}$ is a superposition of states $\ket*{y} \ket*{L'} \ket*{R}$ where $y$ is in $\Im(L' \cup R)$ at least twice. Thus,
\begin{align}
    \norm{V^L_{\gsA \gsL \gsR} \ket*{\psi} - E^L_{\gsA \gsL \gsR} \ket*{\psi}}^2 &= \braket{v} + \braket{w}
\end{align}
\paragraph{Bounding $\braket{v}$.} By changing the order of summation, we can rewrite $\ket*{v}$ as
\begin{align}
    \ket*{v} = \sum_{\substack{y \in [N],\\ (L',R) \in \calR^2}} \ket*{y} \ket*{L'} \ket*{R} \Bigg( \sum_{\substack{(x,L): \\ L' = L \cup \{(x,y)\},\\ y\not\in \Im(L \cup R)}} \alpha_{x,L,R} \Big(\frac{1}{\sqrt{N - \abs{\Im(L \cup R)}}} - \frac{1}{\sqrt{N}}\Big) \Bigg),
\end{align}
and thus
\begin{align}
    \braket{v} &= \sum_{\substack{y \in [N],\\ (L',R) \in \calR^2}} \Bigg( \sum_{\substack{(x,L):\\ L' = L \cup \{(x,y)\},\\ y\not\in \Im(L \cup R)}} \alpha_{x,L,R} \Big(\frac{1}{\sqrt{N - \abs{\Im(L \cup R)}}} - \frac{1}{\sqrt{N}}\Big) \Bigg)^2\\
    &\leq \sum_{\substack{y \in [N],\\ (L',R) \in \calR^2}} \Bigg(\sum_{\substack{(x,L):\\ L' = L \cup \{(x,y)\},\\ y\not\in \Im(L \cup R)}} \abs{\alpha_{x,L,R}}^2 \Bigg) \cdot \Bigg(\sum_{\substack{(x,L):\\ L' = L \cup \{(x,y)\},\\ y\not\in \Im(L \cup R)}} \Big(\frac{1}{\sqrt{N - \abs{\Im(L \cup R)}}} - \frac{1}{\sqrt{N}}\Big)^2 \Bigg),
\end{align}
where the last inequality is by Cauchy-Schwarz. We can bound the summand by writing
\begin{align}
    \sum_{\substack{(x,L):\\ L' = L \cup \{(x,y)\},\\ y\not\in \Im(L \cup R)}} \Big(\frac{1}{\sqrt{N - \abs{\Im(L \cup R)}}} - \frac{1}{\sqrt{N}}\Big)^2 &= \sum_{\substack{(x,L):\\ L' = L \cup \{(x,y)\},\\ y\not\in \Im(L \cup R)}} \Big(\frac{\sqrt{N} - \sqrt{N - \abs{\Im(L \cup R)}}}{\sqrt{N(N - \abs{\Im(L \cup R)})}}\Big)^2\\
    & \leq \sum_{\substack{(x,L):\\ L' = L \cup \{(x,y)\},\\ y\not\in \Im(L \cup R)}} \Big(\frac{\sqrt{\abs{\Im(L \cup R)}}}{\sqrt{N(N - \abs{\Im(L \cup R)})}}\Big)^2 \tag{since $\sqrt{a} - \sqrt{b} \leq \sqrt{a-b}$ when $a\geq b \geq 0$}\\
    &= \sum_{\substack{(x,L):\\ L' = L \cup \{(x,y)\},\\ y\not\in \Im(L \cup R)}} \frac{\abs{\Im(L \cup R)}}{N(N - \abs{\Im(L \cup R)})}\\
    & \leq \frac{(\abs{L} + 1) \cdot \abs{\Im(L \cup R)}}{N(N - \abs{\Im(L \cup R)})} 
\end{align}
where the last inequality uses the fact that for any fixed $L'$, there are at most $\abs{L} +1$ choices of $(x,L)$ that can satisfy $L' = L \cup \{(x,y)\}$. Thus,
\begin{align}
    \braket{v} &\leq \frac{(\abs{L} + 1) \cdot \abs{\Im(L \cup R)}}{N(N - \abs{\Im(L \cup R)})} \cdot \sum_{\substack{y \in [N],\\ (L',R) \in \calR^2}} \Bigg(\sum_{\substack{(x,L):\\ L' = L \cup \{(x,y)\},\\ y\not\in \Im(L \cup R)}} \abs{\alpha_{x,L,R}}^2 \Bigg) \\
    &=\frac{(\abs{L} + 1) \cdot \abs{\Im(L \cup R)}}{N(N - \abs{\Im(L \cup R)})} \cdot \sum_{\substack{x \in [N],\\ (L,R) \in \calR^2}} \abs{\alpha_{x,L,R}}^2 \cdot \Big( \sum_{y \in [N]} \delta_{y \not\in \Im(L \cup R)} \Big) \\
    &\leq \frac{(\abs{L} + 1) \cdot \abs{\Im(L \cup R)}}{N} \cdot \sum_{\substack{x \in [N],\\ (L,R) \in \calR^2}} \abs{\alpha_{x,L,R}}^2 = \frac{(\abs{L} + 1) \cdot \abs{\Im(L \cup R)}}{N}.
\end{align}
\paragraph{Bounding $\braket{w}$.} By changing the order of summation, we can rewrite $\ket*{w}$ as
\begin{align}
    \ket*{w} = \sum_{\substack{y \in [N],\\ (L',R) \in \calR^2}}  \ket*{y} \ket*{L'} \ket*{R} \Bigg( \sum_{\substack{(x,L):\\ L' = L \cup \{(x,y)\},\\ y \in \Im(L \cup R)}} \alpha_{x,L,R} \Big(-\frac{ \sqrt{\num(L,(x,y)) + 1}}{\sqrt{N}}\Big) \Bigg).
\end{align}
Thus,
\begin{align}
    \braket{w} &= \sum_{\substack{y \in [N],\\ (L',R) \in \calR^2}} \Bigg\lvert \sum_{\substack{(x,L):\\ L' = L \cup \{(x,y)\},\\ y \in \Im(L \cup R)}} \alpha_{x,L,R} \Big(-\frac{ \sqrt{\num(L,(x,y)) + 1}}{\sqrt{N}}\Big) \Bigg\rvert^2\\
    &\leq \sum_{\substack{y \in [N],\\ (L',R) \in \calR^2}}  \Bigg(\sum_{\substack{(x,L):\\ L' = L \cup \{(x,y)\},\\ y \in \Im(L \cup R)}} \abs{\alpha_{x,L,R}}^2 \Bigg) \cdot \Bigg(\sum_{\substack{(x,L):\\ L' = L \cup \{(x,y)\},\\ y \in \Im(L \cup R)}} \frac{ \num(L,(x,y)) + 1}{N} \Bigg) \tag{by Cauchy-Schwarz}\\
    & \leq \sum_{\substack{y \in [N],\\ (L',R) \in \calR^2}}  \Bigg(\sum_{\substack{(x,L):\\ L' = L \cup \{(x,y)\},\\ y \in \Im(L \cup R)}} \abs{\alpha_{x,L,R}}^2 \Bigg) \cdot \frac{(\abs{L} + 1)}{N},
\end{align}
where we have used the fact that for any $y,L'$, we have the upper bound \begin{align}
    \sum_{\substack{(x,L):\\ L' = L \cup \{(x,y)\},\\ y \in \Im(L \cup R)}}  \num(L,(x,y)) + 1 \leq \abs{L}+1,
\end{align}
since each tuple in $L'$ increases the value of $\num(L,(x,y))$ by $1$ for at most one $x$. Thus,
\begin{align}
    \braket{w} &= \frac{\abs{L} + 1}{N} \cdot \sum_{\substack{x \in [N],\\(L,R) \in \calR^2}} \abs{\alpha_{x,L,R}}^2 \cdot \Big( \sum_{y \in [N]} \delta_{y \in \Im(L \cup R)}\Big)\\
    &\leq \frac{(\abs{L} + 1) \cdot |\Im(L \cup R)|}{N} \cdot \sum_{\substack{x \in [N],\\(L,R) \in \calR^2}} \abs{\alpha_{x,L,R}}^2  = \frac{(\abs{L} + 1) \cdot |\Im(L \cup R)|}{N}.
\end{align}
Putting everything together, we have that for all $\ket*{\psi}_{\gsA \gsL \gsR}$ in the image of $\Id_{\gsA} \Pi_{\leq t, \gsL \gsR}$,
\begin{align}
    \norm{V^L_{\gsA \gsL \gsR} \ket*{\psi} - E^L_{\gsA \gsL \gsR} \ket*{\psi}} &\leq \sqrt{\frac{2 (\abs{L} + 1) \cdot \abs{\Im(L \cup R)}}{N}} \leq \sqrt{\frac{2t(t+1)}{N}},
\end{align}
since $\abs{\Im(L \cup R)} \leq t$ and $\abs{L}+1\leq t+1$. This completes the claim.
\end{proof}

\begin{claim}
\label{claim:VL-VR-approx-invariance}
    For any positive integer $t$, and any pair of $n$-qubit unitaries $C,D$, we have 
    \begin{align}
        \norm{D_{\gsA} \cdot V^L_{\leq t} \cdot C_{\gsA} - Q[C,D]_{\gsL \gsR} \cdot V^L_{\leq t} \cdot Q[C,D]_{\gsL \gsR}^\dagger }_{\opnorm} &\leq 2 \cdot \sqrt{\frac{2t(t+1)}{N}} \label{eq:VL-approx-invariant}\\
        \norm{D_{\gsA} \cdot V^{R,\dagger}_{\leq t} \cdot C_{\gsA} - Q[C,D]_{\gsL \gsR} \cdot V^{R,\dagger}_{\leq t} \cdot Q[C,D]_{\gsL \gsR}^\dagger}_{\opnorm} &\leq 2 \cdot \sqrt{\frac{2t(t+1)}{N}}. \label{eq:VR-approx-invariant}
    \end{align}
\end{claim}

\begin{proof}
    We first prove~\cref{eq:VL-approx-invariant}. Using~\cref{claim:E-invariance} together and the triangle inequality, we have
    \begin{align}
        &\norm{D_{\gsA} \cdot V^L_{\leq t} \cdot C_{\gsA} - Q[C,D]_{\gsL \gsR} \cdot V^L_{\leq t} \cdot Q[C,D]_{\gsL \gsR}^\dagger }_{\opnorm} \\
        &\leq \norm{D_{\gsA} \cdot V^L_{\leq t} \cdot C_{\gsA} -  D_{\gsA} \cdot E^L_{\leq t} \cdot C_{\gsA} }_{\opnorm}\nonumber \\
        & \quad + \norm{Q[C,D]_{\gsL \gsR} \cdot E^L_{\leq t} \cdot Q[C,D]_{\gsL \gsR}^\dagger- Q[C,D]_{\gsL \gsR} \cdot V^L_{\leq t} \cdot  Q[C,D]_{\gsL \gsR}^\dagger}_{\opnorm}\\
        &\leq 2 \cdot \norm{V^L_{\leq t} -E^L_{\leq t}}_{\opnorm} \tag{by unitary invariance of $\norm{\cdot}_{\opnorm}$}\\
        & \leq 2 \cdot \sqrt{\frac{2t(t+1)}{N}}. \tag{by~\cref{claim:VL-EL-close}}
    \end{align}
    \cref{eq:VR-approx-invariant} follows from a symmetric argument.
\end{proof}

Note that with our convention that $M_{\leq t} = M \cdot \Pi_{\leq t}$, the operator $M_{\leq t}^\dagger = (M \cdot \Pi_{\leq t})^\dagger = \Pi_{\leq t} \cdot M^\dagger$ is not the same as $(M^\dagger)_{\leq t} = M^\dagger \cdot \Pi_{\leq t}$. However, since our $V^L$ and $V^R$ operators map $\Pi_{\leq t}$ to $\Pi_{\leq t +1}$, we have the following identities,
\begin{align}
    (V^{L,\dagger})_{\leq t} &= V^{L,\dagger} \cdot \Pi_{\leq t} = \Pi_{\leq t-1} \cdot V^{L,\dagger} = V^{L,\dagger}_{\leq t-1}\\
    (V^{R,\dagger})_{\leq t} &= V^{R,\dagger} \cdot \Pi_{\leq t} = \Pi_{\leq t-1} \cdot V^{R,\dagger} = V^{R,\dagger}_{\leq t-1}.
\end{align}
As a consequence,~\cref{eq:VR-approx-invariant} also holds for the ``mis-parenthesized'' version. In particular, for any positive integer $t$ and any $C,D$, we have
\begin{align}
    \norm{D_{\gsA} \cdot (V^{R,\dagger})_{\leq t} \cdot C_{\gsA} - Q[C,D]_{\gsL \gsR} \cdot (V^{R,\dagger})_{\leq t} \cdot Q[C,D]_{\gsL \gsR}^\dagger}_{\opnorm} \leq 2 \cdot \sqrt{\frac{2t(t+1)}{N}}. \label{eq:VR-approx-invariant-alt}
\end{align}
To prove the approximate unitary invariance of $V$, we need to utilize the following basic lemma.

\begin{lemma} \label{lem:basic-ABA'B'-bound}
    Given any operators $A, B, A', B'$ with operator norm bounded above by one, we have
    \begin{align}
        \norm{A \cdot B - A' \cdot B'}_{\opnorm} \leq \norm{A - A'}_{\opnorm} + \norm{B - B'}_{\opnorm}.
    \end{align}
\end{lemma}
\begin{proof}
    We can prove this lemma via triangle inequality,
    \begin{align}
    \norm{A \cdot B - A' \cdot B'}_{\opnorm} &\leq \norm{A \cdot B - A' \cdot B}_{\opnorm} + \norm{A' \cdot B - A' \cdot B'}_{\opnorm}\\
    &\leq \norm{A - A'}_{\opnorm} \cdot \norm{B}_{\opnorm} + \norm{A'}_{\opnorm} \cdot \norm{B - B'}_{\opnorm}\\
    &\leq \norm{A - A'}_{\opnorm} + \norm{B - B'}_{\opnorm}.
    \end{align}
    This completes the proof.
\end{proof}

We start by proving the approximate unitary invariance for the projectors $V^L \cdot V^{L, \dagger}$ and $V^R \cdot V^{R, \dagger}$.

\begin{claim}
\label{claim:VLproj-VRproj-approx-invariance}
    For any positive integer $t$, and any pair of $n$-qubit unitaries $C,D$, we have 
    \begin{align}
        \norm{D_{\gsA} \cdot \Big(V^L \cdot V^{L,\dagger}\Big)_{\leq t} \cdot D_{\gsA}^\dagger - Q[C,D]_{\gsL \gsR} \cdot \Big(V^L \cdot V^{L,\dagger}\Big)_{\leq t} \cdot Q[C,D]_{\gsL \gsR}^\dagger }_{\opnorm} &\leq 4 \cdot \sqrt{\frac{2t(t+1)}{N}} \label{eq:VLproj-approx-invariant}\\
        \norm{C_{\gsA}^\dagger \cdot \Big(V^R \cdot V^{R,\dagger}\Big)_{\leq t} \cdot C_{\gsA} - Q[C,D]_{\gsL \gsR} \cdot \Big(V^R \cdot V^{R,\dagger}\Big)_{\leq t} \cdot Q[C,D]_{\gsL \gsR}^\dagger}_{\opnorm} &\leq 4 \cdot \sqrt{\frac{2t(t+1)}{N}}. \label{eq:VRproj-approx-invariant}
    \end{align}
\end{claim}

\begin{proof}
    By the definition of $V^L$, we have $(V^L \cdot V^{L,\dagger})_{\leq t} = V^L_{\leq t-1} \cdot V^{L,\dagger}_{\leq t-1}$. We have
    \begin{align}
        &\norm{D_{\gsA} \cdot V^L_{\leq t-1} \cdot V^{L,\dagger}_{\leq t-1} \cdot D_{\gsA}^\dagger - Q[C,D]_{\gsL \gsR} \cdot V^L_{\leq t-1} \cdot V^{L,\dagger}_{\leq t-1} \cdot Q[C,D]_{\gsL \gsR}^\dagger }_{\opnorm} \\
        &= \Bigg\lVert \Big(D_{\gsA} \cdot V^L_{\leq t-1} \cdot C_{\gsA} \Big)\cdot \Big( C_{\gsA}^\dagger \cdot V^{L,\dagger}_{\leq t-1} \cdot D_{\gsA}^\dagger\Big) \nonumber \\
        & \quad - \Big(Q[C,D]_{\gsL \gsR} \cdot V^L_{\leq t-1} \cdot Q[C,D]^\dagger_{\gsL \gsR} \Big) \cdot \Big( Q[C,D]_{\gsL \gsR} \cdot V^{L,\dagger}_{\leq t-1} \cdot Q[C,D]_{\gsL \gsR}^\dagger \Big) \Bigg\rVert_{\opnorm}\\
        & \leq \Bigg\lVert \Big(D_{\gsA} \cdot V^L_{\leq t-1} \cdot C_{\gsA} \Big) - \Big(Q[C,D]_{\gsL \gsR} \cdot V^L_{\leq t-1} \cdot Q[C,D]^\dagger_{\gsL \gsR} \Big)\Bigg\rVert_{\opnorm} \nonumber \\
        &\quad + \Bigg\lVert \Big( C_{\gsA}^\dagger \cdot V^{L,\dagger}_{\leq t-1} \cdot D_{\gsA}^\dagger\Big) - \Big( Q[C,D]_{\gsL \gsR} \cdot V^{L,\dagger}_{\leq t-1} \cdot Q[C,D]_{\gsL \gsR}^\dagger \Big) \Bigg\rVert_{\opnorm} \tag{by~\cref{lem:basic-ABA'B'-bound}}\\
        & \leq 4 \cdot \sqrt{\frac{2t(t+1)}{N}} \tag{by~\cref{claim:VL-VR-approx-invariance}}
    \end{align}
    The statement for $V^R$ can be proven similarly. This concludes the proof of this claim.
\end{proof}

We can now prove approximate invariance of $V$ (\cref{claim:two-sided-invariance}). By unitary invariance of $\norm{\cdot}_{\opnorm}$ we can restate lemma~\cref{claim:two-sided-invariance} as follows.

\begin{lemma}[\cref{claim:two-sided-invariance}, restated]
    For any positive integer $t$, and any pair of $n$-qubit unitaries $C,D$, we have 
    \begin{align}
        \norm{D_{\gsA} \cdot V_{\leq t} \cdot C_{\gsA}
        - Q[C,D]_{\gsL \gsR} \cdot V_{\leq t} \cdot Q[C,D]_{\gsL \gsR}^\dagger}_{\opnorm} & \leq 16\sqrt{\frac{2t(t+1)}{N}},\\
        \norm{C_{\gsA}^\dagger \cdot (V^\dagger)_{\leq t} \cdot D_{\gsA}^\dagger
        - Q[C,D]_{\gsL \gsR} \cdot (V^\dagger)_{\leq t} \cdot Q[C,D]_{\gsL \gsR}^\dagger}_{\opnorm} & \leq 16\sqrt{\frac{2t(t+1)}{N}},
    \end{align}
\end{lemma}

\begin{proof}
    We will prove the first inequality, as the second follows from a symmetric argument. From the definition of $V$, we have
    \begin{align}
        V = V^L \cdot (\Id - V^R \cdot V^{R,\dagger}) + (\Id - V^L \cdot V^{L,\dagger}) \cdot V^{R,\dagger}.
    \end{align}
    From the definitions of $\Pi_{\leq t}$, $V^L$, and $V^R$, we note that
    \begin{align}
        (V^L \cdot V^R \cdot V^{R,\dagger})_{\leq t} &= V^L_{\leq t} \cdot (V^R \cdot V^{R,\dagger})_{\leq t},\\
        (V^L \cdot V^{L,\dagger} \cdot V^{R, \dagger})_{\leq t} &= (V^L \cdot V^{L,\dagger})_{\leq t} \cdot (V^{R, \dagger})_{\leq t}.
    \end{align}
    Using this fact and the definition of $V$, we can apply the triangle inequality to obtain,
    \begin{align}
        & \norm{D_{\gsA} \cdot V_{\leq t} \cdot C_{\gsA} - Q[C,D]_{\gsL \gsR} \cdot V_{\leq t} \cdot Q[C,D]_{\gsL \gsR}^\dagger }_{\opnorm} \\
        & \leq \Bigg\lVert D_{\gsA} \cdot V^L_{\leq t} \cdot C_{\gsA} - Q[C,D]_{\gsL \gsR} \cdot V^L_{\leq t} \cdot Q[C,D]_{\gsL \gsR}^\dagger \Bigg\rVert_{\opnorm}\label{eq:first-term-DVC} \\
        & + \Bigg\lVert D_{\gsA} \cdot V^L_{\leq t} \cdot (V^R \cdot V^{R,\dagger})_{\leq t} \cdot C_{\gsA} - Q[C,D]_{\gsL \gsR} \cdot V^L_{\leq t} \cdot (V^R \cdot V^{R,\dagger})_{\leq t} \cdot Q[C,D]_{\gsL \gsR}^\dagger \Bigg\rVert_{\opnorm} \label{eq:second-term-DVC} \\
        & + \Bigg\lVert D_{\gsA} \cdot (V^{R, \dagger})_{\leq t} \cdot C_{\gsA} - Q[C,D]_{\gsL \gsR} \cdot (V^{R, \dagger})_{\leq t} \cdot Q[C,D]_{\gsL \gsR}^\dagger \Bigg\rVert_{\opnorm} \label{eq:third-term-DVC} \\
        & + \Bigg\lVert D_{\gsA} \cdot (V^L \cdot V^{L,\dagger})_{\leq t} \cdot (V^{R, \dagger})_{\leq t} \cdot C_{\gsA} - Q[C,D]_{\gsL \gsR} \cdot (V^L \cdot V^{L,\dagger})_{\leq t} \cdot (V^{R, \dagger})_{\leq t} \cdot Q[C,D]_{\gsL \gsR}^\dagger \Bigg\rVert_{\opnorm}. \label{eq:fourth-term-DVC}
    \end{align}
    We now bound each of the four terms. The first term \cref{eq:first-term-DVC} is bounded above by $2 \cdot \sqrt{\frac{2 t (t+1)}{N}}$ from \cref{eq:VL-approx-invariant}.
    The third term \cref{eq:third-term-DVC} is also bounded above by $2 \cdot \sqrt{\frac{2 t (t+1)}{N}}$ from \cref{eq:VR-approx-invariant-alt}.
    The second and fourth terms \cref{eq:second-term-DVC}, \cref{eq:fourth-term-DVC} require the use of \cref{lem:basic-ABA'B'-bound}.
    Hence, we can bound the second term \cref{eq:second-term-DVC} as follows,
    \begin{align}
        & \Bigg\lVert D_{\gsA} \cdot V^L_{\leq t} \cdot (V^R \cdot V^{R,\dagger})_{\leq t} \cdot C_{\gsA} - Q[C,D]_{\gsL \gsR} \cdot V^L_{\leq t} \cdot (V^R \cdot V^{R,\dagger})_{\leq t} \cdot Q[C,D]_{\gsL \gsR}^\dagger \Bigg\rVert_{\opnorm} \\
        & = \Bigg\lVert D_{\gsA} \cdot V^L_{\leq t} \cdot C_{\gsA} \cdot C_{\gsA}^\dagger \cdot (V^R \cdot V^{R,\dagger})_{\leq t} \cdot C_{\gsA} \nonumber\\
        &\quad\quad\quad\quad - Q[C,D]_{\gsL \gsR} \cdot V^L_{\leq t} \cdot Q[C,D]_{\gsL \gsR}^\dagger \cdot Q[C, D]_{\gsL \gsR} \cdot (V^R \cdot V^{R,\dagger})_{\leq t} \cdot Q[C,D]_{\gsL \gsR}^\dagger \Bigg\rVert_{\opnorm} \\
        &\leq \Bigg\lVert D_{\gsA} \cdot V^L_{\leq t} \cdot C_{\gsA} - Q[C,D]_{\gsL \gsR} \cdot V^L_{\leq t} \cdot Q[C,D]_{\gsL \gsR}^\dagger \Bigg\rVert_{\opnorm} \label{eq:second-term-first-term-DVC} \\
        &+ \Bigg\lVert C_{\gsA}^\dagger \cdot (V^R \cdot V^{R,\dagger})_{\leq t} \cdot C_{\gsA} - Q[C,D]_{\gsL \gsR} \cdot (V^R \cdot V^{R,\dagger})_{\leq t} \cdot Q[C,D]_{\gsL \gsR}^\dagger \Bigg\rVert_{\opnorm}, \label{eq:second-term-second-term-DVC}\\
        &\leq 6 \cdot \sqrt{\frac{2 t (t+1)}{N}}.
    \end{align}
    where we used the fact that \cref{eq:second-term-first-term-DVC} is bounded above by $2 \cdot \sqrt{\frac{2 t (t+1)}{N}}$ from \cref{eq:VL-approx-invariant} and \cref{eq:second-term-second-term-DVC} is bounded above by $4 \cdot \sqrt{\frac{2 t (t+1)}{N}}$ from \cref{eq:VRproj-approx-invariant}.
    Similarly, we can bound the fourth term given \cref{eq:fourth-term-DVC} using the same argument to obtain
    \begin{equation}
        \Bigg\lVert D_{\gsA} \cdot V^L_{\leq t} \cdot (V^R \cdot V^{R,\dagger})_{\leq t} \cdot C_{\gsA} - Q[C,D]_{\gsL \gsR} \cdot V^L_{\leq t} \cdot (V^R \cdot V^{R,\dagger})_{\leq t} \cdot Q[C,D]_{\gsL \gsR}^\dagger \Bigg\rVert_{\opnorm}        \leq 6 \cdot \sqrt{\frac{2 t (t+1)}{N}}.
    \end{equation}
    Combining the bounds on the four terms, we obtain
    \begin{equation}
        \norm{D_{\gsA} \cdot V_{\leq t} \cdot C_{\gsA} - Q[C,D]_{\gsL \gsR} \cdot V_{\leq t} \cdot Q[C,D]_{\gsL \gsR}^\dagger }_{\opnorm} \leq 16 \cdot \sqrt{\frac{2 t (t+1)}{N}}.
    \end{equation}
    This completes the proof of the approximate unitary invariance of $V$.
\end{proof}

\section{Proof of~\cref{lem:twirling-strongPRU}}
\label{sec:twirling-strongPRU}

In this section, we prove~\cref{lem:twirling-strongPRU}. For convenience, we restate the lemma below.

\begin{lemma}[\cref{lem:twirling-strongPRU}, restated]
\label{lem:twirling-strongPRU-restated}
    For any unitary $2$-design $\frakD$ and integer $0 \leq t \leq N-1$, we have
    \begin{align}
        \norm{  \E_{C,D \gets \frakD} (C_{\gsA} \otimes Q[C,D]_{\gsL \gsR})^\dagger \cdot \Big( \Pi^{\bij}_{\leq t, \gsL \gsR} - \Pi^{\calD(W)}_{\leq t, \gsA \gsL \gsR}\Big) \cdot (C_{\gsA} \otimes Q[C,D]_{\gsL \gsR}) }_{\opnorm} &\leq  6t \sqrt{\frac{t}{N}},\\
        \norm{  \E_{C,D \gets \frakD} (D^\dagger_{\gsA} \otimes Q[C,D]_{\gsL \gsR})^\dagger \cdot \Big( \Pi^{\bij}_{\leq t, \gsL \gsR} - \Pi^{\calI(W)}_{\leq t, \gsA \gsL \gsR}\Big) \cdot (D^\dagger_{\gsA} \otimes Q[C,D]_{\gsL \gsR}) }_{\opnorm} &\leq 6t \sqrt{\frac{t}{N}},
\end{align}
\end{lemma}
In the above expressions, $\Pi^{\bij}_{\leq t, \gsL \gsR}$ is shorthand for $\Id_{\gsA} \otimes \Pi^{\bij}_{\leq t, \gsL \gsR}$, and thus the operators inside the $\norm{\cdot}_{\opnorm}$ act on $\sA,\sL,\sR$.

\subsection{The domain and image of $W$}

In order to prove~\cref{lem:twirling-strongPRU-restated}, we will first need to give an explicit characterization of the projectors $\Pi^{\calD(W)}$ and $\Pi^{\calI(W)}$.

\begin{definition}
    Let 
    \begin{align}
        \Pi^{\not\in \Dom} &\coloneqq \sum_{\substack{(L,R) \in \calR^2,\\ x\not\in \Dom(L \cup R)}} \ketbra*{x}_{\gsA} \otimes \ketbra*{L}_{\gsL} \otimes \ketbra*{R}_{\gsR}\\
        \Pi^{\not\in \Im} &\coloneqq \sum_{\substack{(L,R) \in \calR^2, \\ y\not\in \Im(L \cup R)}} \ketbra*{y}_{\gsA} \otimes \ketbra*{L}_{\gsL} \otimes \ketbra*{R}_{\gsR}.
    \end{align}
\end{definition}

\begin{definition}
    Let 
    \begin{align}
        \Pi^{\EPR} \coloneqq \ketbra*{\EPR_N} = \Big( \frac{1}{\sqrt{N}} \sum_{x \in [N]} \ket*{x} \ket*{x} \Big) \cdot \Big( \frac{1}{\sqrt{N}} \sum_{y \in [N]} \bra*{y} \bra*{y} \Big).
    \end{align}
\end{definition}

\begin{notation}
    We use the notation $\Pi^{\EPR}_{\darkgray{\sA,\sR^{(r)}_{\sX,i}}}$ for the projector on registers $\sA,\sR^{(r)}$ that applies $\Pi^{\EPR}$ to the registers $\sA$, $\sR^{(r)}_{\sX,i}$ (where $i \in [r]$), and acts as identity on the rest of $\sR^{(r)}$. The same notation applies for $\Pi^{\EPR}_{\darkgray{\sA,\sL^{(\ell)}_{\sY,i}}}$.
\end{notation}

\begin{fact}
\label{fact:pi-bij-R2-xydist}
    The projectors $\Pi^{\calR^2}_{\gsL \gsR}$ and $\Pi^{\xydist}_{\gsL \gsR}$ commute, and moreover
    \begin{align}
        \Pi^{\bij}_{\gsL \gsR} = \Pi^{\calR^2}_{\gsL \gsR} \cdot \Pi^{\xydist}_{\gsL \gsR}
    \end{align}
\end{fact}

\begin{claim}
\label{claim:expand-out-Pi-DW-Pi-IW}
    \begin{align}
        \Pi^{\calD(W)} &= \Pi^{\bij}_{\gsL \gsR} \cdot \Bigg( \Pi^{\not\in \Dom}_{\gsA \gsL \gsR} + \sum_{\substack{\ell,r \geq 0: \\ \ell + r < N}} \frac{N}{N-\ell-r} \Pi_{\ell,\gsL} \otimes \sum_{i \in [r+1]} \Pi^{\EPR}_{\darkgray{\sA,\sR^{(r+1)}_{\sX,i}}} \Bigg) \cdot \Pi^{\bij}_{\gsL \gsR} \label{eq:expand-out-Pi-DW-Pi-IW-1}\\
        \Pi^{\calI(W)} &= \Pi^{\bij}_{\gsL \gsR} \cdot \Bigg( \Pi^{\not\in \Im}_{\gsA \gsL \gsR} + \sum_{\substack{\ell,r \geq 0: \\ \ell + r < N}} \frac{N}{N-\ell-r} \Pi_{r,\gsR} \otimes \sum_{i \in [\ell+1]} \Pi^{\EPR}_{\darkgray{\sA,\sL^{(\ell+1)}_{\sY,i}}} \Bigg) \cdot \Pi^{\bij}_{\gsL \gsR}\label{eq:expand-out-Pi-DW-Pi-IW-2}
    \end{align}
\end{claim}

\begin{proof}
    By~\cref{fact:domain-image-W-WL-WR},
    \begin{align}
        \Pi^{\calD(W)} &= \Pi^{\calD(W^L)} + \Pi^{\calI(W^R)}.
    \end{align}
    To prove~\cref{eq:expand-out-Pi-DW-Pi-IW-1}, it suffices to prove
    \begin{align}
        \Pi^{\calD(W^L)} &= \Pi^{\bij}_{\gsL \gsR} \cdot \Pi^{\not\in \Dom}_{\gsA \gsL \gsR} \cdot \Pi^{\bij}_{\gsL \gsR} \label{eq:expand-out-Pi-D-WL} \\
        \Pi^{\calI(W^R)} &= \Pi^{\bij}_{\gsL \gsR} \cdot \Bigg( \sum_{\substack{\ell,r \geq 0: \\ \ell + r < N}} \frac{N}{N-\ell-r} \Pi_{\ell,\gsL} \otimes \sum_{i \in [r+1]} \Pi^{\EPR}_{\darkgray{\sA,\sR^{(r+1)}_{\sX,i}}} \Bigg) \cdot \Pi^{\bij}_{\gsL \gsR}\label{eq:expand-out-Pi-I-WR}
    \end{align}
    \paragraph{Proof of~\cref{eq:expand-out-Pi-D-WL}.} From the definition of $W^L$, its domain is the image of the projector 
    \begin{align}
        \Pi^{\calD(W^L)} &= \sum_{\substack{(L,R) \in \calR^{2,\dist},\\ x\not\in \Dom(L,R)}} \ketbra*{x}_{\gsA} \otimes \ketbra*{L}_{\gsL} \otimes \ketbra*{R}_{\gsR} \\
        &= \Bigg( \sum_{(L,R) \in \calR^{2,\dist}} \ketbra*{L}_{\gsL} \otimes \ketbra*{R}_{\gsR} \Bigg) \cdot \Bigg(\sum_{\substack{(L,R) \in \calR^2,\\ x\not\in \Dom(L \cup R)}} \ketbra*{x}_{\gsA} \otimes \ketbra*{L}_{\gsL} \otimes \ketbra*{R}_{\gsR}\Bigg) \nonumber \\
        & \quad \cdot \Bigg( \sum_{(L,R) \in \calR^{2,\dist}} \ketbra*{L}_{\gsL} \otimes \ketbra*{R}_{\gsR} \Bigg)\\
        &= \Pi^{\bij}_{\gsL \gsR} \cdot \Pi^{\not\in \Dom}_{\gsA \gsL \gsR} \cdot \Pi^{\bij}_{\gsL \gsR}.
    \end{align}
    \paragraph{Proof of~\cref{eq:expand-out-Pi-I-WR}.} We can expand out
    \begin{align}
        \Pi^{\calI(W^R)} = W^R \cdot W^{R,\dagger} &= W^R \cdot \sum_{\substack{\ell,r \geq 0,\\ \ell + r < N}} \Pi_{\ell,r,\gsL \gsR} \cdot W^{R,\dagger}\\
        &= \sum_{\substack{\ell,r \geq 0,\\ \ell + r < N}} W_{\ell,r}^R \cdot W^{R,\dagger}_{\ell,r}
    \end{align}
    where the second equality uses the fact that the domain of $W^R$ is contained in the image of the projector $\sum_{\ell,r \geq 0, \ell + r < N} \Pi_{\ell,r,\gsL \gsR}$, i.e., $W^R$ is only defined on states where the $\sL$ and $\sR$ registers have sizes $\ell,r \geq 0$ where $\ell + r <N$. Thus, it suffices to prove that for all $\ell,r \geq 0$ such that $\ell+ r<N$ that 
    \begin{align}
        W_{\ell,r}^R \cdot W_{\ell,r}^{R,\dagger} = \frac{N}{N - \ell - r} \Pi^{\bij}_{\ell,r+1,\gsL \gsR} \cdot \Bigg( \Pi_{\ell,\gsL} \otimes \sum_{i \in [r+1]} \Pi^{\EPR}_{\darkgray{\sA,\sR^{(r+1)}_{\sX,i}}} \Bigg) \cdot \Pi^{\bij}_{\ell,r+1,\gsL \gsR},
    \end{align}
    where we use our notational convention that for an operator $B$ acting on a variable-length registers $\sL, \sR$, the operator $B_{\ell,r} = B \cdot \Pi_{\ell,r, \gsL \gsR}$ is the restriction of $B$ to states where the $\sL$ register is length $\ell$ and the $\sR$ register is length $r$.

    We will do this by relating $W^R$ to the $E^R$ operator defined in \cref{def:ELER-definition}.
    From the definition of $W^R$ in \cref{def:ternary-pfo-action}, we immediately have:
    \begin{align}
        W^R_{\ell, r} &\coloneqq \frac{1}{\sqrt{N - \ell - r}} \sum_{\substack{(L,R) \in \calR^{2,\dist},\\ \abs{L} = \ell, \abs{R} = r,\\ x\not\in \Dom(L \cup R),\\ y\not\in \Im(L \cup R)}} \ketbra*{x}{y}_{\gsA} \otimes \ketbra*{L}_{\gsL} \otimes \ketbra*{R \cup \{(x,y)\}}{R}_{\gsR}
    \end{align}
    Using \cref{def:ER-operator-num-form} for $E^{R}$, we have
    \begin{align}
        E^R_{\ell,r} = \frac{1}{\sqrt{N}} \sum_{\substack{(L,R) \in \calR^2,\\ \abs{L} = \ell, \abs{R} = r,\\ x,y \in [N]}} \sqrt{\num(R,(x,y)) + 1} \cdot \ketbra*{x}{y}_{\gsA} \otimes \ketbra*{L}_{\gsL} \otimes \ketbra*{R \cup \{(x,y)\}}{R}_{\gsR} \label{eq:ER-def-repeated-for-WR-proof}
    \end{align}
    By inspection, we can see that
    \begin{align}
        W^{R}_{\ell,r} = \frac{\sqrt{N}}{\sqrt{N - \ell - r}} \cdot \Pi^{\bij}_{\ell, r+1, \gsL \gsR} \cdot E^R_{\ell,r},\label{eq:relate-WR-ER}
    \end{align}
    since multiplying by $\Pi_{\bij}$ restricts the sum in~\cref{eq:ER-def-repeated-for-WR-proof} to $x,y,L,R$ such that $(L,R) \in \calR^{2,\dist}$, $y\not\in \Im(L \cup R)$, and and $x\not \in \Dom(L \cup R)$, and for all such $x,y,L,R$, we have $\sqrt{\num(R,(x,y)) + 1} = 1$.
    Now, recall from \cref{def:ER-operator-sym-form} in \cref{claim:alt-form-of-E} that $E^{R}_{\ell,r}$ can be also be written as 
    \begin{align}
        E^R_{\ell,r} = \frac{1}{\sqrt{N}} \sum_{x,y \in [N]} \ketbra*{x}{y}_{\gsA} \otimes \Pi^{\calR}_{\ell, \gsL} \otimes \Pi^{\calR}_{r+1, \gsR} \cdot \Big(\sqrt{r+1} \cdot \ket*{x,y} \otimes \Pi_{r} \Big)_{\gsR}.
    \end{align}
    And thus, we have the following expression for $E^R_{\ell,r}$
    \begin{align}
        &E^R_{\ell,r} \cdot E^{R,\dagger}_{\ell,r} \\
        & = \Big(\frac{1}{\sqrt{N}} \sum_{x,y \in [N]} \ketbra*{x}{y}_{\gsA} \otimes \Pi^{\calR}_{\ell, \gsL} \otimes \Pi^{\calR}_{r+1, \gsR} \cdot \Big(\sqrt{r+1} \cdot \ket*{x,y} \otimes \Pi_{r} \Big)_{\gsR}\Big) \nonumber \\
        & \quad \cdot \Big(\frac{1}{\sqrt{N}} \sum_{x',y' \in [N]} \ketbra*{y'}{x'}_{\gsA} \otimes \Pi^{\calR}_{\ell, \gsL} \otimes \Big(\sqrt{r+1} \cdot \bra*{x',y'} \otimes \Pi_{r} \Big)_{\gsR} \cdot \Pi^{\calR}_{r+1, \gsR} \Big)\\ 
        & = \frac{1}{N} \sum_{x,x' \in [N]} \ketbra*{x}{x'}_{\gsA} \otimes \Pi^{\calR}_{\ell, \gsL} \otimes \Pi^{\calR}_{r+1, \gsR} \cdot \Big((r+1) \cdot \ketbra*{x}{x'}_{\darkgray{\sR^{(r+1)}_{\sX,1}}} \otimes \sum_{y} \ketbra*{y}_{\darkgray{\sR^{(r+1)}_{\sY,1}}} \otimes \Pi_{r} \Big)_{\gsR} \cdot \Pi^{\calR}_{r+1,\gsR} \\
        & = \frac{1}{N} \sum_{x,x' \in [N]} \ketbra*{x}{x'}_{\gsA} \otimes \Pi^{\calR}_{\ell, \gsL} \otimes \Pi^{\calR}_{r+1, \gsR} \cdot \Big(\sum_{i \in [r+1]} \ketbra*{x}{x'}_{\darkgray{\sR^{(r+1)}_{\sX,i}}} \otimes \sum_{y} \ketbra*{y}_{\darkgray{\sR^{(r+1)}_{\sY,i}}} \otimes \Pi_{r} \Big)_{\gsR} \cdot \Pi^{\calR}_{r+1,\gsR} \\
        &= \Pi^{\calR^2}_{\ell,r+1,\gsL \gsR} \cdot \Big(\Pi_{\ell,\gsL} \otimes \sum_{i \in [r+1]} \Pi^{\EPR}_{\darkgray{\sA,\sR^{(r+1)}_{\sX,i}}}\Big) \cdot \Pi^{\calR^2}_{\ell,r+1,\gsL \gsR}.
    \end{align}
    And thus, combining this with~\cref{eq:relate-WR-ER}, we obtain
    \begin{align}
        &W_{\ell,r}^R \cdot W_{\ell,r}^{R,\dagger}\\
        &= \frac{N}{N - \ell - r} \Pi^{\bij}_{\ell,r+1,\gsL \gsR} \cdot \Bigg( \Pi^{\calR^2}_{\ell,r+1,\gsL \gsR} \cdot \Big(\Pi_{\ell,\gsL} \otimes \sum_{i \in [r+1]} \Pi^{\EPR}_{\darkgray{\sA,\sR^{(r+1)}_{\sX,i}}}\Big) \cdot \Pi^{\calR^2}_{\ell,r+1,\gsL \gsR} \Bigg) \cdot \Pi^{\bij}_{\ell,r+1,\gsL \gsR}\\
        &= \frac{N}{N - \ell - r} \Pi^{\bij}_{\ell,r+1,\gsL \gsR} \cdot \Bigg( \Pi_{\ell,\gsL} \otimes \sum_{i \in [r+1]} \Pi^{\EPR}_{\darkgray{\sA,\sR^{(r+1)}_{\sX,i}}} \Bigg) \cdot \Pi^{\bij}_{\ell,r+1,\gsL \gsR}.
    \end{align}
    This concludes the proof.
\end{proof}

\begin{definition}
\label{def:P-DW-P-IW}
    Define
    \begin{align}
        P^{\calD(W)}_{\gsA \gsL \gsR} &= 
        \Pi^{\xydist}_{\gsL \gsR} \cdot \Bigg( \Pi^{\not\in \Dom}_{\gsA \gsL \gsR} + \sum_{\substack{\ell,r \geq 0: \\ \ell + r < N}} \frac{N}{N-\ell-r} \Pi_{\ell,\gsL} \otimes \sum_{i \in [r+1]} \Pi^{\EPR}_{\darkgray{\sA,\sR^{(r+1)}_{\sX,i}}} \Bigg) \cdot \Pi^{\xydist}_{\gsL \gsR} \\
        P^{\calI(W)}_{\gsA \gsL \gsR} &= 
        \Pi^{\xydist}_{\gsL \gsR} \cdot \Bigg( \Pi^{\not\in \Dom}_{\gsA \gsL \gsR} + \sum_{\substack{\ell,r \geq 0: \\ \ell + r < N}} \frac{N}{N-\ell-r} \Pi_{r,\gsR} \otimes \sum_{i \in [\ell+1]} \Pi^{\EPR}_{\darkgray{\sA,\sL^{(\ell+1)}_{\sY,i}}} \Bigg) \cdot \Pi^{\xydist}_{\gsL \gsR}
    \end{align}
\end{definition}

Combining~\cref{claim:expand-out-Pi-DW-Pi-IW,def:P-DW-P-IW,fact:pi-bij-R2-xydist}, we have the following corollary.

\begin{corollary}
\label{corollary:convert-Pi-to-P}
    \begin{align}
        \Pi^{\calD(W)} &= \Pi^{\calR^2} \cdot P^{\calD(W)} \cdot \Pi^{\calR^2}\\
        \Pi^{\calI(W)} &= \Pi^{\calR^2} \cdot P^{\calI(W)} \cdot \Pi^{\calR^2}
    \end{align}
\end{corollary}

\subsection{An operator upper bound}

\begin{claim}
\label{claim:operator-upper-bound-twirling}
    For any non-negative integers $\ell,r$,
    \begin{align}
        \Pi^{\xydist}_{\ell,r,\gsL \gsR} - P^{\calD(W)}_{\ell,r, \gsA \gsL \gsR} &\preceq  \frac{N}{N-\ell-r+1} \Bigg(\sum_{i \in [\ell]} \Pi^{\mathsf{eq}}_{\darkgray{\sA,\sL^{(\ell)}_{\sX,i}}} + \sum_{i \in [r]} \left(\Pi^{\mathsf{eq}}_{\darkgray{\sA,\sR^{(r)}_{\sX,i}}} - \Pi^{\EPR}_{\darkgray{\sA,\sR^{(r)}_{\sX,i}}} \right) +  2r \sqrt{\frac{\ell + r}{N}} \Id_{\gsA \gsL \gsR}\Bigg) \label{eq:psd-inequality-dist-PDW}\\
        \Pi^{\xydist}_{\ell,r,\gsL \gsR} - P^{\calI(W)}_{\ell,r, \gsA \gsL \gsR}
        &\preceq \frac{N}{N-\ell-r+1} \Bigg(\sum_{i \in [r]} \Pi^{\mathsf{eq}}_{\darkgray{\sA,\sR^{(r)}_{\sY,i}}} + \sum_{i \in [\ell]} \left(\Pi^{\mathsf{eq}}_{\darkgray{\sA,\sL^{(\ell)}_{\sY,i}}} - \Pi^{\EPR}_{\darkgray{\sA,\sL^{(\ell)}_{\sY,i}}} \right) + 2\ell \sqrt{\frac{\ell + r}{N}} \Id_{\gsA \gsL \gsR}\Bigg) \label{eq:psd-inequality-dist-PIW}
    \end{align}
\end{claim}

\begin{proof}
    We will prove the first inequality, and the second will follow from a symmetric argument. We begin by writing out $P^{\calD(W)}_{\ell,r, \gsA \gsL \gsR}$ as
    \begin{align}
        P^{\calD(W)}_{\ell,r, \gsA \gsL \gsR} &= \Pi^{\xydist}_{\gsL \gsR} \cdot \Bigg( \Pi^{\not\in \Dom}_{\gsA \gsL \gsR} + \sum_{\substack{\ell',r' \geq 0: \\ \ell' + r' < N}} \frac{N}{N-\ell'-r'} \Pi_{\ell',\gsL} \otimes \sum_{i \in [r'+1]} \Pi^{\EPR}_{\darkgray{\sA,\sR^{(r'+1)}_{\sX,i}}} \Bigg) \cdot \Pi^{\xydist}_{\gsL \gsR} \cdot \Pi_{\ell,r, \gsL \gsR} \\
        &=\Pi^{\xydist}_{\ell,r, \gsL \gsR} \cdot \Bigg( \Pi^{\not\in \Dom}_{\ell,r, \gsA \gsL \gsR} + \frac{N}{N-\ell- r+ 1} \Pi_{\ell,\gsL} \otimes \sum_{i \in [r]} \Pi^{\EPR}_{\darkgray{\sA,\sR^{(r)}_{\sX,i}}} \Bigg) \cdot \Pi^{\xydist}_{\ell,r,\gsL \gsR},
    \end{align}
    where the second equality uses the fact that $\Pi_{\ell,r,\gsL \gsR}$ commutes with every other projector in the above expression. Note that in the special case where $r = 0$, the above expression simplifies to 
    \begin{align}
        \Pi^{\xydist}_{\ell,r, \gsL \gsR} \cdot \Pi^{\not\in \Dom}_{\ell,r, \gsA \gsL \gsR} \cdot \Pi^{\xydist}_{\ell,r,\gsL \gsR}.
    \end{align}
    Next, we use our expression for $P^{\calD(W)}_{\ell,r,\gsA \gsL \gsR}$ to expand out
    \begin{align}
        \Pi^{\xydist}_{\ell,r,\gsL \gsR} - P^{\calD(W)}_{\ell,r,\gsA \gsL \gsR}
        &= \Pi^{\xydist}_{\ell,r,\gsL \gsR} \cdot \Bigg(\Pi_{\ell,r, \gsL \gsR} - \Pi^{\not\in \Dom}_{\ell,r,\gsA \gsL \gsR} - \frac{N}{N-\ell-r+1} \sum_{i \in [r]} \Pi^{\EPR}_{\darkgray{\sA,\sR^{(r)}_{\sX,i}}} \Bigg) \cdot \Pi^{\xydist}_{\ell,r,\gsL \gsR}. \label{eq:expand-Pi-dist-minus-P-DW}
    \end{align}
    Using the definition of $\Pi_{\ell,r}^{\not\in \Dom}$, we have
    \begin{align}
        \Pi_{\ell,r, \gsL \gsR} - \Pi^{\not\in \Dom}_{\ell,r, \gsA \gsL \gsR} \preceq \sum_{i \in [\ell]} \Pi^{\mathsf{eq}}_{\darkgray{\sA,\sL^{(\ell)}_{\sX,i}}} + \sum_{i \in [r]} \Pi^{\mathsf{eq}}_{\darkgray{\sA,\sR^{(r)}_{\sX,i}}}.
    \end{align}
    This inequality holds because any state in the image of $\Pi_{\ell,r, \gsL \gsR} - \Pi^{\not\in \Dom}$ must have a collision between the $\sA$ register and at least one of the registers $\{\sL^{(\ell)}_{\sX,i}\}_{i \in [\ell]} \cup \{\sR^{(r)}_{\sX,i}\}_{i \in [r]}$, and will therefore be in the image of at least one of the projectors on the right-hand-side. Plugging this inequality into~\cref{eq:expand-Pi-dist-minus-P-DW}, we have
    \begin{align}
        &\Pi^{\xydist}_{\ell,r,\gsL \gsR} - P^{\calD(W)}_{\ell,r,\gsA \gsL \gsR}\\
        &\preceq \Pi^{\xydist}_{\ell,r,\gsL \gsR} \cdot \Bigg(\sum_{i \in [\ell]} \Pi^{\mathsf{eq}}_{\darkgray{\sA,\sL^{(\ell)}_{\sX,i}}} + \sum_{i \in [r]} \Pi^{\mathsf{eq}}_{\darkgray{\sA,\sR^{(r)}_{\sX,i}}} - \frac{N}{N-\ell-r+1} \sum_{i \in [r]} \Pi^{\EPR}_{\darkgray{\sA,\sR^{(r)}_{\sX,i}}} \Bigg) \cdot \Pi^{\xydist}_{\ell,r,\gsL \gsR}\\
        &\preceq \frac{N}{N-\ell-r+1} \underbrace{\Bigg( \Pi^{\xydist}_{\ell,r,\gsL \gsR} \cdot \Big(\sum_{i \in [\ell]} \Pi^{\mathsf{eq}}_{\darkgray{\sA,\sL^{(\ell)}_{\sX,i}}}\Big) \cdot \Pi^{\xydist}_{\ell,r,\gsL \gsR} \Bigg)}_{\mathrm{Term}_1} \nonumber \\
        & \quad + \frac{N}{N-\ell-r+1} \underbrace{\Bigg(\Pi^{\xydist}_{\ell,r,\gsL \gsR} \cdot \Big(\sum_{i \in [r]} \Big(\Pi^{\mathsf{eq}}_{\darkgray{\sA,\sR^{(r)}_{\sX,i}}} -  \Pi^{\EPR}_{\darkgray{\sA,\sR^{(r)}_{\sX,i}}}\Big) \Big) \cdot \Pi^{\xydist}_{\ell,r,\gsL \gsR} \Bigg)}_{\mathrm{Term}_2} \label{eq:psd-inequality-dist-PDW-proof-1}.
    \end{align}
    Comparing (\ref{eq:psd-inequality-dist-PDW-proof-1}) to the right-hand side of~\cref{eq:psd-inequality-dist-PDW}, it suffices to prove that:
    \begin{align}
        \mathrm{Term}_1 &\preceq \sum_{i \in [\ell]} \Pi^{\mathsf{eq}}_{\darkgray{\sA,\sL^{(\ell)}_{\sX,i}}}, \label{eq:term-1-psd-bound}\\
        \mathrm{Term}_2 &\preceq \sum_{i \in [r]} \left(\Pi^{\mathsf{eq}}_{\darkgray{\sA,\sR^{(r)}_{\sX,i}}} - \Pi^{\EPR}_{\darkgray{\sA,\sR^{(r)}_{\sX,i}}} \right) + 2 r \sqrt{\frac{\ell+r}{N}} \Id_{\gsA \gsL \gsR}.  \label{eq:term-2-psd-bound}
    \end{align}
    \paragraph{Bounding the first term.} To prove~\cref{eq:term-1-psd-bound}, it suffices to prove for each $i \in [\ell]$,
    \begin{align}
        \Pi^{\xydist}_{\ell,r,\gsL \gsR} \cdot  \Pi^{\mathsf{eq}}_{\darkgray{\sA,\sL^{(\ell)}_{\sX,i}}} \cdot \Pi^{\xydist}_{\ell,r,\gsL \gsR} &\preceq \Pi^{\mathsf{eq}}_{\darkgray{\sA,\sL^{(\ell)}_{\sX,i}}}.
    \end{align}
    This holds because $\Pi^{\mathsf{eq}}_{\darkgray{\sA,\sL^{(\ell)}_{\sX,i}}}$ commutes with $\Pi^{\xydist}_{\ell,r,\gsL \gsR}$ (since both are diagonal in the standard basis) and thus:
    \begin{align}
        \Pi^{\xydist}_{\ell,r,\gsL \gsR} \cdot  \Pi^{\mathsf{eq}}_{\darkgray{\sA,\sL^{(\ell)}_{\sX,i}}} \cdot \Pi^{\xydist}_{\ell,r,\gsL \gsR} &= \Pi^{\mathsf{eq}}_{\darkgray{\sA,\sL^{(\ell)}_{\sX,i}}} \cdot \Pi^{\xydist}_{\ell,r,\gsL \gsR} \cdot \Pi^{\mathsf{eq}}_{\darkgray{\sA,\sL^{(\ell)}_{\sX,i}}} \preceq \Pi^{\mathsf{eq}}_{\darkgray{\sA,\sL^{(\ell)}_{\sX,i}}} \cdot \Pi_{\ell,r,\gsL \gsR} \cdot \Pi^{\mathsf{eq}}_{\darkgray{\sA,\sL^{(\ell)}_{\sX,i}}} = \Pi^{\mathsf{eq}}_{\darkgray{\sA,\sL^{(\ell)}_{\sX,i}}}. \label{eq:piD-piDPfwpiD}
    \end{align}

    \paragraph{Bounding the second term.} To prove~\cref{eq:term-2-psd-bound}, it suffices to prove for each $i \in [r]$,
    \begin{align}
        \Pi^{\xydist}_{\ell,r,\gsL \gsR} \cdot  \Big(\Pi^{\mathsf{eq}}_{\darkgray{\sA,\sR^{(r)}_{\sX,i}}} -  \Pi^{\EPR}_{\darkgray{\sA,\sR^{(r)}_{\sX,i}}}\Big) \cdot \Pi^{\xydist}_{\ell,r,\gsL \gsR} \preceq \left(\Pi^{\mathsf{eq}}_{\darkgray{\sA,\sR^{(r)}_{\sX,i}}} - \Pi^{\EPR}_{\darkgray{\sA,\sR^{(r)}_{\sX,i}}} \right) + 2 \sqrt{\frac{\ell + r}{N}} \Id_{\gsA \gsL \gsR} \label{eq:psd-ineq-goal-term-2}
    \end{align}
    Note that $\Pi^{\xydist}_{\ell,r,\gsL \gsR}$ and $\Big(\Pi^{\mathsf{eq}}_{\darkgray{\sA,\sR^{(r)}_{\sX,i}}} -  \Pi^{\EPR}_{\darkgray{\sA,\sR^{(r)}_{\sX,i}}}\Big)$ do not commute, so we cannot simply apply the argument we used to bound the first term. However, these two operators \emph{almost} commute. In particular, 
    \begin{align}
        \norm{\Pi^{\xydist}_{\ell,r, \gsL \gsR} \cdot \Big(\Pi^{\mathsf{eq}}_{\darkgray{\sA,\sR^{(r)}_{\sX,i}}} -  \Pi^{\EPR}_{\darkgray{\sA,\sR^{(r)}_{\sX,i}}}\Big)  - \Big(\Pi^{\mathsf{eq}}_{\darkgray{\sA,\sR^{(r)}_{\sX,i}}} -  \Pi^{\EPR}_{\darkgray{\sA,\sR^{(r)}_{\sX,i}}}\Big) \cdot \Pi^{\xydist}_{\ell,r, \gsL \gsR}}_{\opnorm} \leq \sqrt{\frac{\ell+r}{N}}. \label{eq:Pi-eq-epr-dist-almost-commute}
    \end{align}
    We prove~\cref{eq:Pi-eq-epr-dist-almost-commute} in~\cref{claim:approx-commute-EPR}, which follows this proof. To simplify notation for the following steps, let us write $A \coloneqq \Pi^{\xydist}_{\ell,r,\gsL \gsR}$ and $B \coloneqq \Pi^{\mathsf{eq}}_{\darkgray{\sA,\sR^{(r)}_{\sX,i}}} -  \Pi^{\EPR}_{\darkgray{\sA,\sR^{(r)}_{\sX,i}}}$. Note that $B$ is a projector because $\Pi^{\EPR}_{\darkgray{\sA,\sR^{(r)}_{\sX,i}}}$ projects onto a subspace of the image of $\Pi^{\mathsf{eq}}_{\darkgray{\sA,\sR^{(r)}_{\sX,i}}}$. 
    
    Using the definition of operator norm, we have
    \begin{align}
        A \cdot B \cdot A - B \cdot A \cdot B \preceq \norm{A \cdot B \cdot A - B \cdot A \cdot B}_{\opnorm} \cdot \Id.
    \end{align}
    Adding $B \cdot A \cdot B$ to both sides, we get
    \begin{align}
        A \cdot B \cdot A &\preceq B \cdot A \cdot B + \norm{A \cdot B \cdot A - B \cdot A \cdot B}_{\opnorm} \cdot \Id \\
        &\preceq B + \norm{A \cdot B \cdot A - B \cdot A \cdot B}_{\opnorm} \cdot \Id, \label{eq:ABA-psd-inequality}
    \end{align}
    where the second inequality uses the fact that $B \cdot A \cdot B \preceq B \cdot \Id \cdot B = B$ for projectors $A$ and $B$. Now, 
    \begin{align}
        \norm{A \cdot B \cdot A - B \cdot A \cdot B}_{\opnorm} &\leq \norm{A \cdot B \cdot A - A \cdot B \cdot A \cdot B}_{\opnorm} +  \norm{A \cdot B \cdot A \cdot B - B \cdot A \cdot B}_{\opnorm} \tag{triangle inequality}\\
        &= \norm{A \cdot B \cdot (B \cdot A - A \cdot B)}_{\opnorm} +  \norm{(A \cdot B - B \cdot A) \cdot A \cdot B}_{\opnorm} \tag{since $A^2 = A$ and $B^2 = B$}\\
        &\leq \norm{A \cdot B}_{\opnorm} \cdot \norm{B \cdot A - A \cdot B}_{\opnorm} +  \norm{A \cdot B - B \cdot A}_{\opnorm} \cdot \norm{A \cdot B}_{\opnorm} \\
        &\leq \norm{B \cdot A - A \cdot B}_{\opnorm} +  \norm{A \cdot B - B \cdot A}_{\opnorm} \\
        & \leq 2 \sqrt{\frac{\ell + r}{N}} \tag{by~\cref{claim:approx-commute-EPR}}
    \end{align}
    Plugging this inequality into~\cref{eq:ABA-psd-inequality}, and plugging in $A = \Pi^{\xydist}_{\ell,r,\gsL \gsR}$ and $B = \Pi^{\mathsf{eq}}_{\darkgray{\sA,\sR^{(r)}_{\sX,i}}} -  \Pi^{\EPR}_{\darkgray{\sA,\sR^{(r)}_{\sX,i}}}$ yields~\cref{eq:psd-ineq-goal-term-2}. This completes the proof.
\end{proof}

We now prove~\cref{claim:approx-commute-EPR}, which we used in the previous proof.

\begin{claim} 
\label{claim:approx-commute-EPR}
    For any integers $\ell,r \geq 0$ such that $\ell + r \leq N$ and any index $i \in [r]$ (if such an $i$ exists\footnote{Note that $[r] \coloneqq \{1,2,\dots,r\}$, so no $i$ exists when $r = 0$. But in this case, the bound is trivially satisfied.}), we have
    \begin{align}
        \norm{\Pi^{\xydist}_{\ell,r, \gsL \gsR} \cdot \Big(\Pi^{\mathsf{eq}}_{\darkgray{\sA,\sR^{(r)}_{\sX,i}}} -  \Pi^{\EPR}_{\darkgray{\sA,\sR^{(r)}_{\sX,i}}}\Big)  - \Big(\Pi^{\mathsf{eq}}_{\darkgray{\sA,\sR^{(r)}_{\sX,i}}} -  \Pi^{\EPR}_{\darkgray{\sA,\sR^{(r)}_{\sX,i}}}\Big) \cdot \Pi^{\xydist}_{\ell,r, \gsL \gsR}}_{\opnorm} \leq \sqrt{\frac{\ell+r}{N}}. \label{eq:Pi-eq-epr-dist-almost-commute-repeated}
    \end{align}
    Similarly, for any integers $\ell, r \geq 0$ such that $\ell + r \leq N$, and any index $i \in [\ell]$ (if such an $i$ exists), we have
    \begin{align}
        \norm{\Pi^{\xydist}_{\ell,r, \gsL \gsR} \cdot \Big(\Pi^{\mathsf{eq}}_{\darkgray{\sA,\sL^{(\ell)}_{\sY,i}}} -  \Pi^{\EPR}_{\darkgray{\sA,\sL^{(\ell)}_{\sY,i}}}\Big)  - \Big(\Pi^{\mathsf{eq}}_{\darkgray{\sA,\sL^{(\ell)}_{\sY,i}}} -  \Pi^{\EPR}_{\darkgray{\sA,\sL^{(\ell)}_{\sY,i}}}\Big) \cdot \Pi^{\xydist}_{\ell,r, \gsL \gsR}}_{\opnorm} \leq \sqrt{\frac{\ell+r}{N}}. \label{eq:Pi-eq-epr-dist-almost-commute-repeated-alt}
    \end{align}
\end{claim}

\begin{proof}
    We will prove the first inequality, and the second will follow from a symmetric argument. Assume without loss of generality that $i = 1$. Since $\Pi^{\xydist}_{\ell,r, \gsL \gsR}$ and $\Pi^{\mathsf{eq}}_{\darkgray{\sA,\sR^{(r)}_{\sX,1}}}$ commute,
    \begin{align}
        &\Pi^{\xydist}_{\ell,r, \gsL \gsR} \cdot \Big(\Pi^{\mathsf{eq}}_{\darkgray{\sA,\sR^{(r)}_{\sX,1}}} -  \Pi^{\EPR}_{\darkgray{\sA,\sR^{(r)}_{\sX,1}}}\Big)  - \Big(\Pi^{\mathsf{eq}}_{\darkgray{\sA,\sR^{(r)}_{\sX,1}}} -  \Pi^{\EPR}_{\darkgray{\sA,\sR^{(r)}_{\sX,1}}}\Big) \cdot \Pi^{\xydist}_{\ell,r, \gsL \gsR}\\
        &= \Pi^{\EPR}_{\darkgray{\sA,\sR^{(r)}_{\sX,1}}} \cdot \Pi^{\xydist}_{\ell,r, \gsL \gsR} - \Pi^{\xydist}_{\ell,r, \gsL \gsR} \cdot \Pi^{\EPR}_{\darkgray{\sA,\sR^{(r)}_{\sX,1}}}.
    \end{align}
    We will write down an explicit expression for the operator $\Pi^{\EPR}_{\darkgray{\sA,\sR^{(r)}_{\sX,1}}} \cdot \Pi^{\xydist}_{\ell,r, \gsL \gsR} - \Pi^{\xydist}_{\ell,r, \gsL \gsR} \cdot \Pi^{\EPR}_{\darkgray{\sA,\sR^{(r)}_{\sX,1}}}$, which we will then use to bound the operator norm.

    We can expand out $\Pi^{\xydist}_{\ell,r, \gsL \gsR}$ as
    \begin{align}
        \Pi^{\xydist}_{\ell,r, \gsL \gsR} &= \sum_{\substack{(x,x') \in [N]^{\ell+r}_{\dist},\\ (y,y') \in [N]^{\ell+r}_{\dist}}} \ketbra*{x,y}_{\darkgray{\sL^{(\ell)}}} \otimes \ketbra*{x',y'}_{\darkgray{\sR^{(r)}}},
    \end{align}
    where in the above sum, $x,y \in [N]^\ell$, $x',y' \in [N]^r$, and $\ket*{x,y}_{\darkgray{\sL^{(\ell)}}} = \ket*{x_1,y_1,\dots,x_\ell,y_\ell}_{\darkgray{\sL^{(\ell)}}}$. 

    We can then expand out
    \begin{align}
        &(\Id_{\gsA} \otimes \Pi^{\xydist}_{\ell,r, \gsL \gsR}) \cdot \Pi^{\EPR}_{\darkgray{\sA,\sR^{(r)}_{\sX,1}}} \\
        &= \Bigg(\sum_{\substack{(x,x') \in [N]^{\ell+r}_{\dist},\\ (y,y') \in [N]^{\ell+r}_{\dist}}} \ketbra*{x,y}_{\darkgray{\sL^{(\ell)}}} \otimes \ketbra*{x',y'}_{\darkgray{\sR^{(r)}}}\Bigg) \cdot \Bigg(\frac{1}{N} \sum_{u,v \in [N]} \ketbra*{u}{v}_{\gsA} \otimes \ketbra*{u}{v}_{\darkgray{\sR^{(r)}_{\sX,1}}} \Bigg)\\
        &= \sum_{\substack{(x,x'_{[2,r]}) \in [N]^{\ell+r-1}_{\dist},\\ (y,y') \in [N]^{\ell+r}_{\dist}}} \ketbra*{x,y}_{\darkgray{\sL^{(\ell)}}} \otimes \ketbra*{x'_{[2:r]},y'}_{\darkgray{\sR \setminus \sR^{(r)}_{\sX,1}}} \otimes \frac{1}{N} \sum_{\substack{v \in [N],\\ u\not\in \{x,x'_{[2,r]}\}}} \ketbra*{u}{v}_{\gsA} \otimes \ketbra*{u}{v}_{\darkgray{\sR^{(r)}_{\sX,1}}}, \label{eq:explicit-dist-epr}
    \end{align}
    where $x'_{[2,r]} \in [N]^{r-1}$, the notation $u \not\in \{x,x'_{[2,r]}\}$ means $u$ must be distinct from each element of $x$ and $x'_{[2,r]}$, and $\sR \setminus \sR^{(r)}_{\sX,1}$ refers to all of the registers $\sR = (\sR^{(r)}_{\sX,1},\sR^{(r)}_{\sY_1},\dots, \sR^{(r)}_{\sX,r},\sR^{(r)}_{\sY_r})$ \emph{except} for $\sR^{(r)}_{\sX,1}$.
    
    Next, to get an explicit expression for $\Pi^{\EPR}_{\darkgray{\sA,\sR^{(r)}_{\sX,1}}} \cdot \Pi^{\xydist}_{\ell,r, \gsL \gsR}$, we can just take the conjugate transpose of~\cref{eq:explicit-dist-epr}. By exploiting symmetry, we can write the result in such a way that it looks nearly identical to (\ref{eq:explicit-dist-epr}), except for the part highlighted in red.
    \begin{align}
        &\Pi^{\EPR}_{\darkgray{\sA,\sR^{(r)}_{\sX,1}}} \cdot (\Id_{\gsA} \otimes \Pi^{\xydist}_{\ell,r, \gsL \gsR}) \\
        &=\sum_{\substack{(x,x'_{[2,r]}) \in [N]^{\ell+r-1}_{\dist},\\ (y,y') \in [N]^{\ell+r}_{\dist}}} \ketbra*{x,y}_{\darkgray{\sL^{(\ell)}}} \otimes \ketbra*{x'_{[2:r]},y'}_{\darkgray{\sR \setminus \sR^{(r)}_{\sX,1}}} \otimes \frac{1}{N} \sum_{\substack{\textcolor{red}{u \in [N],}\\ \textcolor{red}{v\not\in \{x,x'_{[2,r]}\}}}} \ketbra*{u}{v}_{\gsA} \otimes \ketbra*{u}{v}_{\darkgray{\sR^{(r)}_{\sX,1}}}. \label{eq:explicit-epr-dist}
    \end{align}
    
    Subtracting (\ref{eq:explicit-dist-epr}) from (\ref{eq:explicit-epr-dist}), we get
    \begin{align}
        & \Pi^{\EPR}_{\darkgray{\sA,\sR^{(r)}_{\sX,1}}} \cdot (\Id_{\gsA} \otimes \Pi^{\xydist}_{\ell,r, \gsL \gsR}) -  (\Id_{\gsA} \otimes \Pi^{\xydist}_{\ell,r, \gsL \gsR}) \cdot \Pi^{\EPR}_{\darkgray{\sA,\sR^{(r)}_{\sX,1}}} \\
        &=\sum_{\substack{(x,x'_{[2,r]}) \in [N]^{\ell+r-1}_{\dist},\\ (y,y') \in [N]^{\ell+r}_{\dist}}} \ketbra*{x,y}_{\darkgray{\sL^{(\ell)}}} \otimes \ketbra*{x'_{[2:r]},y'}_{\darkgray{\sR \setminus \sR^{(r)}_{\sX,1}}} \nonumber \\
        &\quad \otimes \frac{1}{N} \sum_{\substack{u \in [N],\\ v\not\in \{x,x'_{[2,r]}\}}} \ketbra*{u}{v}_{\gsA} \otimes \ketbra*{u}{v}_{\darkgray{\sR^{(r)}_{\sX,1}}} - \frac{1}{N} \sum_{\substack{v \in [N],\\ u\not\in \{x,x'_{[2,r]}\}}} \ketbra*{u}{v}_{\gsA} \otimes \ketbra*{u}{v}_{\darkgray{\sR^{(r)}_{\sX,1}}}.
        \end{align}
        Since this is a block diagonal matrix with blocks indexed by $x,y,x'_{[2:r]},y'$, the operator norm is
        \begin{align}
            &\norm{\Pi^{\EPR}_{\darkgray{\sA,\sR^{(r)}_{\sX,1}}} \cdot (\Id_{\gsA} \otimes \Pi^{\xydist}_{\ell,r, \gsL \gsR}) -  (\Id_{\gsA} \otimes \Pi^{\xydist}_{\ell,r, \gsL \gsR}) \cdot \Pi^{\EPR}_{\darkgray{\sA,\sR^{(r)}_{\sX,1}}}}_{\opnorm}  \\
            &= \max_{\substack{(x,x'_{[2,r]}) \in [N]^{\ell+r-1}_{\dist},\\ (y,y') \in [N]^{\ell+r}_{\dist}}} \Bigg\lVert \frac{1}{N} \sum_{\substack{u \in [N],\\ v\not\in \{x,x'_{[2,r]}\}}} \ketbra*{u}{v}_{\gsA} \otimes \ketbra*{u}{v}_{\darkgray{\sR^{(r)}_{\sX,1}}} - \frac{1}{N} \sum_{\substack{v \in [N],\\ u\not\in \{x,x'_{[2,r]}\}}} \ketbra*{u}{v}_{\gsA} \otimes \ketbra*{u}{v}_{\darkgray{\sR^{(r)}_{\sX,1}}} \Bigg\rVert_{\opnorm} \label{eq:opnorm-blocks-uvvu}.
        \end{align}
        Fix any choice of $(x,x'_{[2,r]}) \in [N]^{\ell+r-1}_{\dist}, (y,y') \in [N]^{\ell+r}_{\dist}$. We can rewrite
        \begin{align}
            &\frac{1}{N} \sum_{\substack{u \in [N],\\ v\not\in \{x,x'_{[2,r]}\}}} \ketbra*{u}{v}_{\gsA} \otimes \ketbra*{u}{v}_{\darkgray{\sR^{(r)}_{\sX,1}}} - \frac{1}{N} \sum_{\substack{v \in [N],\\ u\not\in \{x,x'_{[2,r]}\}}} \ketbra*{u}{v}_{\gsA} \otimes \ketbra*{u}{v}_{\darkgray{\sR^{(r)}_{\sX,1}}} \\
            &= \frac{1}{N} \sum_{\substack{u \in  \{x,x'_{[2,r]}\},\\ v\not\in \{x,x'_{[2,r]}\}}} \ketbra*{u}{v}_{\gsA} \otimes \ketbra*{u}{v}_{\darkgray{\sR^{(r)}_{\sX,1}}} - \frac{1}{N} \sum_{\substack{u\not\in \{x,x'_{[2,r]}\},\\ v \in  \{x,x'_{[2,r]}\}}} \ketbra*{u}{v}_{\gsA} \otimes \ketbra*{u}{v}_{\darkgray{\sR^{(r)}_{\sX,1}}}\\
            &= \frac{\sqrt{(N-r-\ell+1)(r+\ell-1)}}{N} \Big( \ketbra*{\phi}{\psi} - \ketbra*{\psi}{\phi} \Big),
        \end{align}
        where $\ket*{\psi}$ and $\ket*{\phi}$ are the following vectors
        \begin{align}
            \ket*{\psi} &\coloneqq \frac{1}{\sqrt{N-r-\ell+1}} \sum_{ u\not\in \{x,x'_{[2,r]}\}} \ket*{u}_{\gsA} \ket*{u}_{\darkgray{\sR^{(r)}_{\sX,1}}},\\
            \ket*{\phi} &\coloneqq \frac{1}{\sqrt{r + \ell - 1}} \sum_{ u\in \{x,x'_{[2,r]}\}} \ket*{u}_{\gsA} \ket*{u}_{\darkgray{\sR^{(r)}_{\sX,1}}},
        \end{align}
        Note that if $r = 1$ and $\ell = 0$, these vectors have norm $0$ and we are done. For all other choices of $r \geq 1$, $\ell \geq 0$ such that $r+ \ell \leq N$, these are orthogonal unit vectors, and thus $\ketbra*{\psi}{\phi} - \ketbra*{\phi}{\psi}$ has operator norm $1$. It follows that the overall operator norm is at most 
        \begin{align}
        \frac{\sqrt{(N-r-\ell+1)(r+\ell-1)}}{N} \leq \sqrt{\frac{\ell + r}{N}}.
    \end{align}
    which completes the proof.
\end{proof}

\subsection{An intermediate lemma on $2$-design twirling}

\begin{claim}[Twirling]
\label{claim:twirling-ell-r}
    For any unitary $2$-design $\frakD$ and any non-negative integers $\ell,r$,
    \begin{align}
    &\norm{\E_{C,D \gets \frakD} (C_{\gsA} \otimes Q[C,D]_{\gsL \gsR})^\dagger \cdot \Bigg(\sum_{i \in [\ell]} \Pi^{\mathsf{eq}}_{\darkgray{\sA,\sL^{(\ell)}_{\sX,i}}} + \sum_{i \in [r]} \left(\Pi^{\mathsf{eq}}_{\darkgray{\sA,\sR^{(r)}_{\sX,i}}} - \Pi^{\EPR}_{\darkgray{\sA,\sR^{(r)}_{\sX,i}}} \right) \Bigg) \cdot (C_{\gsA} \otimes Q[C,D]_{\gsL \gsR})}_{\opnorm} \nonumber \\
    &\leq \frac{2\ell + r}{N+1},\\
    &\norm{\E_{C,D \gets \frakD} (D^\dagger_{\gsA} \otimes Q[C,D]_{\gsL \gsR})^\dagger \cdot \Bigg( \sum_{i \in [r]} \Pi^{\mathsf{eq}}_{\darkgray{\sA,\sR^{(r)}_{\sX,i}}} + \sum_{i \in [\ell]} \left(\Pi^{\mathsf{eq}}_{\darkgray{\sA,\sL^{(\ell)}_{\sX,i}}} - \Pi^{\EPR}_{\darkgray{\sA,\sL^{(\ell)}_{\sX,i}}} \right) \Bigg) \cdot (D^\dagger_{\gsA} \otimes Q[C,D]_{\gsL \gsR})}_{\opnorm} \nonumber \\
    &\leq \frac{2r + \ell}{N+1} .
\end{align}
\end{claim}

Before we prove~\cref{claim:twirling-ell-r}, let us recall the following claims that were proven in the preliminaries.

\begin{claim}[\cref{claim:equality-clifford-twirl}, restated]
    For any $n$-qubit unitary $2$-design $\frakD$,
    \begin{align}
        \E_{U\gets \frakD} \Big[ (U \otimes U)^\dagger \cdot \Pi^{\mathsf{eq}} \cdot (U \otimes U)\Big] = \frac{2}{N+1} \cdot \Pi_{\ssym}^{N, 2}.
    \end{align}
\end{claim}

\begin{claim}[\cref{claim:equality-clifford-conjugated-twirl}, restated]
    For any $n$-qubit unitary $2$-design $\frakD$,
    \begin{align}
        \E_{U\gets \frakD} \Big[ (U \otimes \overline{U})^\dagger \cdot \Big( \Pi^{\mathsf{eq}} - \Pi^{\EPR} \Big) \cdot (U \otimes \overline{U})\Big] = \frac{1}{N+1}\cdot (\Id - \Pi^{\EPR}).
    \end{align}
\end{claim}

\begin{proof}[Proof of~\cref{claim:twirling-ell-r}]
    We prove the first inequality, and the proof for the second one is symmetric. By the triangle inequality, 
    \begin{align}
        & \norm{\E_{C,D \gets \frakD} (C_{\gsA} \otimes Q[C,D]_{\gsL \gsR})^\dagger \cdot \Bigg(\sum_{i \in [\ell]} \Pi^{\mathsf{eq}}_{\darkgray{\sA,\sL^{(\ell)}_{\sX,i}}} + \sum_{i \in [r]} \left(\Pi^{\mathsf{eq}}_{\darkgray{\sA,\sR^{(r)}_{\sX,i}}} - \Pi^{\EPR}_{\darkgray{\sA,\sR^{(r)}_{\sX,i}}} \right) \Bigg) \cdot (C_{\gsA} \otimes Q[C,D]_{\gsL \gsR})}_{\opnorm}\\
        &\leq \sum_{i \in [\ell]} \norm{\E_{C,D \gets \frakD} (C_{\gsA} \otimes Q[C,D]_{\gsL \gsR})^\dagger \cdot \Pi^{\mathsf{eq}}_{\darkgray{\sA,\sL^{(\ell)}_{\sX,i}}} \cdot (C_{\gsA} \otimes Q[C,D]_{\gsL \gsR})}_{\opnorm} \nonumber \\
        & \quad + \sum_{i \in [r]} \norm{\E_{C,D \gets \frakD} (C_{\gsA} \otimes Q[C,D]_{\gsL \gsR})^\dagger \cdot \Big(\Pi^{\mathsf{eq}}_{\darkgray{\sA,\sR^{(r)}_{\sX,i}}} - \Pi^{\EPR}_{\darkgray{\sA,\sR^{(r)}_{\sX,i}}} \Big) \cdot (C_{\gsA} \otimes Q[C,D]_{\gsL \gsR})}_{\opnorm} \label{eq:big-twirling-lemma-first-step}
    \end{align}
    Plugging in
    \begin{align}
        Q[C,D]_{\gsL \gsR} = (C \otimes D^T)^{\otimes *}_{\gsL} \otimes (\overline{C} \otimes D^{\dagger})^{\otimes *}_{\gsR},
    \end{align}
    we can rewrite (\ref{eq:big-twirling-lemma-first-step}) as
    \begin{align}
        (\ref{eq:big-twirling-lemma-first-step}) &= \sum_{i \in [\ell]} \norm{\E_{C \gets \frakD} (C_{\gsA} \otimes C_{\darkgray{\sL^{(\ell)}_{\sX,i}}})^\dagger \cdot \Pi^{\mathsf{eq}}_{\darkgray{\sL^{(\ell)}_{\sX,i}}} \cdot (C_{\gsA} \otimes C_{\darkgray{\sA,\sL^{(\ell)}_{\sX,i}}})}_{\opnorm} \nonumber \\
        & \quad + \sum_{i \in [r]} \norm{\E_{C \gets \frakD} (C_{\gsA} \otimes \overline{C}_{\darkgray{\sR^{(r)}_{\sX,i}}})^\dagger \cdot \Big(\Pi^{\mathsf{eq}}_{\darkgray{\sA,\sR^{(r)}_{\sX,i}}} - \Pi^{\EPR}_{\darkgray{\sA,\sR^{(r)}_{\sX,i}}} \Big) \cdot (C_{\gsA} \otimes \overline{C}_{\darkgray{\sR^{(r)}_{\sX,i}}})}_{\opnorm} \label{eq:big-twirling-lemma-second-step}\\
        &\leq \frac{2 \ell}{N+1} + \frac{r}{N+1} \tag{by~\cref{claim:equality-clifford-twirl,claim:equality-clifford-conjugated-twirl}}.
    \end{align}
    This completes the proof. 
\end{proof}

\subsection{Finishing the proof of~\cref{lem:twirling-strongPRU}}

We now prove~\cref{lem:twirling-strongPRU}, which we state again for convenience.

\begin{lemma}[\cref{lem:twirling-strongPRU}, restated]
    For any unitary $2$-design $\frakD$ and integer $0 \leq t \leq N/2$, we have
    \begin{align}
        \norm{  \E_{C,D \gets \frakD} (C_{\gsA} \otimes Q[C,D]_{\gsL \gsR})^\dagger \cdot \Big( \Pi^{\bij}_{\leq t, \gsL \gsR} - \Pi^{\calD(W)}_{\leq t, \gsA \gsL \gsR}\Big) \cdot (C_{\gsA} \otimes Q[C,D]_{\gsL \gsR}) }_{\opnorm} &\leq 6t \sqrt{\frac{t}{N}}, \label{eq:twirling-main-eq-1}\\
        \norm{  \E_{C,D \gets \frakD} (D^\dagger_{\gsA} \otimes Q[C,D]_{\gsL \gsR})^\dagger \cdot \Big( \Pi^{\bij}_{\leq t, \gsL \gsR} - \Pi^{\calI(W)}_{\leq t, \gsA \gsL \gsR}\Big) \cdot (D^\dagger_{\gsA} \otimes Q[C,D]_{\gsL \gsR}) }_{\opnorm} &\leq 6t \sqrt{\frac{t}{N}}, \label{eq:twirling-main-eq-2}
    \end{align}
\end{lemma}

\begin{proof}
    Using~\cref{fact:pi-bij-R2-xydist,corollary:convert-Pi-to-P}, we have
    \begin{align}
        \Pi^{\bij}_{\gsL \gsR} &= \Pi^{\calR^2}_{\gsL \gsR} \cdot \Pi^{\xydist}_{\gsL \gsR} \cdot \Pi^{\calR^2}_{\gsL \gsR},\\
        \Pi^{\calD(W)}_{\gsA \gsL \gsR} &= \Pi^{\calR^2}_{\gsL \gsR} \cdot P^{\calD(W)}_{\gsA \gsL \gsR} \cdot \Pi^{\calR^2}_{\gsL \gsR},
    \end{align}
    Since $\Pi^{\calR^2}$ commutes with $\Pi_{\leq t}$, this implies
    \begin{align}
        \Pi^{\bij}_{\leq t, \gsL \gsR} &= \Pi^{\calR^2}_{\gsL \gsR} \cdot \Pi^{\xydist}_{\leq t, \gsL \gsR} \cdot \Pi^{\calR^2}_{\gsL \gsR},\\
        \Pi^{\calD(W)}_{\leq t,\gsA \gsL \gsR} &= \Pi^{\calR^2}_{\gsL \gsR} \cdot P^{\calD(W)}_{\leq t, \gsA \gsL \gsR} \cdot \Pi^{\calR^2}_{\gsL \gsR}.
    \end{align}
    Plugging this into the left-hand side of~\cref{eq:twirling-main-eq-1}, we get
    \begin{align}
        &= \norm{\E_{C, D \leftarrow \frakD}  (C_{\gsA} \otimes Q[C,D]_{\gsL \gsR})^\dagger \cdot \Big(\Pi^{\calR^2}_{\gsL \gsR} \cdot \Pi^{\xydist}_{\leq t, \gsL \gsR} \cdot \Pi^{\calR^2}_{\gsL \gsR}  - \Pi^{\calR^2}_{\gsL \gsR} \cdot P_{\leq t, \gsA \gsL \gsR}^{\calD(W)} \cdot \Pi^{\calR^2}_{\gsL \gsR} \Big) \cdot (C_{\gsA} \otimes Q[C,D]_{\gsL \gsR})}_{\opnorm}\\
        &= \norm{ \Pi^{\calR^2}_{\gsL \gsR} \cdot \Bigg( \E_{C, D \leftarrow \frakD}  (C_{\gsA} \otimes Q[C,D]_{\gsL \gsR})^\dagger \cdot \Big( \Pi^{\xydist}_{\leq t, \gsL \gsR}  -  P_{\leq t, \gsA \gsL \gsR}^{\calD(W)} \Big) \cdot (C_{\gsA} \otimes Q[C,D]_{\gsL \gsR}) \Bigg) \cdot \Pi^{\calR^2}_{\gsL \gsR}}_{\opnorm}\\
        &\leq \norm{ \E_{C, D \leftarrow \frakD}  (C_{\gsA} \otimes Q[C,D]_{\gsL \gsR})^\dagger \cdot \Big( \Pi^{\xydist}_{\leq t, \gsL \gsR}  -  P_{\leq t, \gsA \gsL \gsR}^{\calD(W)} \Big) \cdot (C_{\gsA} \otimes Q[C,D]_{\gsL \gsR})}_{\opnorm} \label{eq:twirling-main-eq-1-upper-bound},
    \end{align}
    where the second equality follows from the fact that $\Pi^{\calR^2}_{\gsL \gsR}$ commutes with $Q[C,D]_{\gsL \gsR}$, and the inequality uses the fact that $\norm{\Pi \cdot M \cdot \Pi}_{\opnorm} \leq \norm{M}_{\opnorm}$ for any projector $\Pi$.
    
    Next, we use the fact that the operators $\Pi^{\xydist}_{\leq t}$ and $P_{\leq t,\gsA \gsL \gsR}^{\calD(W)}$ are block diagonal with respect to $\{\Pi_{\ell,r,\gsL \gsR}\}_{\ell,r \geq 0}$ (recall that $\Pi_{\ell,r}$ denotes the projector that restricts the registers $\sL$ and $\sR$ to have lengths $\ell$ and $r$, respectively), i.e., they map the image of $\Pi_{\ell,r}$ to the image of $\Pi_{\ell,r}$. Thus,
    \begin{align}
        (\ref{eq:twirling-main-eq-1-upper-bound}) &= \max_{\substack{\ell,r \geq 0:\\ \ell + r \leq t}} \norm{ \E_{C, D \leftarrow \frakD}  (C_{\gsA} \otimes Q[C,D]_{\gsL \gsR})^\dagger \cdot \Big( \Pi^{\xydist}_{\ell,r, \gsL \gsR}  -  P_{\ell,r, \gsA \gsL \gsR}^{\calD(W)} \Big) \cdot (C_{\gsA} \otimes Q[C,D]_{\gsL \gsR})}_{\opnorm}\\
        &\leq \max_{\substack{\ell,r \geq 0:\\ \ell + r \leq t}} \Bigg\lVert \frac{N}{N-\ell-r+1}\cdot \E_{C, D \leftarrow \frakD}  (C_{\gsA} \otimes Q[C,D]_{\gsL \gsR})^\dagger \cdot  \Bigg(\sum_{i \in [\ell]} \Pi^{\mathsf{eq}}_{\darkgray{\sA,\sL^{(\ell)}_{\sX,i}}} \nonumber \\
        & \quad + \sum_{i \in [r]} \left(\Pi^{\mathsf{eq}}_{\darkgray{\sA,\sR^{(r)}_{\sX,i}}} - \Pi^{\EPR}_{\darkgray{\sA,\sR^{(r)}_{\sX,i}}} \right) + 2r \sqrt{\frac{\ell + r}{N}} \Id_{\gsA \gsL \gsR}\Bigg) \cdot (C_{\gsA} \otimes Q[C,D]_{\gsL \gsR}) \Bigg\rVert_{\opnorm} \tag{by~\cref{claim:operator-upper-bound-twirling}}\\
        &\leq \max_{\substack{\ell,r \geq 0:\\ \ell + r \leq t}} \frac{N}{N-\ell-r+1} \cdot \Bigg(\frac{2\ell + r}{N+1} + 2r \sqrt{\frac{\ell + r}{N}}\Bigg). \tag{by~\cref{claim:twirling-ell-r} and the triangle inequality}
    \end{align}
    This expression is maximized by setting $r = t$ and $\ell = 0$, which yields a final upper bound of
     \begin{align}
        \frac{N}{N-t+1} \cdot \Bigg( \frac{t}{N+1} + 2t \sqrt{\frac{t}{N}}\Bigg) \leq \frac{N}{N-t} \cdot \Bigg( 3t \sqrt{\frac{t}{N}} \Bigg) \leq 6t \sqrt{\frac{t}{N}}
    \end{align}
    where the last inequality uses the assumption that $t \leq N/2$. 
\end{proof}

\newpage
\bibliography{ref}

\newcommand{\etalchar}[1]{$^{#1}$}
\begin{thebibliography}{CGAH{\etalchar{+}}17}

\bibitem[AAB{\etalchar{+}}19]{arute2019quantum}
Frank Arute, Kunal Arya, Ryan Babbush, Dave Bacon, Joseph~C Bardin, Rami
  Barends, Rupak Biswas, Sergio Boixo, Fernando~GSL Brandao, David~A Buell,
  et~al.
\newblock Quantum supremacy using a programmable superconducting processor.
\newblock {\em Nature}, 574(7779):505--510, 2019.

\bibitem[AGKL24]{ananth2024pseudorandom}
Prabhanjan Ananth, Aditya Gulati, Fatih Kaleoglu, and Yao-Ting Lin.
\newblock Pseudorandom isometries.
\newblock In {\em Annual International Conference on the Theory and
  Applications of Cryptographic Techniques}, pages 226--254. Springer, 2024.

\bibitem[AMR20]{alagic2020efficient}
Gorjan Alagic, Christian Majenz, and Alexander Russell.
\newblock Efficient simulation of random states and random unitaries.
\newblock In {\em Advances in Cryptology--EUROCRYPT 2020: 39th Annual
  International Conference on the Theory and Applications of Cryptographic
  Techniques, Zagreb, Croatia, May 10--14, 2020, Proceedings, Part III 39},
  pages 759--787. Springer, 2020.

\bibitem[AQY22]{ananth2022cryptography}
Prabhanjan Ananth, Luowen Qian, and Henry Yuen.
\newblock Cryptography from pseudorandom quantum states.
\newblock In {\em Annual International Cryptology Conference}, pages 208--236.
  Springer, 2022.

\bibitem[BFNV19]{bouland2019complexity}
Adam Bouland, Bill Fefferman, Chinmay Nirkhe, and Umesh Vazirani.
\newblock On the complexity and verification of quantum random circuit
  sampling.
\newblock {\em Nature Physics}, 15(2):159--163, 2019.

\bibitem[BM24]{brakerski2024real}
Zvika Brakerski and Nir Magrafta.
\newblock Real-valued somewhat-pseudorandom unitaries.
\newblock {\em arXiv preprint arXiv:2403.16704}, 2024.

\bibitem[CBB{\etalchar{+}}24]{chen2024efficient}
Chi-Fang Chen, Adam Bouland, Fernando~GSL Brand{\~a}o, Jordan Docter, Patrick
  Hayden, and Michelle Xu.
\newblock Efficient unitary designs and spseudorandom unitaries from
  permutations.
\newblock {\em arXiv preprint arXiv:2404.16751}, 2024.

\bibitem[CGAH{\etalchar{+}}17]{cotler2017black}
Jordan~S Cotler, Guy Gur-Ari, Masanori Hanada, Joseph Polchinski, Phil Saad,
  Stephen~H Shenker, Douglas Stanford, Alexandre Streicher, and Masaki Tezuka.
\newblock Black holes and random matrices.
\newblock {\em Journal of High Energy Physics}, 2017(5):1--54, 2017.

\bibitem[CHJLY17]{cotler2017chaos}
Jordan Cotler, Nicholas Hunter-Jones, Junyu Liu, and Beni Yoshida.
\newblock Chaos, complexity, and random matrices.
\newblock {\em Journal of High Energy Physics}, 2017(11):1--60, 2017.

\bibitem[CSM{\etalchar{+}}23]{choi2023preparing}
Joonhee Choi, Adam~L Shaw, Ivaylo~S Madjarov, Xin Xie, Ran Finkelstein, Jacob~P
  Covey, Jordan~S Cotler, Daniel~K Mark, Hsin-Yuan Huang, Anant Kale, et~al.
\newblock Preparing random states and benchmarking with many-body quantum
  chaos.
\newblock {\em Nature}, 613(7944):468--473, 2023.

\bibitem[EFH{\etalchar{+}}22]{elben2022randomized}
Andreas Elben, Steven~T Flammia, Hsin-Yuan Huang, Richard Kueng, John Preskill,
  Beno{\^\i}t Vermersch, and Peter Zoller.
\newblock The randomized measurement toolbox.
\newblock {\em arXiv preprint arXiv:2203.11374}, 2022.

\bibitem[EFL{\etalchar{+}}24]{engelhardt2024cryptographic}
Netta Engelhardt, {\AA}smund Folkestad, Adam Levine, Evita Verheijden, and Lisa
  Yang.
\newblock Cryptographic censorship.
\newblock {\em arXiv preprint arXiv:2402.03425}, 2024.

\bibitem[GGM86]{goldreich1986construct}
Oded Goldreich, Shafi Goldwasser, and Silvio Micali.
\newblock How to construct random functions.
\newblock {\em Journal of the ACM (JACM)}, 33(4):792--807, 1986.

\bibitem[HBC{\etalchar{+}}22]{huang2022quantum}
Hsin-Yuan Huang, Michael Broughton, Jordan Cotler, Sitan Chen, Jerry Li, Masoud
  Mohseni, Hartmut Neven, Ryan Babbush, Richard Kueng, John Preskill, et~al.
\newblock Quantum advantage in learning from experiments.
\newblock {\em Science}, 376(6598):1182--1186, 2022.

\bibitem[HBK23]{haug2023pseudorandom}
Tobias Haug, Kishor Bharti, and Dax~Enshan Koh.
\newblock Pseudorandom unitaries are neither real nor sparse nor noise-robust.
\newblock {\em arXiv preprint arXiv:2306.11677}, 2023.

\bibitem[HKP20]{huang2020predicting}
Hsin-Yuan Huang, Richard Kueng, and John Preskill.
\newblock Predicting many properties of a quantum system from very few
  measurements.
\newblock {\em Nat. Phys.}, 16, 2020.

\bibitem[HKT{\etalchar{+}}22]{huang2022provably}
Hsin-Yuan Huang, Richard Kueng, Giacomo Torlai, Victor~V Albert, and John
  Preskill.
\newblock Provably efficient machine learning for quantum many-body problems.
\newblock {\em Science}, 377(6613):eabk3333, 2022.

\bibitem[HLSW04]{hayden2004randomizing}
Patrick Hayden, Debbie Leung, Peter~W Shor, and Andreas Winter.
\newblock Randomizing quantum states: Constructions and applications.
\newblock {\em Communications in Mathematical Physics}, 250:371--391, 2004.

\bibitem[JLS17]{ji2017pseudorandom}
Zhengfeng Ji, Yi-Kai Liu, and Fang Song.
\newblock Pseudorandom states, non-cloning theorems and quantum money.
\newblock {\em arXiv preprint arXiv:1711.00385}, 2017.

\bibitem[JLS18]{ji2018pseudorandom}
Zhengfeng Ji, Yi-Kai Liu, and Fang Song.
\newblock Pseudorandom quantum states.
\newblock In {\em Advances in Cryptology--CRYPTO 2018: 38th Annual
  International Cryptology Conference, Santa Barbara, CA, USA, August 19--23,
  2018, Proceedings, Part III 38}, pages 126--152. Springer, 2018.

\bibitem[KLR{\etalchar{+}}08]{knill2008randomized}
Emanuel Knill, Dietrich Leibfried, Rolf Reichle, Joe Britton, R~Brad Blakestad,
  John~D Jost, Chris Langer, Roee Ozeri, Signe Seidelin, and David~J Wineland.
\newblock Randomized benchmarking of quantum gates.
\newblock {\em Physical Review A—Atomic, Molecular, and Optical Physics},
  77(1):012307, 2008.

\bibitem[KQST23]{kretschmer2023quantum}
William Kretschmer, Luowen Qian, Makrand Sinha, and Avishay Tal.
\newblock Quantum cryptography in algorithmica.
\newblock In {\em Proceedings of the 55th Annual ACM Symposium on Theory of
  Computing}, pages 1589--1602, 2023.

\bibitem[KTP20]{kim2020ghost}
Isaac Kim, Eugene Tang, and John Preskill.
\newblock The ghost in the radiation: Robust encodings of the black hole
  interior.
\newblock {\em Journal of High Energy Physics}, 2020(6):1--65, 2020.

\bibitem[LMW24]{lombardi2024one}
Alex Lombardi, Fermi Ma, and John Wright.
\newblock A one-query lower bound for unitary synthesis and breaking quantum
  cryptography.
\newblock In {\em Proceedings of the 56th Annual ACM Symposium on Theory of
  Computing}, pages 979--990, 2024.

\bibitem[LQS{\etalchar{+}}23]{lu2023quantum}
Chuhan Lu, Minglong Qin, Fang Song, Penghui Yao, and Mingnan Zhao.
\newblock Quantum pseudorandom scramblers.
\newblock {\em arXiv preprint arXiv:2309.08941}, 2023.

\bibitem[Mov23]{movassagh2023hardness}
Ramis Movassagh.
\newblock The hardness of random quantum circuits.
\newblock {\em Nature Physics}, 19(11):1719--1724, 2023.

\bibitem[MPSY24]{metger2024simple}
Tony Metger, Alexander Poremba, Makrand Sinha, and Henry Yuen.
\newblock Simple constructions of linear-depth t-designs and pseudorandom
  unitaries.
\newblock {\em arXiv preprint arXiv:2404.12647}, 2024.

\bibitem[NVH18]{nahum2018operator}
Adam Nahum, Sagar Vijay, and Jeongwan Haah.
\newblock Operator spreading in random unitary circuits.
\newblock {\em Physical Review X}, 8(2):021014, 2018.

\bibitem[Reg09]{regev2009lattices}
Oded Regev.
\newblock On lattices, learning with errors, random linear codes, and
  cryptography.
\newblock {\em Journal of the ACM (JACM)}, 56(6):1--40, 2009.

\bibitem[SHH24]{schuster2024random}
Thomas Schuster, Jonas Haferkamp, and Hsin-Yuan Huang.
\newblock Random unitaries in extremely low depth.
\newblock {\em arXiv preprint arXiv:2407.07754}, 2024.

\bibitem[YE23]{yang2023complexity}
Lisa Yang and Netta Engelhardt.
\newblock The complexity of learning (pseudo) random dynamics of black holes
  and other chaotic systems.
\newblock {\em arXiv preprint arXiv:2302.11013}, 2023.

\bibitem[Zha16]{zhandry2016note}
Mark Zhandry.
\newblock A note on quantum-secure prps.
\newblock {\em arXiv preprint arXiv:1611.05564}, 2016.

\bibitem[Zha19]{zhandry2019record}
Mark Zhandry.
\newblock How to record quantum queries, and applications to quantum
  indifferentiability.
\newblock In {\em Advances in Cryptology--CRYPTO 2019: 39th Annual
  International Cryptology Conference, Santa Barbara, CA, USA, August 18--22,
  2019, Proceedings, Part II 39}, pages 239--268. Springer, 2019.

\bibitem[Zha21]{zhandry2021construct}
Mark Zhandry.
\newblock How to construct quantum random functions.
\newblock {\em Journal of the ACM (JACM)}, 68(5):1--43, 2021.

\end{thebibliography}
\bibliographystyle{alpha}

\newpage
\appendix
\part{Appendices}
\label{part:app}

\section{Efficient circuit implementation of path-recording oracle}
\label{sec:efficient-implementation-pro}

We briefly describe how to efficiently implement the path-recording oracles on a quantum computer to simulate forward and inverse queries of a Haar-random unitary up to inverse-exponential error.

\subsection{Implementing relation states}

There are multiple ways to implement the relation state $\ket*{R}$ for a relation $R$.
We describe one simple approach.
We represent $\ket*{R}$ using $|R|$ $2n$-qubit registers by sorting the tuples $(x, y) \in [N]^2$ in the relation $R$.
For example, consider $R$, we can store $\ket*{R}$ on a quantum computer as
\begin{equation}
    \ket*{x_1} \ket*{y_1} \ldots \ket*{x_{|R|}} \ket*{y_{|R|}},
\end{equation}
where $R = \{ (x_i, y_i) \}_{i=1}^{|R|}$ and $(x_1, y_1) \leq \ldots \leq (x_{|R|}, y_{|R|})$. Here, $(x, y) \leq (x', y')$ denotes the lexicographical ordering, which means either (a) $x < x'$ or (b) $x = x'$ and $y \leq y'$.

\subsection{Implementing forward queries}
\label{subsec:imp-forward-q}

In~\cref{sec:compressed-haar}, we defined a (standard) path-recording oracle $V$ that simulates forward (but not inverse) queries to a Haar-random unitary. In this subsection, we describe how to implement this linear map efficiently.

\begin{definition}
\label{def:choisometry-replicate}[\cref{def:choisometry}, repeated]
The path-recording oracle $\pr$ is a linear map $\pr: \calH_{\sA} \otimes \calH_{\sR} \rightarrow \calH_{\sA} \otimes \calH_{\sR}$ defined as follows. For all $x \in [N]$ and injective relations $R \in \calR^{\inj}$ such that $\abs{R} < N$,
\begin{equation}
    \pr: \ket*{x}_{\gsA} \ket*{R}_{\gsR} \mapsto \frac{1}{\sqrt{N - \abs{R}}} \sum_{\substack{y \in [N], \\ y \not\in \Im(R)}} \ket*{y}_{\gsA} \ket*{R \cup\{(x,y)\}}_{\gsR}.
\end{equation}
\end{definition}

We briefly sketch how to implement $V$ on an input $\ket*{x} \ket*{R}$. 
\begin{enumerate}
    \item[] \textbf{Input:} a state of the form $\ket*{x} \ket*{R}$.
    \item The first step is to perform the map
    \begin{align}
        \ket*{x} \ket*{R} \mapsto \frac{1}{\sqrt{N -\abs{R}}} \sum_{\substack{y \in [N],\\ y\not\in \Im(R)}} \ket*{y} \ket*{x} \ket*{R}.
    \end{align}
    This can be done as follows:
    \begin{enumerate}
        \item First, prepare a new register containing a uniform superposition over $\{1,2,\dots,N - \abs{R}\}$. 
        \item Next, observe that given $y_1 \leq \dots \leq y_{\abs{R}}$ (which are stored in subregisters of $\ket*{R}$), there is an efficiently computable bijection $f$ from $\{1,2,\dots,N - \abs{R}\}$ to the set $\{y: y\in [N], y\not\in \Im(R)\}$: on input $x$, compute the number $n_x$ of elements $y_i$ in the list $(y_1,\dots,y_{\abs{R}})$ such that $x \geq y_i$, and output $x + n_x$. Using a similar algorithm, we can also compute $f^{-1}$ efficiently. 
        
        This allows us to efficiently compute $f$ \emph{in place}. Applying this to the uniform superposition over $\{1,2,\dots,N - \abs{R}\}$ produces the desired superposition. 
    \end{enumerate}
    \item Next, compute the function that maps $(y,x,R)$ to $R \cup \{(x,y)\}$ (where $R$ and $R \cup \{(x,y)\}$ are represented as sorted lists of ordered pairs). This step corresponds to the map
    \begin{align}
        \frac{1}{\sqrt{N -\abs{R}}} \sum_{\substack{y \in [N],\\ y\not\in \Im(R)}} \ket*{y} \ket*{x} \ket*{R} \mapsto \frac{1}{\sqrt{N -\abs{R}}} \sum_{\substack{y \in [N],\\ y\not\in \Im(R)}} \ket*{y} \ket*{x} \ket*{R} \ket*{R \cup \{(x,y)\}}.
    \end{align}
    \item Finally, use the $\ket*{y} \ket*{R \cup \{(x,y)\})}$ register to uncompute the $\ket*{x} \ket*{R}$ registers; note that $x,R$ can be uniquely computed from $y, R \cup \{(x,y)\}$, since $y$ is guaranteed to be outside $\Im(R)$. This corresponds to the map 
    \begin{align}
        \frac{1}{\sqrt{N -\abs{R}}} \sum_{\substack{y \in [N],\\ y\not\in \Im(R)}} \ket*{y} \ket*{x} \ket*{R} \ket*{R \cup \{(x,y)\}} \mapsto \frac{1}{\sqrt{N -\abs{R}}} \sum_{\substack{y \in [N],\\ y\not\in \Im(R)}} \ket*{y} \ket*{R \cup \{(x,y)\}},
    \end{align}
    which corresponds to the output of $V$. 
\end{enumerate}

\subsection{Implementing forward and inverse queries}
\label{subsec:imp-forward-inverse-q}

In~\cref{sec:symmetric-V}, we defined a (symmetric) path-recording oracle $V$ that simulates both forward and inverse queries to a Haar-random unitary. In this subsection, we describe how to implement this linear map efficiently. Recall that this linear map $V$ is defined in terms of two helper linear maps $V^L$ and $V^R$.

\begin{definition}[left and right partial isometries] \label{def:V-sym-PRO-replicate}[\cref{def:V-sym-PRO}, repeated]
    Let $V^L$ be the linear operator that acts as follows. For $x \in [N]$ and $(L,R) \in \mathcal{R}^{2,\leq N-1}$,
    \begin{align}
        V^L \cdot \ket*{x}_{\gsA} \ket*{L}_{\gsL} \ket*{R}_{\gsR} \coloneqq \sum_{\substack{y \in [N]:\\ y\not\in \Im(L \cup R)}} \frac{1}{\sqrt{N - \abs{\Im(L \cup R)}}} \ket*{y}_{\gsA} \ket*{L \cup \{(x,y)\}}_{\gsL} \ket*{R}_{\gsR}.
    \end{align}
    Define $V^R$ to be the linear operator such that for all $y \in [N]$ and $(L,R) \in \mathcal{R}^{2,\leq N-1}$,
    \begin{align}
        V^R \cdot \ket*{y}_{\gsA} \ket*{L}_{\gsL} \ket*{R}_{\gsR} \coloneqq \sum_{\substack{x \in [N]:\\ x\not\in \Dom(L \cup R)}} \frac{1}{\sqrt{N - \abs{\Dom(L \cup R)}}} \ket*{x}_{\gsA} \ket*{L}_{\gsL} \ket*{R \cup \{(x, y)\} }_{\gsR}.
    \end{align}
\end{definition}

Efficient implementation of forward queries to $V^L$ and $V^R$ can be done similarly to \cref{subsec:imp-forward-q}. Now, let $U^L$ denote the efficient unitary implicit in the procedure described in~\cref{subsec:imp-forward-q}, which satisfies the guarantee that $U^L \ket*{x} \ket*{R} \ket*{0^m} = (V^L \ket*{x} \ket*{R}) \otimes \ket*{0^{m'}}$, where $m$ and $m'$ denotes the number of ancilla qubits in the input and output respectively. Define $U^R$ similarly. Technically, the number of ancillas depends on the size of $L$ and $R$. However, we can always assume that the sizes of $L$ and $R$ are upper bounded by the number of queries so far, and so we can also use this to bound the size of the ancillas needed to implement the $t$-th query.

\paragraph{Implementing $V^{L, \dagger}, V^{R, \dagger}$}
Since we can efficiently implement $U^L$, we can also implement its inverse $U^{L,\dagger}$. Note that on states of the form $\ket*{\psi}\ket*{0^{m'}}$, where $\ket*{\psi}$ is in the image of $V^L$, applying $U^{L,\dagger}$ produces the output state $(V^{L,\dagger} \ket*{\psi}) \ket*{0^m}$. Thus, we can use $U^{L,\dagger}$ to implement $V^{L,\dagger}$.

\paragraph{Implementing coherent measurements on $V^L \cdot V^{L, \dagger}$ and $V^R \cdot V^{R, \dagger}$}
Before we describe how we implement $V$, we will need to describe how to perform measurements corresponding to the projectors $V^L \cdot V^{L, \dagger}$ and $V^R \cdot V^{R, \dagger}$. In fact, we will need to implement these measurements coherently, i.e., apply the unitary that computes the binary measurement outcome onto an external qubit. This can be done as follows (for simplicity, we only describe the procedure for implementing the coherent $V^L \cdot V^{L, \dagger}$ measurement, as the coherent $V^R \cdot V^{R, \dagger}$ measurement is symmetric).
\begin{enumerate}
    \item[] \textbf{Input:} a state $\ket*{\psi} = \sum_{x, L, R} \alpha_{x, L, R} \ket*{x} \ket*{L} \ket*{R}$. 
    \item Add $m'$ ancillary qubits $\ket*{0^{m'}}$ and then apply $U^{L,\dagger}$. 
    \item Initialize a new one-qubit register to $\ket*{0}_{\gsB}$. Controlled on the last $m$ qubits of the rest of the state (i.e., all registers except for $\sB$) being $0^m$, apply a Pauli $X$ to the $\sB$ register. By definition of $U^{L,\dagger}$, this Pauli $X$ is applied if and only if the original input state was in the image of $V^L \cdot V^{L,\dagger}$.
    \item Apply $U^{L}$ to the non-$\sB$ registers. 
\end{enumerate}

Now, recall the definition of $V$.

\begin{definition}
\label{def:symmetric-V-replicate}
    The symmetric path-recording oracle is the operator $V$ defined as
    \begin{align}
        V &= V^L \cdot (\Id - V^R \cdot V^{R,\dagger}) + (\Id - V^L \cdot V^{L,\dagger}) \cdot V^{R,\dagger}.
    \end{align}
\end{definition}

We sketch an implementation of a forward query to the symmetric path-recording oracle $V$. The inverse query is symmetric by swapping $L$ and $R$.
\begin{enumerate}
    \item[] \textbf{Input:} a state $\ket*{\psi} = \sum_{x, L, R} \alpha_{x, L, R} \ket*{x} \ket*{L} \ket*{R}$. 
    \item Add two ancilla qubits initialized at $\ket*{0}$ to obtain $\ket*{0} \ket*{0} \ket*{\psi}$.
    \item Implement coherent measurement $V^R \cdot V^{R, \dagger}$, writing the outcome onto the first ancilla qubit.
    \item Apply the following controlled operation:
    \begin{itemize}
        \item Controlled on the first ancilla qubit being $1$, apply $V^L$.
        \item Controlled on the first ancilla qubit being $0$, apply $V^{R,\dagger}$. Then, apply the coherent measurement $V^L \cdot V^{L, \dagger}$, writing the outcome onto the second ancilla qubit.
    \end{itemize}
    \item Measure the second ancilla qubit, and abort if the outcome is $1$.
    \item Apply the coherent measurement of $V^L \cdot V^{L, \dagger}$ with the outcome applied onto the first ancilla qubit. 
    \item Trace out the remaining ancilla qubit (the first one), which is guaranteed to be $\ket{0}$.
\end{enumerate}

\section{The path-recording framework}
\label{sec:cuframework}

In this section, we develop a mathematical framework by generalizing the path-recording oracle introduced in \cref{sec:compressed-haar}.
This new framework enables one to develop modified versions of path-recording oracle in which the set of relations that the oracle uses is restricted to a subset $\mathcal{S}^{\inj} \subseteq \calR^{\inj}$ of the set of all injective relations.
This mathematical framework offers flexibility for establishing indistinguishability from Haar-random unitary via the path-recording oracle.

\vspace{1em}
\noindent To develop the path-recording framework, we define the following notations.
\begin{itemize}
    \item $t_{\mathrm{max}}$ is an integer between $1$ and $N$ sets the maximum size of the relations. This integer also sets the limit on how many queries we can make to the path-recording oracle. In the canonical path-recording oracle introduced in \cref{subsec:define-pr-oracle}, $t_{\mathrm{max}}$ is equal to $N$.
    \item $\mathcal{S}^{\inj}_t$ is a subset of all the injective relations $\calR^{\inj}_t$ of size $t$ for any $0 \leq t \leq t_{\max}$. In particular, we require the subset for the maximum $t$ to be non-empty: $|\mathcal{S}^{\inj}_{t_{\mathrm{max}}}| \geq 1$.
    \item $\mathcal{S}^{\inj} \coloneqq \cup_{t=0}^{t_{\mathrm{max}}} \mathcal{S}^{\inj}_t$. The set $\mathcal{S}$ restricts the relations that the path-recording oracle could use.
\end{itemize}
We define the following two constraints on the restricted set $\mathcal{S}^{\inj}$.

\begin{definition}[Consistency]
\label{def:consistency-set}
We say the set $\mathcal{S}^{\inj}$ of relations is consistent if
\begin{align}
    &\forall (x_1, \ldots, x_t) \in [N]^t, \quad \exists (y_1, \ldots, y_t) \in [N]^t,\\
    &  \,\,\text{such that} \quad \{(x_i, y_i)\}_{i=1}^t \in \mathcal{S}^{\inj}.
\end{align}
Furthermore, if $\{(x_i, y_i)\}_{i=1}^t \in \mathcal{S}^{\inj}$, then for any $0 \leq \tau \leq t$, $\{(x_i, y_i)\}_{i=1}^{\tau} \in \mathcal{S}^{\inj}$.
\end{definition}

The \emph{consistency constraint} ensures that all possible $(x_1, \ldots, x_t) \in [N]^t$ are valid. This is central for path-recording oracle because the adversary algorithm can choose the inputs $x_1, \ldots, x_t$ arbitrarily.
The constraint also ensures that all relations in $\mathcal{S}^{\inj}$ are ``meaningful'' because they can all be obtained by adding in each tuple $(x_i, y_i)$ one by one while maintaining in the restricted subset $\mathcal{S}^{\inj}$.

\begin{definition}[Uniform growth]
\label{def:uniform-growth-constraint}
We say the set $\mathcal{S}^{\inj}$ of relations satisfies the uniform growth constraint if for all $0 \leq t < t_{\max}$, there exists $\calZ_{t} \geq 1$, such that for all $x \in [N]$ and $R \in \mathcal{S}^{\inj}_t$,
\begin{equation}
    \calZ_{t} = \sum\limits_{\substack{y \in [N], \, \mathrm{s.t.} \\ R \cup \{(x, y)\} \in \mathcal{S}^{\inj}_{t+1}}} 1.
\end{equation}
\end{definition}

The \emph{uniform growth constraint} ensures the number of $y$ that can be used to grow the relation $R$ by by size one is uniform across all $x \in [N]$ and all relations $R$ of the same size.
We illustrate these two constraints with the following examples.
\begin{itemize}
    \item $\mathcal{S}^{\inj}$ contains all relations where the first $k$ bits in $y_1, \ldots, y_t \in [N]$ are distinct. In this case, $t_{\max} = 2^{k}$. Furthermore, $\mathcal{S}^{\inj}$ is consistent and satisfies the uniform growth constraint.
    \item $\mathcal{S}^{\inj}$ contains all relations $R$ where $R = \{(x_i, x_i)\}_{i=1}^{|R|}$. In this case, $t_{\max} = N$. And $\mathcal{S}^{\inj}$ is consistent and satisfies the uniform growth constraint.
    \item $\mathcal{S}^{\inj} = \{ R \in \calR^{\inj} \,\, | \,\, |R| = N \}$. In this case, $t_{\max} = N$. However, $\mathcal{S}^{\inj}$ is not consistent and does not satisfy the uniform growth constraint because it violates $\calZ_t \geq 1$.
\end{itemize}
For any consistent set $\mathcal{S}^{\inj}$ of relations, we have $\varnothing \in \mathcal{S}^{\inj}$ because we can take $\tau = 0$ in \cref{def:consistency-set} for any relation $R \in \mathcal{S}^{\inj}$ to obtain that $\varnothing \in \mathcal{S}^{\inj}$.

\subsection{Defining $\pr(\mathcal{S}^{\inj})$ and the $\pr(\mathcal{S}^{\inj})$ state}

We now define the behavior of the $\mathcal{S}^{\inj}$-restricted path-recording oracle.

\begin{definition}[$\mathcal{S}^{\inj}$-restricted path-recording oracle]
\label{def:mchoisometry}
Given any consistent set $\mathcal{S}^{\inj}$ of relations.
The $\mathcal{S}^{\inj}$-restricted path-recording oracle $\pr(\mathcal{S}^{\inj})$ is a linear map
\begin{equation}
    \pr(\mathcal{S}^{\inj}): \calH_{\sA} \otimes \calH_{\sR} \rightarrow \calH_{\sA} \otimes \calH_{\sR}
\end{equation}
defined as follows. For all $0 \leq t < t_{\max}$, $R \in \mathcal{S}^{\inj}_t$, and $x \in [N]$, 
\begin{equation}
    \pr(\mathcal{S}^{\inj}): \ket*{x}_{\gsA} \ket*{R}_{\gsR} \mapsto \frac{1}{\sqrt{\calZ_{x,R}}} \sum_{\substack{y \in [N], \\ R \cup \{(x, y)\} \in \mathcal{S}^{\inj}_{t+1}}} \ket*{y}_{\gsA} \ket*{R \cup \{(x,y)\}}_{\gsR},
\end{equation}
The normalization factor $\calZ_{x, R}$ is given by
\begin{equation}
    \calZ_{x, R} \coloneqq \sum\limits_{\substack{y \in [N], \\ R \cup \{(x, y)\} \in \mathcal{S}^{\inj}_{t+1}}} 1 \geq 1,
\end{equation}
where the last inequality follows from the consistency constraint that for any $(x_1, \ldots, x_t) \in [N]^t$, there exists $(y_1, \ldots, y_t) \in [N]^t$, such that $\{(x_i, y_i)\}_{i=1}^t \in \mathcal{S}^{\inj}$.
\end{definition}

Next, we define the $G$-rotated $\pr(\mathcal{S}^{\inj})$ state, which represents the global state after an adversary has queried the $\mathcal{S}^{\inj}$-restricted path-recording oracle multiple times.

\begin{definition}[$G$-rotated $\pr(\mathcal{S}^{\inj})$ state]\label{def:mchostate}
Given a consistent set $\mathcal{S}^{\inj}$, an $n$-qubit unitary $G$ and a $t$-query adversary $\Adv$ with forward queries specified by a $t$-tuple of unitaries $(A_{1, \gsA \gsB},\dots,A_{t, \gsA \gsB})$, the $G$-rotated $\pr(\mathcal{S}^{\inj})$ state is
\begin{equation}
    \ket*{\Adv^{\pr(\mathcal{S}^{\inj}) \cdot G}_t}_{\gsA \gsB \gsR} \coloneqq \prod_{i = 1}^t \Big( \pr(\mathcal{S}^{\inj}) \cdot G_{\gsA} \cdot A_{i, \gsA \gsB} \Big) \ket*{0}_{\gsA \gsB} \ket*{\varnothing}_{\gsR}.
\end{equation}
\end{definition}

The $G$-rotated $\pr(\mathcal{S}^{\inj})$ state $\ket*{\Adv^{\pr(\mathcal{S}^{\inj}) \cdot G}_t}_{\gsA \gsB \gsR}$ is the state of the entire system after the adversary has made $t$ queries to $\pr(\mathcal{S}^{\inj}) \cdot G$, and includes the adversary's query register ($\sA$), auxiliary register ($\sB$), and the purifying registers ($\sR$), after $t$ queries to the oracle.

\subsection{$\pr(\mathcal{S}^{\inj})$ is a partial isometry}

A crucial property of the $G$-rotated $\pr(\mathcal{S}^{\inj})$ state is that it maintains unit norm up to $t_{\max}$ queries. We formalize this in the following lemma:

\begin{lemma}[$G$-rotated $\pr(\mathcal{S}^{\inj})$ state has unit norm]
\label{lemma:purified-mcho-state-unit-norm}
    For any consistent set $\mathcal{S}^{\inj}$ of relations, any adversary $\Adv$ making $t \leq t_{\max}$ forward queries to an $n$-qubit oracle, and any $n$-qubit unitary $G$, the $G$-rotated $\pr(\mathcal{S}^{\inj})$ state $\ket*{\Adv^{\pr(\mathcal{S}^{\inj}) \cdot G}_t}_{\gsA \gsB \gsR}$ has unit norm.
\end{lemma}

To prove this lemma, we first need to establish that the $\mathcal{S}^{\inj}$-restricted path-recording oracle $\pr(\mathcal{S}^{\inj})$ acts as a partial isometry on certain states. This is formalized in the following lemma:

\begin{lemma}[Partial Isometry]
\label{lem:mcho-isometry}
For any consistent set $\mathcal{S}^{\inj}$ of relations, the $\mathcal{S}^{\inj}$-restricted path-recording oracle $\pr(\mathcal{S}^{\inj})$ is an isometry on the subspace of $\calH_{\sA} \otimes \calH_{\sR}$ spanned by the states $\ket*{x}\ket*{R}$ for $x \in [N]$ and $R \in \mathcal{S}^{\inj}$ such that $\abs{R} < t_{\max}$.
\end{lemma}

\begin{proof}[Proof of Lemma~\ref{lem:mcho-isometry}]
To prove that $\pr(\mathcal{S}^{\inj})$ is an isometry on the specified subspace, it suffices to show that for all $x,x' \in [N]$ and $R,R' \in \mathcal{S}^{\inj}$ with $\abs{R}, \abs{R'} < t_{\max}$,
\begin{equation}
    \bra*{x'}_{\gsA} \bra*{R'}_{\gsR} \pr(\mathcal{S}^{\inj})^\dagger \cdot \pr(\mathcal{S}^{\inj}) \ket*{x}_{\gsA} \ket*{R}_{\gsR} = \braket*{x'}{x}_{\gsA} \cdot \braket*{R'}{R}_{\gsR} \label{eq:mcho-partial-isometry}.
\end{equation}
The proof proceeds in the same way as the proof of \cref{lem:cho-isometry} after one notes the fact that the normalization factor $\calZ{x, R} \geq 1$ from the consistency of the set $\mathcal{S}^{\inj}$.
\end{proof}

We can now prove~\cref{lemma:purified-mcho-state-unit-norm}.

\begin{proof}[Proof of~\cref{lemma:purified-mcho-state-unit-norm}]
Note that $\mathcal{S}^{\inj}$ is consistent implies $\varnothing \in \mathcal{S}^{\inj}_0$. Hence, we can use \cref{lem:mcho-isometry} to establish this lemma via the same mathematical induction as the proof of \cref{lemma:purified-cho-state-unit-norm}.
\end{proof}

\subsection{$\pr(\mathcal{S}^{\inj})$ is right unitary invariant}

So far, we have not used the uniform growth constraint.
To show that $\pr(\mathcal{S}^{\inj})$ is (exactly) right unitary invariant, we need to utilize the uniform growth constraint.

\begin{lemma}[Right unitary invariance]\label{lem:mcho-transfer}
Given a consistent set $\mathcal{S}^{\inj}$ of relations that satisfies the uniform growth constraint. For any $n$-qubit unitary $G$, we have
\begin{equation}
    \ket*{\Adv^{\pr(\mathcal{S}^{\inj}) \cdot G}_t}_{\gsA \gsB \gsR} = (G_{\darkgray{\sR_{\sX, 1}^{(t)}}} \otimes \ldots \otimes G_{\darkgray{\sR_{\sX, t}^{(t)}}}) \ket*{\Adv^{\pr(\mathcal{S}^{\inj})}_t}_{\gsA \gsB \gsR}.
\end{equation}
\end{lemma}

\cref{lem:mcho-transfer} implies right unitary invariance since
\begin{align}
    &\Tr_{\sR}(\ketbra*{\Adv^{\pr(\mathcal{S}^{\inj}) \cdot G}_t}_{\gsA \gsB \gsR}) \nonumber \\
    &= \Tr_{\sR}((G_{\darkgray{\sR_{\sX, 1}^{(t)}}} \otimes \ldots \otimes G_{\darkgray{\sR_{\sX, t}^{(t)}}}) \ketbra*{\Adv^{\pr(\mathcal{S}^{\inj})}_t}_{\gsA \gsB \gsR} (G_{\darkgray{\sR_{\sX, 1}^{(t)}}} \otimes \ldots \otimes G_{\darkgray{\sR_{\sX, t}^{(t)}}})^\dagger) \tag{by~\cref{lem:mcho-transfer}}\\
    &= \Tr_{\sR}(\ketbra*{\Adv^{\pr(\mathcal{S}^{\inj})}_t}_{\gsA \gsB \gsR}). \tag{by the cyclic property of $\Tr_{\sR}$}
\end{align}
The first line corresponds to the adversary's view after making $t$ queries to $\pr(\mathcal{S}^{\inj}) \cdot G_{\gsA}$, while the last line corresponds to its view after making $t$ queries to $\pr(\mathcal{S}^{\inj})$. 

\begin{fact}[Explicit form of the $G$-rotated $\pr(\mathcal{S}^{\inj})$ state] \label{fact:expli-form-mcho-state}
Given a consistent set $\mathcal{S}^{\inj}$ of relations that satisfies the uniform growth constraint. The definition of $\pr(\mathcal{S}^{\inj})$ and $\ket*{R}_{\gsR}$ enable us to expand the $\pr(\mathcal{S}^{\inj})$ state after $t$ queries to obtain
\begin{align}
\ket*{\Adv^{\pr(\mathcal{S}^{\inj}) \cdot G}_t}_{\gsA \gsB \gsR} &= \sqrt{\prod_{i=0}^{t-1} \frac{1}{\calZ_i}} \sum_{\substack{(x_1, \ldots, x_t) \in [N]^t \\ (y_1, \ldots, y_t) \in [N]^t_{\dist} \\ R = \{(x_i, y_i)\}_{i=1}^t \in \mathcal{S}^{\inj}_t}}
\left[ \, \prod_{i = 1}^t \Big( \ketbra*{y_i}{x_i}_{\gsA} \cdot G_{\gsA} \cdot A_{i, \gsA \gsB} \Big) \ket*{0}_{\gsA \gsB} \, \right] \otimes \ket*{R}_{\gsR}\\
&= \sqrt{\prod_{i=0}^{t-1} \frac{1}{\calZ_i}} \sum_{\substack{(x_1, \ldots, x_t) \in [N]^t \\ (y_1, \ldots, y_t) \in [N]^t_{\dist} \\ R = \{(x_i, y_i)\}_{i=1}^t \in \mathcal{S}^{\inj}_t}}
\left[ \, \prod_{i = 1}^t \Big( \ketbra*{y_i}{x_i}_{\gsA} \cdot G_{\gsA} \cdot A_{i, \gsA \gsB} \Big) \ket*{0}_{\gsA \gsB} \, \right] \nonumber \\
&\otimes \frac{1}{\sqrt{t!}} \sum_{\pi \in \sSym_t} \left( R_\pi \ket*{x_1}_{\darkgray{\sR_{\sX, 1}^{(t)}}} \dots \ket*{x_t}_{\darkgray{\sR_{\sX, t}^{(t)}}} \right) \otimes \left( R_\pi \ket*{y_1}_{\darkgray{\sR_{\sY, 1}^{(t)}}} \dots \ket*{y_t}_{\darkgray{\sR_{\sY, t}^{(t)}}} \right).
\end{align}
\end{fact}

\begin{proof}[Proof of~\cref{lem:mcho-transfer}]
By utilizing \cref{fact:expli-form-mcho-state} instead of \cref{fact:expli-form-cho-state}, we can prove this lemma in the same way as the proof of \cref{lem:cho-transfer}.
\end{proof}

\subsection{Relation between $\pr(\mathcal{S}^{\inj})$ and $\pr$ state}

Using the states $\ket*{R}_{\gsR}$, we can define the following restricted subspace projector based on the restricted subspace $\Pi^{\mathsf{restrict}, t}_{\gsR}$ .

\begin{definition}[Restricted subspace projector]
\label{def:restricted-subspace-projector}
For $0 \leq t \leq t_{\max}$, we define the size-$t$ restricted subspace projector $\Pi^{\mathsf{restrict}, t}_{\gsR}$ as follows:
\begin{align}
    \Pi^{\mathsf{restrict}, t}_{\gsR} &\coloneqq \sum_{R \in \mathcal{S}^{\inj}_t} \ketbra*{R}_{\gsR}.
\end{align}
The restricted subspace projector is defined as:
\begin{equation}
    \Pi^{\mathsf{restrict}}_{\gsR} \coloneqq \sum_{t=0}^{t_{\max}} \Pi^{\mathsf{restrict}, t}_{\gsR}.
\end{equation}
\end{definition}

From \cref{fact:expli-form-mcho-state} and the projector defined above, we immediately obtain the following equation relating the $G$-rotated $\pr(\mathcal{S}^{\inj})$ and the $\ket*{\Adv^{\pr \cdot G}_t}_{\gsA \gsB \gsR}$s.

\begin{fact}[Relation between $\pr(\mathcal{S}^{\inj})$ and $\pr$ state] \label{fact:relation-mcho-cho}
Given a consistent set $\mathcal{S}^{\inj}$ of relations that satisfies the uniform growth constraint.
For any $n$-qubit unitary $G$, we have
\begin{equation}
     \sqrt{\prod_{i=0}^{t-1} \calZ_i} \cdot \sqrt{\frac{(N-t)!}{N!}} \cdot \ket*{\Adv^{\pr(\mathcal{S}^{\inj}) \cdot G}_t}_{\gsA \gsB \gsR} = \Pi^{\mathsf{restrict}}_{\gsR} \ket*{\Adv^{\pr \cdot G}_t}_{\gsA \gsB \gsR},
\end{equation}
where the prefactor $\sqrt{\prod_{i=0}^{t-1} \calZ_i} \cdot \sqrt{\frac{(N-t)!}{N!}}$ is between $0$ and $1$ because the $\pr(\mathcal{S}^{\inj})$ state $\ket*{\Adv^{\pr(\mathcal{S}^{\inj}) \cdot G}_t}_{\gsA \gsB \gsR}$ has unit norm and the projected $\pr$ state $\Pi^{\mathsf{restrict}}_{\gsR} \ket*{\Adv^{\pr \cdot G}_t}_{\gsA \gsB \gsR}$ has norm at most one.
\end{fact}

\subsection{$\pr(\mathcal{S}^{\inj})$ is indistinguishable from Haar random unitaries}

When the restricted set $\mathcal{S}^{\inj}$ of relations has a large enough growth, then $\pr(\mathcal{S}^{\inj})$ is indistinguishable from Haar-random unitaries. 
This is formally captured by the following theorem.

\begin{theorem}[$\pr(\mathcal{S}^{\inj})$ is indistinguishable from Haar random]\label{theorem:haar-mcho}
    Given a consistent set $\mathcal{S}^{\inj}$ of relations that satisfies the uniform growth constraint.
    Let $\Adv$ be a $t$-query oracle adversary with forward queries. Then
     \begin{align}
        &\TD\left(\E_{\calO \gets \mu_{\mathsf{Haar}}} \ketbra*{\Adv_t^{\calO}}, \Tr_{\sR}\left( \ketbra*{\Adv^{\pr(\mathcal{S}^{\inj})}_t}_{\gsA \gsB \gsR} \right) \right)\\ &\leq \frac{2t(t-1)}{N+1} + 2 \left(1 - \prod_{i=0}^{t-1} \calZ_i \cdot \frac{(N-t)!}{N!} \right).
    \end{align}
\end{theorem}
\begin{proof}
Using \cref{theorem:haar-cho} and triangle inequality, we have
\begin{align}
    &\TD\left(\E_{\calO \gets \mu_{\mathsf{Haar}}} \ketbra*{\Adv_t^{\calO}}, \Tr_{\sR}\left( \ketbra*{\Adv^{\pr(\mathcal{S}^{\inj})}_t}_{\gsA \gsB \gsR} \right) \right) \\
    &\leq \frac{2t(t-1)}{N+1} + \TD\left(\Tr_{\sR}\left( \ketbra*{\Adv^{\pr}_t}_{\gsA \gsB \gsR} \right), \Tr_{\sR}\left( \ketbra*{\Adv^{\pr(\mathcal{S}^{\inj})}_t}_{\gsA \gsB \gsR} \right) \right).
\end{align}
We can bound the second term as follows,
\begin{align}
    &\TD\left(\Tr_{\sR}\left( \ketbra*{\Adv^{\pr}_t}_{\gsA \gsB \gsR} \right), \Tr_{\sR}\left( \ketbra*{\Adv^{\pr(\mathcal{S}^{\inj})}_t}_{\gsA \gsB \gsR} \right) \right)\\
    &\leq \TD\left(\Tr_{\sR}\left( \ketbra*{\Adv^{\pr}_t}_{\gsA \gsB \gsR} \right), \Tr_{\sR}\left( \Pi^{\mathsf{restrict}}_{\gsR} \ketbra*{\Adv^{\pr \cdot G}_t}_{\gsA \gsB \gsR} \Pi^{\mathsf{restrict}}_{\gsR} \right) \right) \label{eq:project-cho-to-mcho}\\
    &+ \TD\left(\Tr_{\sR}\left( \Pi^{\mathsf{restrict}}_{\gsR} \ketbra*{\Adv^{\pr \cdot G}_t}_{\gsA \gsB \gsR} \Pi^{\mathsf{restrict}}_{\gsR} \right), \Tr_{\sR}\left( \ketbra*{\Adv^{\pr(\mathcal{S}^{\inj})}_t}_{\gsA \gsB \gsR} \right) \right). \label{eq:project-cho-to-mcho-2nd}
\end{align}
Using \cref{lemma:gentle} and \cref{fact:relation-mcho-cho}, \cref{eq:project-cho-to-mcho} is equal to
\begin{align}
    1 - \bra*{\Adv^{\pr \cdot G}_t} \Pi^{\mathsf{restrict}}_{\gsR} \ket*{\Adv^{\pr \cdot G}_t}_{\gsA \gsB \gsR} = \left(1 - \prod_{i=0}^{t-1} \calZ_i \cdot \frac{(N-t)!}{N!} \right).
\end{align}
Again, using \cref{fact:relation-mcho-cho}, \cref{eq:project-cho-to-mcho-2nd} is equal to
\begin{align}
    &\TD\left( \prod_{i=0}^{t-1} \calZ_i \cdot \frac{(N-t)!}{N!} \cdot \Tr_{\sR}\left( \ketbra*{\Adv^{\pr(\mathcal{S}^{\inj})}_t}_{\gsA \gsB \gsR} \right), \Tr_{\sR}\left( \ketbra*{\Adv^{\pr(\mathcal{S}^{\inj})}_t}_{\gsA \gsB \gsR} \right) \right) \\
    &= \left(1 - \prod_{i=0}^{t-1} \calZ_i \cdot \frac{(N-t)!}{N!} \right).
\end{align}
Together, we obtained the stated result.
\end{proof}

\section{An elementary proof of the gluing lemma}
\label{sec:gluing}

In this section, we show how to use the path-recording framework to establish an elementary proof of the gluing lemma recently shown in \cite{schuster2024random}. The proof in \cite{schuster2024random} makes use of representation theory and Weingarten calculus. Here, we present an elementary proof using the path-recording framework for analyzing Haar-random unitaries.

The gluing lemma presented in \cite{schuster2024random} shows that the composition of two Haar-random unitaries on system $\sA_1 \sA_2$ and $\sA_2 \sA_3$ that overlap only on a small number of qubits is indistinguishable from a Haar-random unitary on the entire system $\sA_1 \sA_2 \sA_3$.

\begin{theorem}[Gluing two Haar-random unitaries]
Consider three systems $\sA_1, \sA_2, \sA_3$ of qubits with $\sA = \sA_1 \sA_2 \sA_3$. Let $|\sA_1|, |\sA_2|, |\sA_3|$ denote the number of qubits in each system. Let $U_{\sA_1 \sA_2}, U_{\sA_2 \sA_3}, U_{\sA}$ be Haar-random unitaries on system $\sA_1 \sA_2, \sA_2 \sA_3, \sA$, respectively. We have
\begin{align}
    &\TD\left(\E_{U_{\sA_1 \sA_2}, U_{\sA_2 \sA_3}} \ketbra*{\Adv_t^{U_{\sA_1 \sA_2} U_{\sA_2 \sA_3}}}, \E_{U_{\sA}} \ketbra*{\Adv_t^{U_{\sA}}} \right) \leq \frac{9 t(t-1)}{2^{|\sA_2|}}.
\end{align}
\end{theorem}
\begin{proof}
We approximate the three Haar-random unitaries $U_{\sA_1 \sA_2}, U_{\sA_2 \sA_3}, U_{\sA_1 \sA_2 \sA_3}$ by three restricted sets $\mathcal{S}^{\inj}_{\sA_1 \sA_2}, \mathcal{S}^{\inj(\mathsf{D})}_{\sA_2 \sA_3}, \mathcal{S}^{\inj\mathsf{(CD)}}_{\sA_1 \sA_2 \sA_3}$ of relations. The three subsets of injective relations are given as follows.
\begin{itemize}
    \item $\mathcal{S}^{\inj(\sA_2)}_{\sA_1\sA_2}$: Injective relations $R$ over system $\sA_1 \sA_2$ such that the system $\sA_2$ part of elements in the image $\Im(R)$ are distinct, i.e., given $\Im(R) = \{y_1, \ldots, y_{|R|}\}$, $y_{i, \sA_2} \neq y_{j, \sA_2}$ for all $i \neq j$.
    \item $\mathcal{S}^{\inj(\sA_2)}_{\sA_2 \sA_3}$: Injective relations $R$ over system $\sA_2 \sA_3$ such that the system $\sA_2$ part of elements in the image $\Im(R)$ are distinct.
    \item $\mathcal{S}^{\inj(\sA_2)}_{\sA_1 \sA_2 \sA_3}$: Injective relations $R$ over system $\sA_1 \sA_2 \sA_3$ such that the system $\sA_1 \sA_2$ part of elements in the image $\Im(R)$ are distinct.
\end{itemize}
Here, given a bitstring $y$, we denote $y_{\sA_2}$ to be the substring corresponding to bits in $\sA_2$.
It is not hard to see that these restricted sets is consistent and satisfies the uniform growth constraint.
We consider the path recording oracle $\pr(\mathcal{S}^{\inj(\sA_2)}_{\sA_1 \sA_2})$ to act on system $\sA_1, \sA_2, \sR_1$, the oracle $\pr(\mathcal{S}^{\inj(\sA_2)}_{\sA_2 \sA_3})$ to act on system $\sA_2, \sA_3, \sR_2$, and the oracle $\pr(\mathcal{S}^{\inj(\sA_2)}_{\sA_1 \sA_2 \sA_3})$ to act on system $\sA_1, \sA_2, \sA_3, \sR_3$.
Let
\begin{align}
    \rho_{1} &\coloneq \E_{U_{\sA_1 \sA_2}, U_{\sA_2 \sA_3}} \ketbra*{\Adv_t^{U_{\sA_1 \sA_2} U_{\sA_2 \sA_3}}},\\
    \rho_{2} &\coloneq \E_{U_{\sA_1 \sA_2}} \Tr_{\sR_1} \ketbra*{\Adv_t^{U_{\sA_1 \sA_2} \pr(\mathcal{S}^{\inj(\mathsf{\sA_2})}_{\sA_2 \sA_3})}},\\
    \rho_{3} &\coloneq \Tr_{\sR_1} \Tr_{\sR_2} \ketbra*{\Adv_t^{\pr(\mathcal{S}^{\inj}_{\sA_1\sA_2}) \pr(\mathcal{S}^{\inj(\sA_2)}_{\sA_2 \sA_3})}},\\
    \rho_{4} &\coloneq \Tr_{\sR_3} \ketbra*{\Adv_t^{\pr(\mathcal{S}^{\inj(\sA_2)}_{\sA_1 \sA_2 \sA_3})}},\\
    \rho_{5} &\coloneq \E_{U_{\sA_1 \sA_2 \sA_3}} \ketbra*{\Adv_t^{U_{\sA_1 \sA_2 \sA_3}}} .
\end{align}
Using \cref{theorem:haar-mcho} and properly computing the normalization factor $\mathcal{Z}_t$, we have
\begin{align}
    \TD(\rho_{1}, \rho_{2}) &\leq \frac{2t (t-1)}{2^{|\sA_2| + |\sA_3|}} + 2 \left(1 - \prod_{i=0}^{t-1} (2^{|\sA_2| + |\sA_3|} - i 2^{|\sA_3|}) \cdot \frac{(2^{|\sA_2| + |\sA_3|} -t)!}{(2^{|\sA_2| + |\sA_3|})!} \right),\\
    &\leq \frac{2t (t-1)}{2^{|\sA_2| + |\sA_3|}} + 2 \left(1 - \prod_{i=0}^{t-1} \left( 1 - i 2^{-|\sA_2|} \right) \right) \leq \frac{2t (t-1)}{2^{|\sA_2| + |\sA_3|}} + \frac{t (t-1)}{2^{|\sA_2|}} \leq \frac{3 t (t-1)}{2^{|\sA_2|}}.
\end{align}
Similarly, we have
\begin{align}
    \TD(\rho_{2}, \rho_{3}) &\leq \frac{2t (t-1)}{2^{|\sA_1| + |\sA_2|}} + 2 \left(1 - \prod_{i=0}^{t-1} (2^{|\sA_1| + |\sA_2|} - i 2^{|A_1|} ) \cdot \frac{(2^{|\sA_1| + |\sA_2|} -t)!}{(2^{|\sA_1| + |\sA_2|})!} \right) \leq \frac{3 t (t-1)}{2^{|\sA_2|}},\\
    \TD(\rho_{4}, \rho_{5}) &\leq \frac{2t (t-1)}{2^{|\sA|}} + 2 \left(1 - \prod_{i=0}^{t-1} (2^{|\sA|} - i 2^{|\sA_1| + |\sA_3|}) \cdot \frac{(2^{|\sA|} -t)!}{(2^{|\sA|})!} \right) \leq \frac{3 t (t-1)}{2^{|\sA_2|}}.
\end{align}
Finally, we show that $\rho_3 = \rho_4$. Let $x || y$ denote bitstring concatenation. Using the definition of the restricted subsets of injective relations, the explicit form of the purified state in \cref{fact:expli-form-mcho-state} yields
\begin{align}
&\ket*{\Adv^{\pr(\mathcal{S}^{\inj(\sA_2)}_{\sA_1\sA_2}) \pr(\mathcal{S}^{\inj(\sA_2)}_{\sA_2 \sA_3})}_t}_{\gsA \gsB \gsR_{\darkgray 1} \gsR_{\darkgray 2}} = \sqrt{\prod_{i=0}^{t-1} \frac{1}{2^{|\sA_2| + |\sA_3|} - i 2^{|\sA_3|}}} \cdot \sqrt{\prod_{i=0}^{t-1} \frac{1}{2^{|\sA_1| + |\sA_2|} - i 2^{|\sA_1|}}} \cdot\\
&\quad \sum_{\substack{(x_1, \ldots, x_t) \in [2^{|A|}]^t \\ (y_1, \ldots, y_t) \in [2^{|A|}]^t_{\dist} \\ (z_1, \ldots, z_t) \in [2^{|A_2|}]^t_{\dist} \\ \mathrm{s.t.} \,\, (y_{1, A_2}, \ldots, y_{t, A_2}) \in [2^{|A_2|}]^t_{\dist} \\
R = \{ ( x_{i, A_2} || x_{i, A_3}, z_{i, A_2} || y_{i, A_3} ) \}_{i=1}^t \in \mathcal{S}^{\inj(\sA_2)}_{\sA_2 \sA_3} \\
S = \{ ( x_{i, A_1} || z_{i, A_2}, y_{i, A_1} || y_{i, A_2} ) \}_{i=1}^t \in \mathcal{S}^{\inj(\sA_2)}_{\sA_1\sA_2} }}
\left[ \, \prod_{i = 1}^t \Big( \ketbra*{y_i}{x_i}_{\gsA} \cdot A_{i, \gsA \gsB} \Big) \ket*{0}_{\gsA \gsB} \, \right] \otimes \ket*{R}_{\gsR_{\darkgray 1}} \otimes \ket*{S}_{\gsR_{\darkgray 2}},
\end{align}
and, similarly, \cref{fact:expli-form-mcho-state} also yields
\begin{align}
&\ket*{\Adv^{\pr(\mathcal{S}^{\inj(\sA_2)}_{\sA_1\sA_2\sA_3})}_t}_{\gsA \gsB \gsR_{\darkgray 3}} = \sqrt{\prod_{i=0}^{t-1} \frac{1}{2^{|\sA|} - i 2^{|\sA_1| + |\sA_3|}}} \cdot\\
&\quad\quad\quad \sum_{\substack{(x_1, \ldots, x_t) \in [2^{|A|}]^t \\ (y_1, \ldots, y_t) \in [2^{|A|}]^t_{\dist} \\ \mathrm{s.t.} \,\, (y_{1, A_2}, \ldots, y_{t, A_2}) \in [2^{|A_2|}]^t_{\dist} \\
T = \{ ( x_{i}, y_{i} ) \}_{i=1}^t \in \mathcal{S}^{\inj(\sA_2)}_{\sA_1\sA_2\sA_3}  }}
\left[ \, \prod_{i = 1}^t \Big( \ketbra*{y_i}{x_i}_{\gsA} \cdot A_{i, \gsA \gsB} \Big) \ket*{0}_{\gsA \gsB} \, \right] \otimes \ket*{T}_{\gsR_{\darkgray 3}}.
\end{align}
We define a linear map $\Uncompress$ that maps registers $\sR_3$ to registers $\sR_1, \sR_2$. For any $T = \{ ( x_{i}, y_{i} ) \}_{i=1}^t$ such that $(x_1, \ldots, x_t) \in [2^{|A|}]^t$, $(y_1, \ldots, y_t) \in [2^{|A|}]^t_{\dist}$, and $(y_{1, A_2}, \ldots, y_{t, A_2}) \in [2^{|A_2|}]^t_{\dist}$, the linear map $\Uncompress$ acts as
\begin{equation}
\Uncompress \ket*{T}_{\gsR_{\darkgray 3}} \coloneq \sqrt{\prod_{i=0}^{t-1} \frac{1}{2^{|\sA_2|} - i}} \sum_{\substack{ (z_1, \ldots, z_t) \in [2^{|A_2|}]^t_{\dist} \\ \mathrm{s.t.} \,\,
R = \{ ( x_{i, A_2} || x_{i, A_3}, z_{i, A_2} || y_{i, A_3} ) \}_{i=1}^t \\
S = \{ ( x_{i, A_1} || z_{i, A_2}, y_{i, A_1} || y_{i, A_2} ) \}_{i=1}^t }} \ket*{R}_{\gsR_{\darkgray 1}} \ket*{S}_{\gsR_{\darkgray 2}}.
\end{equation}
One can directly check that $\Uncompress$ is a partial isometry and $\ket*{\Adv^{\pr(\mathcal{S}^{\inj(\sA_2)}_{\sA_1\sA_2\sA_3})}_t}_{\gsA \gsB \gsR_{\darkgray 3}}$ is in the domain of $\Uncompress$ by construction. Furthermore, by definition, we have
\begin{equation}
\Uncompress \ket*{\Adv^{\pr(\mathcal{S}^{\inj(\sA_2)}_{\sA_1\sA_2\sA_3})}_t}_{\gsA \gsB \gsR_{\darkgray 3}} = \ket*{\Adv^{\pr(\mathcal{S}^{\inj(\sA_2)}_{\sA_1\sA_2}) \pr(\mathcal{S}^{\inj(\sA_2)}_{\sA_2 \sA_3})}_t}_{\gsA \gsB \gsR_{\darkgray 1} \gsR_{\darkgray 2}}.
\end{equation}
Because $\Uncompress$ acts isometric in its domain, we have
\begin{align}
    \rho_{3} &= \Tr_{\sR_1} \Tr_{\sR_2} \ketbra*{\Adv_t^{\pr(\mathcal{S}^{\inj}_{\sA_1\sA_2}) \pr(\mathcal{S}^{\inj(\sA_2)}_{\sA_2 \sA_3})}}\\
    &= \Tr_{\sR_3} \ketbra*{\Adv_t^{\pr(\mathcal{S}^{\inj(\sA_2)}_{\sA_1 \sA_2 \sA_3})}} = \rho_{4}.
\end{align}
Together, by a series of triangle inequality, we have
\begin{equation}
    \TD(\rho_1, \rho_5) \leq \TD(\rho_1, \rho_2) + \TD(\rho_2, \rho_3) + \TD(\rho_4, \rho_5) \leq \frac{9 t(t-1)}{2^{|\sA_2|}},
\end{equation}
which concludes the proof of this theorem by noting the definition of $\rho_1$ and $\rho_5$.
\end{proof}

As shown in \cite{schuster2024random}, one can iteratively apply the gluing lemma to glue many small Haar-random unitaries over small number of qubits into a pseudorandom unitary on the entire system. If we have an $n$-qubit system $\sA$ that is separated into consecutive subsystems $\sA_1, \ldots, \sA_{2m}$, we can glue together small Haar-random unitaries $U_{\sA_{1} \sA_{2}}, U_{\sA_{2} \sA_{3}}, \ldots$ as follows,
\begin{equation}
    U_{\mathsf{glued}} \coloneq \prod_{k=0}^{m-1}(U_{\sA_{2+2k} \sA_{3+2k}}) \prod_{k=0}^{m-1}(U_{\sA_{1+2k} \sA_{2+2k}}).
\end{equation}
Using triangle inequality, the trace distance between a $t$-query adversary output state that queries $U_{\mathsf{glued}}$ versus Haar-random unitary is upper bounded by
\begin{equation}
    2m \cdot \frac{ 9 t (t - 1)}{2^{\min_{i \in [2m]} |\sA_i|}}.
\end{equation}
If each subsystem is of size $\omega(\log n)$, then the glued unitary $U_{\mathsf{glued}}$ will be a pseudorandom unitary secure against $\poly(n)$-time adversary. This can be seen by noting that a $\poly(n)$-time adversary can only make $t = \poly(n)$ queries, so the trace distance between the adversary state querying $U_{\mathsf{glued}}$ and Haar-random unitary is upper bounded by
\begin{equation}
    2m \cdot \frac{ 9 t (t - 1)}{2^{\min_{i \in [2m]} |\sA_i|}} \leq \frac{\poly(n)}{2^{\omega(\log n)}} = \mathrm{negl}(n).
\end{equation}
By replacing the small Haar-random unitaries over $\omega(\log n)$ qubits with small pseudorandom unitaries secure against subexponential adversary, one can show that the glued unitary $U$ is an $n$-qubit pseudorandom unitary secure against $\poly(n)$-time adversary.

Assuming the subexponential hardness of LWE \cite{regev2009lattices}, we have proved that a pseudorandom unitary secure against subexponential adversary can be generated in polynomial-depth on any circuit geometry using the $PFC$ construction, including a 1D geometry.
Hence, an $n$-qubit pseudorandom unitary secure against polynomial adversary can be generated in $\poly \log(n)$ depth on any circuit geometry.
This work fills in the gap in \cite{schuster2024random} that assumes the $PFC$ construction forms a pseudorandom unitary under LWE hardness to rigorously establish the surprising fact that pseudorandom unitaries can be generated in extremely low depth under standard cryptographic assumptions.

\end{document}